\documentclass[12pt]{article}

\onecolumn 

\usepackage{graphicx}

\usepackage{array} 

\usepackage{natbib, mathtools}
\usepackage{setspace}
\usepackage{enumerate}
\usepackage{booktabs}
\usepackage[utf8]{inputenc}

\usepackage{color}
\usepackage{comment}

\usepackage{appendix}

\usepackage{lipsum}

\usepackage{amsmath,amssymb,lmodern,amsthm}

\usepackage{float}

\usepackage{hyperref} 
\usepackage{cleveref} 

\hypersetup{colorlinks = true, 
        citecolor = blue}

\usepackage{algorithm}
\usepackage{algorithmicx}
\usepackage{algpseudocode}

\def\spacingset#1{\renewcommand{\baselinestretch}%
{#1}\small\normalsize} \spacingset{1}

\renewcommand{\algorithmicrequire}{ \textbf{Input:}} 
\renewcommand{\algorithmicensure}{ \textbf{Output:}} 

\newcommand{\EE}{\mathbb E}
\newcommand{\PP}{\mathbb P}
\newcommand{\II}{\mathbb I}

\usepackage{geometry}
\graphicspath{{Fig/}}

\newtheorem{theorem}{Theorem}
\newtheorem{proposition}[theorem]{Proposition}%
\newtheorem{example}{Example}%
\newtheorem{remark}{Remark}%

\newtheorem{thm}{Theorem}

\newtheorem{property}{Property}
\newtheorem{lemma}{Lemma}

\newtheorem{corollary}{Corollary}

\title{A Conformalized Empirical Bayes Method for Multiple Testing with Side Information}

\author{Zinan Zhao$^{1}$ \ and \ Wenguang Sun$^{2,3}$}
\date{}

\begin{document}

\maketitle

\abstract{This article presents a Conformalized Locally Adaptive Weighting (CLAW) approach to multiple testing with side information. The proposed method employs innovative data-driven strategies to construct pairwise exchangeable scores, which are integrated into a generic algorithm that leverages a mirror process for controlling the false discovery rate (FDR). By combining principles from empirical Bayes with powerful techniques in conformal inference, CLAW provides a valid and efficient framework for incorporating structural information from both test data and auxiliary covariates. Unlike existing empirical Bayes FDR methods that primarily offer asymptotic validity, often under strong regularity conditions, CLAW controls the FDR in finite samples under weaker conditions. Extensive numerical studies using both simulated and real data demonstrate that CLAW exhibits superior performance compared to existing methods.}

\noindent%

\footnotetext[1]{Center for Data Science and School of Mathematical Sciences, Zhejiang University, China.}  

\footnotetext[2]{Center for Data Science and School of Management, Zhejiang University, China. } 

\footnotetext[3]{Author for correspondence: \url{wgsun@zju.edu.cn}. Address: 866 Yuhangtang Road, Hangzhou, Zhejiang Province, China}

\spacingset{1}

\section{Introduction}
\label{sec:intro}

\subsection{Multiple testing with side information}
\label{subsec:side_info}

In concurrent data-intensive fields, such as genomics, neuroimaging, and signal processing, the collection of vast volumes of data is a routine practice. These data are often accompanied by side information, adding valuable context to both analysis and interpretation processes. In large-scale testing problems, side information can be extracted from various sources. For example, researchers can derive side information from intrinsic data patterns, such as temporal and spatial ordering \citep{benjamini07spatial,sunwei11}, as well as grouping or hierarchical structures \citep{efron08,yekutieli08hierarchical, goeman08, sunwei15, goeman20}. Additionally, external sources, such as prior studies and domain-specific knowledge, can be utilized to extract valuable insights \citep{roeder09gw,du14,DFKO15,licaili23transfer}. Finally, within the same study, auxiliary sequences can be  constructed to uncover pertinent structural information \citep{huber10,huber16, CARS, HART}. 

Various approaches have been proposed to incorporate side information into false discovery rate (FDR; \citealp{bh95}) analysis, aiming to produce more meaningful scientific findings and facilitate informed decision-making. This extensively studied field has explored several important directions, including grouping-based methods  \citep{cs09, pfilter17,pfilter19}, weighting-based methods via either procedural weights  \citep{GRW06,roquain09,durand19} or decision weights \citep{bh97,basu2018,gang2023ranking}, as well as covariate-adaptive methods modifying existing p-value based algorithms \citep{du14,lei18adapt,SABHA19,Ignatiadis21IHW,LAWS22}, z-value based algorithms \citep{scott15,CARS,HART,ZAP}, and variable selection algorithms \citep{ren23knockoff}.  

Suppose we are interested in testing $m$ hypotheses $\{H_i: i\in[m]\equiv\{1, \ldots, m\}\}$, where each $H_i$ is associated with a primary data point $T_i$ and a corresponding covariate $S_i$; both $T_i$ and $S_i$ can be either univariate or multivariate. Let $\mathbf{T}=(T_{i})_{i=1}^m$ and $\mathbf{S}=(S_{i})_{i=1}^m$. Additionally, we assume that a set of null samples ${\mathbf T}^0=\{T^{0}_{j}: j\in\mathcal D_0\}$ has been obtained. Under the conventional multiple testing setup where the null distribution $F_0$ is known precisely, ${\mathbf T}^0$ can be directly sampled from $F_0$. Under the semi-supervised multiple testing setup \citep{Blaetal10, mary22semi}, the null samples can be collected from previous experiments or generated via specialized null sampling machines.

\subsection{A covariate-adaptive working mixture model}\label{subsec:cov-mix-model}

The presence of covariates \(\mathbf{S}\) complicates the task of finding a suitable model that accurately captures the intricacies of the data generation process. To address this challenge, we consider a covariate-adaptive model motivated by an empirical Bayes perspective, which allows us to integrate side information in a principled manner.

Let $\theta_{i}\in\{0, 1\}$ denote a binary variable, with $\theta_i=0/1$ indicating that $H_i$ is true/false. The model captures the probabilistic relationships between the test data points and their corresponding covariates $(T_j, S_j)_{j=1}^m$ through a hierarchical approach: 
\begin{equation}\label{model:eb-mix}
(\theta_{j}|S_{j}=s)  \sim  \mbox{Bernoulli}(\pi_s), \quad (T_j|S_j, \theta_j)  \sim  (1-\theta_j) F_0(\cdot) + \theta_j F_1(\cdot|S_j). 
\end{equation}
The specification of this model incorporates several important considerations. 

Firstly, the covariate-adaptive model \eqref{model:eb-mix} should be regarded as a working model and the utilization of the empirical Bayes framework serves purely as a means to inspire and motivate our methodology. As shown in subsequent sections, our inference remains valid even when the working model \eqref{model:eb-mix} deviates from the true data-generating model. While the underlying state $\theta_i$ is conceptualized as a binary variable, our theory specifically focuses on the frequentist FDR, treating $(\theta_i)_{i=1}^m$ as a non-random sequence. 

Secondly, the covariate \( S_{i} \in \mathcal{X} \) can take on either discrete or continuous values, and it can be either deterministic or stochastic. The joint distribution of \( \mathbf{S} \) is left unspecified. This provides flexibility to accommodate diverse types of covariates. 

Thirdly, the dependence of $\theta_j$ on $S_j$ is captured through the local sparsity level $\pi_s=\PP(\theta_j=1|S_j=s)$. In the scenario where $S_j$ represents, say, group memberships (or spatial locations), $\pi_s$ indicates varying sparsity levels across different groups (or local neighborhoods in a spatial region), thereby providing critical structural information that can be leveraged to construct more efficient FDR procedures \citep{SABHA19,  LAWS22}. 

Fourthly, the test data points $T_j$ are modeled using a mixture distribution that depends on both $\theta_j$ and $S_j$. The mixture distribution comprises two components: the null distribution $F_0(\cdot)$ and the non-null distribution $F_1(\cdot|S_j)$. A key assumption in  model \eqref{model:eb-mix} is that $F_0$ is invariant with respect to the covariate $S_j$, i.e.,  
\begin{equation}\label{cond:invariance}
(T_j|S_j, \theta_j=0)\sim F_0(t|S_j) \equiv F_0, \;\; j\in[m];\;\; \mbox{and}\;\; (T_j^0|S_j) \sim F_0(\cdot), \;\; j\in\mathcal{D}_0. 
\end{equation}
Similar assumptions have been employed in the literature on structured multiple testing \citep{lei18adapt,SABHA19,Ignatiadis21IHW,LAWS22}, where it is commonly assumed that the null p-values remain independent and super-uniform, given the auxiliary covariates and the remaining non-null p-values. Assumption \eqref{cond:invariance} will be revisited when discussing relevant exchangeability conditions in Section \ref{subsec:pw_exch}.

Finally, in contrast to the assumption of a fixed null distribution \(F_0\) across all \(i \in [m]\), model \eqref{model:eb-mix} allows the non-null distribution \(F_1(\cdot|S_j)\) to vary across different values of \(S_j\). This flexibility is crucial for accommodating the heterogeneity among the non-null units, which is commonly encountered in practice.

If we further assume that (a) $T_j$ is a continuous random variable and (b) $\theta_j$'s are independent with each other, then model \eqref{model:eb-mix} can be equivalently expressed as follows: 
\begin{equation}
(T_{j}|S_{j}=s) \stackrel{ind.}{\sim} f_{s}(t) = (1-\pi_{s})f_{0}(t) + \pi_{s}f_{1s}(t),\quad  i\in[m], \label{model:mixture} 
\end{equation}
where $f_{0}(t)$ and $f_{1s}(t)$ are the density functions of $F_0(t)$ and $F_1(t|s)$, respectively. 
This covariate-adaptive mixture density function \eqref{model:mixture}, which has been widely employed in empirical Bayes FDR procedures \citep{Egil08, scott15,Tanetal18, CARS}, extends the classical two-group mixture model \citep{efron01, sc07}:
$
T_{j} \stackrel{i.i.d.}{\sim} f(t) = (1-\pi)f_{0}(t) + \pi f_{1}(t),
$
to the more complex scenario with side information.

\subsection{Empirical Bayes methods: basics, challenges and our proposal}

Multiple testing involves solving a compound decision problem, where harnessing the overall structure of many parallel problems enhances the efficiency of simultaneous inference \citep{Rob51,sc07}. 
To investigate the optimal utilization of side information, we start by examining the ideal scenario where an oracle possesses pertinent knowledge of the working model \eqref{model:mixture}. Within this setting, the optimal FDR procedure takes the form of a thresholding rule based on a covariate-informed statistic known as the conditional local FDR \citep{cs09, CARS}: 
\begin{equation}\label{clfdr-stat}
\mbox{Clfdr}(T_i, S_i)=\PP(\theta_i=0|T_i,S_i)=\frac{(1-\pi_{S_i})f_{0}(T_i)}{f_{S_i}(T_i)}. 
\end{equation}
In practical scenarios where estimating the Clfdr is necessary, various approaches have been proposed. These include Bayesian computational methods utilizing parametric priors \citep{scott15, Tanetal18}, as well as nonparametric empirical Bayes (NEB) methods employing $f$-modeling \citep{CARS, HART} or $g$-modeling techniques \citep{GuKoe23, gang2023ranking}. While the parametric Bayesian methods may encounter issues if the priors are mis-specified, the NEB methods offer greater flexibility and robustness, exhibiting desirable frequentist properties. However, the theoretical analysis of these methods is inherently complex. The validity theory often relies on asymptotic arguments and assumes conditions that may not hold or be difficult to validate in real-word scenarios.

In this article, we address the challenges  by leveraging recent advancements in key areas such as knockoff filters \citep{barber15knockoff, ren23knockoff}, conformal inference \citep{vovk05, lei14prediction,marandon22mlfdr}, and e-values \citep{wang22ev, ren2023derandomized}. We propose the Conformalized Locally Adaptive Weighting (CLAW) approach, which offers a compelling demonstration of how empirical Bayes ideas can be effectively implemented within a principled frequentist framework. Unlike Bayesian methods, CLAW eliminates the need for correctly specified priors or strong regularity conditions, and provides valid and efficient inference in finite samples. 

The development of CLAW consists of two crucial steps. In the first step (Section \ref{sec:preliminary}), we establish fundamental principles and lay the theoretical foundations for conformalized multiple testing with side information. In the second step (Section \ref{sec:claw}), we develop innovative strategies to construct conformity scores that integrate  side information into inference effectively. The new method achieves improved statistical power and rigorous theoretical guarantees simultaneously under mild conditions of exchangeability. In Section \ref{sec:extension}, we demonstrate that CLAW can be further extended to handle semi-supervised setups and integrate side information from multiple sources. Our numerical results show that CLAW substantially improves the performance of existing methods across various settings.

\subsection{Connections and distinctions with related work}
\label{subsec:conformal_inference}

CLAW is closely related to three significant lines of research (see Section \ref{app:discuss} of the Supplement for a detailed discussion on the connections and distinctions between CLAW and related methods). The first direction focuses on incorporating side information through weighting. For example, IHW \citep{Ignatiadis21IHW} divides hypotheses into different groups based on covariate values and generates cross-fitting weights for each group. SABHA \citep{SABHA19} and LAWS \citep{LAWS22} develop sparsity-adaptive weights to adjust the corresponding p-values. However, our numerical studies reveal that these weighting strategies are suboptimal due to information loss in the grouping step or the omission of other important structural information in the test data. In contrast, CLAW develops covariate-assisted weights to emulate the optimal decision rule under the empirical-Bayes setup, demonstrating superior performance across various settings.

The second approach involves learning covariate-modulated decision boundaries by gradually unmasking the data (AdaPT, \citealp{lei18adapt}; adaptive knockoffs, \citealp{ren23knockoff}). All three methods (CLAW, AdaPT, and adaptive knockoffs) operate as generalizations of the Selective SeqStep+ algorithm \citep{barber15knockoff}. Both AdaPT and adaptive knockoffs assume the prior availability of mirror-conservative p-values or anti-symmetric statistics, with covariates utilized separately at a later stage to determine the adaptive masking rules. In contrast, CLAW directly constructs powerful conformity scores by aggregating the side information through the working model \eqref{model:mixture}, offering a direct, intuitive, and principled method for covariate-assisted inference.

The third approach, which falls within the framework of conformal inference, exemplified by the BONuS \citep{yang21bonus} and AdaDetect \citep{marandon22mlfdr}, aims to utilize test data to construct more powerful conformity scores. Our proposed method combines NEB modeling and conformal inference techniques, aligning with the ideas in these recent advancements. However, CLAW departs from the strict requirement of joint exchangeability imposed by BONuS and AdaDetect by constructing covariate-adaptive scores that fulfill a weaker pairwise exchangeability condition. This new framework improves the flexibility and efficiency in both the modeling and inference stages. 

Finally, our work is related to the PLIS procedure in \citet{zhao2023plis}, which aims to leverage the dependencies in structured probabilistic models. However, PLIS requires explicitly specified models, such as hidden Markov models, to capture these dependencies, and it cannot  handle the generic setup where side information is encoded as a covariate sequence. CLAW constructs novel bivariate score functions to incorporate side information, which represents a substantial departure from the strategy employed in PLIS. 

\subsection{Outline}

The article is organized as follows. Section \ref{sec:preliminary} outlines the basic framework, followed by Section \ref{sec:claw}, which details the CLAW method and its theoretical properties. Section \ref{sec:extension} presents extensions and connections to existing works. We investigate the numerical performance of CLAW using both simulated data (Section \ref{sec:simu}) and real data (Section \ref{sec:application}). Section \ref{sec:discuss} concludes with a discussion on future directions. Further elaborations, technical proofs, and additional numerical results are provided in the Supplement. The code for replicating all our experiments is available for download at \url{https://github.com/zzndotzhangzhinan/clawpaper.git}.

\section{Preliminaries and Basic Framework}
\label{sec:preliminary}

Section \ref{subsec:propotype} introduces the problem formulation and presents a prototype algorithm. Section \ref{subsec:pw_exch} explores the fundamental principles that govern the construction of valid and efficient test scores. In Section \ref{subsec:fdr}, we establish finite-sample FDR theory for the prototype algorithm presented in Section \ref{subsec:propotype}, building upon the principles outlined in Section \ref{subsec:pw_exch}. The theory in Section \ref{subsec:fdr} draws upon the concept of generalized e-values, serving as the foundation for a generic information-pooling framework detailed in Section \ref{subsec:deran-claw}. 

\subsection{Problem formulation and a prototype algorithm}\label{subsec:propotype}

A multiple testing procedure can be represented by a binary decision rule $\pmb\delta=(\delta_i: 1\leq i\leq m)\in \{0, 1\}^m$, where $\delta_i=1$ indicates that we reject $H_i$ and $\delta_i=0$ otherwise. Let $\mathcal{R}=\{i\in[m]: \delta_i=1\}$ denote the index set of rejected hypotheses, and $\mathcal H_0=\{i\in[m]: H_i \mbox{ is true}\}$ the index set of null hypotheses. Then the \emph{false discovery proportion} (FDP) and true discovery proportion (TDP) are respectively defined as 
\begin{equation}\label{FDP-TDP}
\mathrm{FDP}(\mathcal R)= \frac{|\mathcal{R}\cap\mathcal{H}_{0}|}{|\mathcal{R}|\vee1}\; \mbox{ and } \;  \mathrm{TDP}(\mathcal R)= \frac{|\mathcal{R}\setminus\mathcal{H}_{0}|}{|\mathcal{H}_{0}^{c}|\vee1}, 
\end{equation}
where $|\mathcal{A}|$ represents the cardinality of a set $\mathcal{A}$. The FDR is the expected value of the FDP: $\mathrm{FDR}=\EE\{\mathrm{FDP}(\mathcal R)\}$, where the expectation is taken over the joint distribution of the null samples ${\mathbf T}^0$, test data $\mathbf T$ and auxiliary data $\mathbf S$. We employ the \emph{average power} (AP), defined as $\mathrm{AP}=\EE\left\{\mathrm{TDP}(\mathcal R)\right\}$, to compare the efficiency of different multiple testing procedures. 


Our prototype algorithm operates with the following pairs \(\{(u_i, \tilde{u}_i) : i \in [m]\}\), which represent the test and calibration scores, respectively. The construction of \(\{(u_i, \tilde{u}_i) : i \in [m]\}\) involves selecting \(m\) null samples from \(\mathcal{D}_0\) to form a calibration set \(\mathcal{D}^{cal}\), leading to the basic requirement that \(|\mathbf{T}^0| \geq m\). Let \(\tilde{\mathbf{T}} = (T_i^0 : i \in \mathcal{D}^{cal}) \coloneqq (\tilde{T}_i)_{i=1}^m\). If the data points in \(\mathbf{T}^0 = \{T_i : i \in \mathcal{D}_0\}\) are exchangeable conditoinal on \(\mathbf{S}\), then the above operation is equivalent to randomly selecting \(\tilde{T}_i\) from \(\mathbf{T}^0\) (without replacement) to form the triples \((T_i, \tilde{T}_i, S_i)_{i=1}^m\).

\begin{remark}\rm{
We briefly discuss several issues regarding the utilization of the null samples \(\mathbf{T}^0\). First, if the points in \(\mathbf{T}^0\) are non-exchangeable conditional on \(\mathbf{S}\), then randomly sampling \(\tilde{T}_i\) from \(\mathbf{T}^0\) may be inappropriate; careful attention is required to ensure the fulfillment of the exchangeability condition outlined in Section \ref{subsec:pw_exch}; see Example \ref{example:multi-class} in Section \ref{app:pwexch} of the Supplement. Second, if \(|\mathbf{T}^0| \gg m\), then the additional null samples \((T_i^0: i \in \mathcal{D}_0 \setminus \mathcal{D}^{cal})\), denoted by $\mathbf{T}^{tr0}$, can be incorporated into the training dataset \(\mathbf{T}^{tr}\) to build a predictive model within a semi-supervised framework; further discussion can be found in Section \ref{app:subsub-pu-group} and Section \ref{subsec:ssmt} of the Supplement. Moreover, the extra null samples may be utilized to derandomize our algorithm (Section \ref{subsec:deran-claw}), as demonstrated in \cite{ren2023derandomized} and \cite{bashari2023derandomized}.
}
\end{remark}

Both $u_i$ and $\tilde u_i$ are computed via a bivariate function, denoted as $g(\cdot, \cdot)$, and can be represented in the following form: 
\begin{equation}\label{biv-scores}
\left\{u_i\equiv g(T_i, S_i), \tilde u_i\equiv g(\tilde T_i, S_i): i\in[m]\right\}. 
\end{equation}
The bivariate funtion $g(\cdot, \cdot)$  is carefully designed to incorporate information from relevant datasets $\mathbf T\cup \tilde{\mathbf T}\cup \mathbf S$, guaranteeing that $u_i$ and $\tilde{u}_i$ fulfill the principle of \emph{pairwise exchangeability}, a fundamental notion thoroughly developed and explained in Section \ref{subsec:pw_exch}. As a warm-up, the primary focus of Section \ref{sec:preliminary} is to outline the basic structure of our algorithm and present a generic theory that facilitates the understanding of the core principles in later methodological developments. The intricacies in constructing $g(t, s)$ are deferred to Section \ref{sec:claw}. 

The complexity associated with the scores \eqref{biv-scores} significantly exceeds that of conventional significance indices, such as the p-value, making the derivation of the null distribution for these scores a challenging and often infeasible task. Consequently, we adopt the perspective of conformal inference, where $u_i$ are interpreted as conformity scores, assessing how well the scores in the test set {conform to} those computed from the null samples in $\mathcal D^{cal}$. This framework offers a significant advantage by eliminating the need for a known null distribution. Instead, the decision process solely relies on the relative ranks of the scores. By convention, a lower score corresponds to a higher rank, providing strong evidence against the null hypothesis. 

Denote $\mathcal{U}=\{u_i\equiv g(T_i, S_i): i\in[m]\}$ and $\Tilde{\mathcal{U}}=\{\Tilde{u}_{i}\equiv g(\Tilde T_i, S_i): i\in[m]\}$ the sets of conformity scores computed for the test and calibration sets, respectively. We focus on a class of decision rules that reject $H_{i}$ if (a) $u_{i}$ is smaller than its calibrated counterpart $\tilde u_i$ and (b) $u_{i}$ falls below a data-driven threshold, which will be determined using the following $Q(t)$ process: 
\begin{equation} \label{confq}
            \tau = \max\left\{ t\in\mathcal{U}\cup\Tilde{\mathcal{U}}: Q(t)\equiv\frac{1+\sum_{i=1}^{m}\II\{u(\tilde T_i, S_i)\leq t\wedge u(T_i, S_i)\}}{\left[\sum_{i=1}^{m}\II\{u(T_i, S_i)\leq t\wedge u(\tilde T_i, S_i)\}\right]\vee 1} \leq \alpha \right\}. 
        \end{equation}  
This above formulation draws inspiration from techniques employed in knockoff filters for variable selection problems \citep{barber15knockoff,weinstein17counting} and the empirical process perspective for the conformal BH algorithm \citep{mary22semi, marandon22mlfdr}. The corresponding decisions are given by $\pmb\delta=(\delta_i: 1\leq i\leq m)$, where $\delta_i=\II\{u_{i}\leq\tau\wedge\Tilde{u}_{i}\}$. According to mathematical conventions, we set $\tau=-\infty$ if the set $\{t\in\mathcal{U}\cup\Tilde{\mathcal{U}}:Q(t)\leq\alpha\}$ is empty, and thus no rejection is made. The aforementioned steps are summarized in Algorithm \ref{algo:claw} below.

\begin{algorithm}
\renewcommand{\algorithmicrequire}{\textbf{Input : }}
\renewcommand{\algorithmicensure}{\textbf{Output : }}
\caption{A prototype algorithm} \label{algo:claw}
\begin{algorithmic}[1]
\Require Pre-specified FDR level $\alpha$, the null samples ${\mathbf T}^0=\{T_i: i\in\mathcal D_0\}$, test data with the corresponding covariate sequence $(T_i, S_i)_{j=1}^m$. 
\Ensure The set of rejected indices $\mathcal{R}\subset[m]$.

\State Learn conformity scores $\mathcal{U}=\{u_{i}:i\in[m]\}$ and corresponding calibration scores $\Tilde{\mathcal{U}}=\{\Tilde{u}_{i}:i\in[m]\}$ such that $u_i$ and $\Tilde u_i$ are pairwise exchangeable.

\State Determine the threshold $\tau$ according to the $Q(t)$ process defined in \eqref{confq} and reject hypotheses in the set $\mathcal{R}=\{i\in[m]: u_{i}\leq\tau\wedge\Tilde{u}_{i}\}$. 

\State \textbf{Return} the set of rejected indices $\mathcal{R}$.
\end{algorithmic}
\end{algorithm}

We conclude the subsection by explaining the rationale behind Algorithm \ref{algo:claw}. In Equation \eqref{confq}, our objective is to determine the maximum threshold that ensures the estimated FDP remains below the nominal level $\alpha$. To accomplish this, we employ $Q(t)$ as a conservative estimator of the FDP, where the number of false rejections $\sum_{i\in\mathcal H_0} \II\{{u}_{i}\leq t\wedge \Tilde u_{i}\}$ is ``overestimated'' by $1+\sum_{i=1}^{m} \II\{\Tilde{u}_{i}\leq t\wedge u_{i}\}$. The efficacy of using $Q(t)$ to approximate the true FDP relies on how well $\sum_{i\in\mathcal{H}_{0}} \II\{\Tilde{u}_{i}\leq t\wedge u_{i}\}$ can mirror $\sum_{i\in\mathcal{H}_{0}} \II\{u_{i}\leq t\wedge\Tilde{u}_{i}\}$. Therefore, the validity of the algorithm critically depends on the fundamental assumption of \emph{pairwise exchangeability} between $u(T_i, S_i)$ and $u(\tilde{T}_i, S_i)$ for $i\in\mathcal{H}_{0}$. This key concept will be thoroughly elucidated next.


\subsection{Exchangeability notions in presence of side information}\label{subsec:pw_exch}


We first review commonly used notions of exchangeability for both data samples and conformity scores, and then extend these definitions to accommodate side information. Finally, we rigorously define the pairwise exchangeability between conformity scores.

The random elements in $\mathbf{Z}=(Z_{i}: i\in[m])$ are (jointly) exchangeable if their joint distribution is permutation-invariant, i.e.
$
(Z_1, \cdots, Z_m) \overset{d}{=} (Z_{\Pi_1}, \cdots, Z_{\Pi_m}),
$
where $(\Pi_1, \cdots, \Pi_m)$ represents any permutation of the indices $\{1, \cdots, m\}$.
A commonly employed assumption in conformal inference is the joint exchangeability between null samples:
\begin{equation}
    \left(T_i^{0}, i\in \mathcal D_0; T_{j},j\in\mathcal{H}_{0}\right) \text{ are exchangeable conditional on } (T_{j}:j\notin\mathcal{H}_{0}). \label{jointexch} 
\end{equation}
If the exchangeability condition \eqref{jointexch} holds, \cite{bates23} proposed a split-conformal strategy to construct  scores that fulfill the following exchangeability condition:
\begin{equation}
    \left(\Tilde{u}_i, i\in \mathcal{D}^{\mathrm{cal}}; u_{j}, j\in\mathcal{H}_{0}\right) \text{ are exchangeable conditional on } (u_{j}:j\notin\mathcal{H}_{0}).\label{jointexch-scores} 
\end{equation}
We emphasize that preserving the exchangeability property from \eqref{jointexch} to \eqref{jointexch-scores} poses substantial challenges when there is a need to incorporate extra data, such as the test data $\mathbf T$ and auxiliary covariates $\mathbf S$, alongside the training data $\mathbf{T}^{tr}$,  to construct score functions. 
Notably, \cite{yang21bonus} and \cite{marandon22mlfdr} designed an innovative class of score functions with specific permutation-invariance properties, allowing the integration of test data $\mathbf T$ into the score construction while ensuring the exchangeability condition \eqref{jointexch-scores}. This advancement improves the overall power of the analysis, while guaranteeing that the resulting conformal p-values (cf. Section \ref{app:relation}) remain super-uniform and still possess the PRDS property (positive regression dependency on subsets; cf. \citealp{by01}). However, the incorporation of covariates \(\mathbf{S}\) into the score construction process has yet to be explored.

Next, we extend the exchangeability assumption \eqref{jointexch} to encompass scenarios where side information is available. This generalized exchangeability assumption is formally stated as follows: 
\begin{equation}
    \left(T_i^{0}, i\in \mathcal D_0; T_{j},j\in\mathcal{H}_{0}\right) \text{ are exchangeable conditional on } (T_{j}:j\notin\mathcal{H}_{0};\mathbf{S}). \label{jointexch-covariate} 
\end{equation}
Assumption \eqref{jointexch-covariate} asserts that the joint structure of null samples remains unchanged, conditional on $\mathbf S$ and the remaining non-null samples. Initially, this assumption may appear to be strong. However, in light of the empirical Bayes model \eqref{model:eb-mix}-\eqref{model:mixture}, it becomes evident that \eqref{jointexch-covariate} is well-aligned with commonly utilized conditions in the multiple testing literature; see Section \ref{app:pwexch} of the Supplement for further examples and justifications on this condition.


\begin{remark}\rm{
While assumption \eqref{jointexch-covariate} primarily concerns the joint exchangeability of all null samples, which requires null data to be equally correlated, our methodology and theory remain applicable even when this assumption is relaxed to pairwise exchangeability of the null samples. In order to maintain conciseness, we have abstained from introducing additional new exchangeability concepts in the main text and provided extended discussions on generalized notions and theories in Section \ref{app:pwexch} of the Supplement.}
\end{remark}

We now introduce the \emph{pairwise exchangeability of conformity scores}, which serves as a foundational principle of Algorithm \ref{algo:claw}. This property can be rigorously stated as 
\begin{equation}
    \left( u_{i},\tilde{u}_{i}, \mathbf{U}_{-i},\tilde{\mathbf{U}}_{-i} \right) \overset{d}{=} \left( \tilde{u}_{i},u_{i}, \mathbf{U}_{-i},\tilde{\mathbf{U}}_{-i} \right) ,\quad\forall i\in\mathcal{H}_{0},\label{pwexch}  
\end{equation}
where $\mathbf{U}_{-i}=(u_{1},\cdots,u_{i-1},u_{i+1},\cdots,u_{m})$ and $\tilde{\mathbf{U}}_{-i}=(\tilde{u}_{1},\cdots,\tilde{u}_{i-1},\tilde{u}_{i+1},\cdots,\tilde{u}_{m})$. In contrast to \eqref{jointexch-covariate}, the auxiliary covariates $\mathbf{S}$ have been integrated into the conformity scores in \eqref{pwexch}, and therefore, the covariates are not explicitly represented in the conditions. The pairwise exchangeability condition \eqref{pwexch}, initially introduced in \cite{barber15knockoff}, has played a critical role in the development of knockoff filters for variable selection in regression models. Our research expands the scope of this notion beyond its original context by illustrating its applicability and effectiveness for conformalized multiple testing with side information. We highlight that our method fundamentally differs from the knockoff filters with side information \citep{ren23knockoff}. This point has been briefly mentioned in Section \ref{subsec:conformal_inference}, with further details provided in Section \ref{app:kn} of the Supplement.
 
The joint exchangeability \eqref{jointexch-scores} can be regarded as a more stringent form of pairwise exchangeability \eqref{pwexch}. It is not feasible to construct jointly exchangeable scores that satisfy \eqref{jointexch-scores} while incorporating \(\mathbf{S}\), as these covariates inherently introduce heterogeneity. In contrast, the construction of pairwise exchangeable scores that satisfy \eqref{pwexch} offers a practical approach for integrating side information. Thus, leveraging pairwise exchangeability provides greater flexibility and utility, resulting in covariate-informed conformity scores that exhibit both improved power and enhanced interpretability.
The construction of conformity scores that satisfy \eqref{pwexch} using data that obeys \eqref{jointexch-covariate} represents a pivotal yet highly challenging task. Addressing this challenge involves first formulating foundational principles (Sections \ref{subsec:fdr} and \ref{subsec:build_pw_exch_score}) and subsequently developing practical data-driven algorithms (Sections \ref{subsec:clfdr}-\ref{subsec:claw_clfdr}).

\subsection{Finite-sample theory on FDR control}
\label{subsec:fdr}

This section establishes the FDR theory of the prototype algorithm by linking Algorithm \ref{algo:claw} with the e-BH procedure \citep{wang22ev}. While alternative techniques, such as martingale or leave-one-out arguments, could also be used to establish finite-sample FDR theory, we employ the e-BH perspective for its flexibility in information aggregation. In Section \ref{subsec:deran-claw}, we examine this aspect in depth, highlighting its significant implications for integrative inference across various data sources, models, and methods.

Let \( E_j \) denote a non-negative random variable associated with \( H_j \), \( j \in [m] \). We define \( \{E_{j}, j \in [m]\} \) as a set of generalized e-variables if 
\begin{equation}\label{cond:generalev}
\mathbb{E} \left\{ \textstyle\sum_{j \in \mathcal H_0} E_{j} \right\}\leq m. 
\end{equation} 
Denote $e_j$ the observed value of \( E_j \). \citet{wang22ev} proposed the e-BH procedure for FDR control based on the classical Benjamini-Hochberg (BH) procedure \citep{bh95}. The rejection set of e-BH is given by 
\(
\mathcal{R}_{ebh} = \{ j : e_{j} \geq e_{(\hat{k})} \},
\) 
where \( e_{(1)} \geq e_{(2)} \geq \cdots \geq e_{(m)} \) are the order statistics, and the threshold 
\(
\hat{k} = \max\{i : \frac{i e_{(i)}}{m} \geq \frac{1}{\alpha}\}.
\) 
\citet{wang22ev} show that e-BH controls the FDR if \eqref{cond:generalev} holds.

Suppose the conformity scores in $\mathcal{U}=\{u_{i}:i\in[m]\}$ and $\Tilde{\mathcal{U}}=\{\Tilde{u}_{i}:i\in[m]\}$ are pairwise exchangeable. Define 
\begin{equation}\label{newev}
    e_{j} = \frac{m \mathbb{I}\{ u_{j}\leq \tau\wedge\Tilde{u}_{j}\} }{1+\sum_{i=1}^{m} \mathbb{I}\{ \Tilde{u}_{i}\leq \tau\wedge u_{i}\}}, 
\end{equation}
where $\tau$ represents the threshold specified in Algorithm \ref{algo:claw}. The next proposition reveals that Algorithm \ref{algo:claw} is equivalent to the e-BH algorithm employing generalized e-values defined in \eqref{newev}. 

\begin{proposition}\label{thm:ev}
    The variables $\{e_{j}: j\in[m]\}$ defined in (\ref{newev}) constitute a set of generalized e-values if the pairwise exchangeability \eqref{pwexch} holds and there is no tie between $u_{i}$ and $\Tilde{u}_{i}$ almost surely. When implementing the e-BH procedure with these e-values, the resulting rejection set $\mathcal{R}_{ebh}=\mathcal{R}$, where $\mathcal{R}=\{i:u_{i}\leq \tau\wedge\Tilde{u}_{i}\}$ is the index set of rejections output by Algorithm \ref{algo:claw}. 
\end{proposition}


The following theorem can be easily established as a corollary of Proposition \ref{thm:ev} and the theory for e-BH presented in \cite{wang22ev}. 

\begin{thm}\label{thm:fdr}
    If the pairwise exchangeability condition \eqref{pwexch} holds and there is no tie between $u_{i}$ and $\Tilde{u}_{i}$ almost surely, then Algorithm \ref{algo:claw} controls the FDR at level  $\alpha$.
\end{thm}

Our theory is provably valid even in scenarios where the empirical Bayes working model \eqref{model:eb-mix} diverges from the true data generating model. This notable robustness is attained by relying exclusively on the mild exchangeability condition \eqref{pwexch}, which significantly alleviates the strict assumptions prevalent in current theories.


\section{The CLAW Procedure and Its Theoretical Properties}
\label{sec:claw}


In Section \ref{subsec:build_pw_exch_score}, we highlight key issues and subsequently establish foundational principles for constructing conformity scores that satisfy \eqref{pwexch}. Detailed illustrations of the score construction process under a conformalized NEB framework are provided in Sections \ref{subsec:clfdr}-\ref{subsec:claw_clfdr}. This section assumes the null distribution $F_0$ is known; the scenario where $F_0$ is unknown but the null samples $\mathbf{T}^0$ are available (i.e., the semi-supervised setup) is addressed in Section \ref{sec:extension}.

\subsection{Constructing conformity scores: basic strategies and roadmap}
\label{subsec:build_pw_exch_score}


Consider a class of conformity scores described in the form of \eqref{biv-scores}:
$u_i = g(T_i, S_i)$ and $\Tilde{u}_i = g(\Tilde{T}_i, S_i)$, where both $u_i$ and $\Tilde{u}_i$ employ the same $g(\cdot, \cdot)$ with the same $S_i$. In an ideal scenario where $g(\cdot, S)$ is non-random given the auxiliary covariate $S$, the pairwise exchangeability condition \eqref{pwexch} naturally follows from the prescribed condition \eqref{jointexch-covariate}. For instance, if an oracle possesses relevant knowledge of the working model, then $g(\cdot, S_i)$ can be taken as the Clfdr function \eqref{clfdr-stat}. Under the oracle setting where \( f_0 \), \( \pi_{S_i} \), and \( f_{S_i}(\cdot) \) are known, the scores \( u_{i} = \mathrm{Clfdr}(T_{i}, S_{i}) \) and \( \tilde{u}_{i} = \mathrm{Clfdr}(\tilde{T}_{i}, S_{i}) \) are pairwise exchangeable and can therefore be utilized in Algorithm \ref{algo:claw}.

However, specifying the Clfdr function typically requires knowledge of unknown quantities, such as $\pi_{S_i}$ and $f_{S_i}(\cdot)$ , which need to be estimated from data in practical scenarios. The random nature of the data-driven function $g(t, s)$ complicates the matter significantly. In general, constructing an efficient $g(t, s)$ often requires utilizing training, test, calibration, and auxiliary data jointly. Previous studies, such as the AdaDetect algorithm proposed by \cite{marandon22mlfdr}, have highlighted the challenge of incorporating test data for training score functions. In our problem setting, the presence of covariates introduces an additional layer of complexity.

Next, we introduce a theorem that consolidates relevant theories to serve as guiding principles for constructing exchangeable score functions, encompassing both pairwise exchangeability and joint exchangeability notions. To simplify the notation, we introduce two operations: $(\mathbf{T},\Tilde{\mathbf{T}})_{\mathrm{swap}(\mathcal{J})}$ and $(\mathbf T, \Tilde{\mathbf{T}})_{\Pi}$, with the former denoting the swapping of $T_{j}$ and $\Tilde{T}_{j}$ for each $j\in\mathcal{J}$, $\mathcal J\subset [m]$, across $\mathbf T$ and $\Tilde{\mathbf{T}}$ (two vectors of equal length), while the latter representing an arbitrary permutation of the elements in the vector $(\mathbf T, \Tilde{\mathbf{T}})\equiv(T_1,\cdots, T_m, \Tilde T_1, \cdots, \Tilde T_m)$.
 
\begin{thm} \label{thm:exch}
Consider a class of score functions in the form of $g(\cdot,S_{i})\equiv g\left(\cdot,S_{i};(\mathbf{T},\Tilde{\mathbf{T}}),\mathbf{S}\right)$. Denote $u_{i}=g(T_{i},S_{i})$, $\Tilde{u}_{i}=g(\Tilde{T}_{i},S_{i})$, $\mathbf{U}=(u_{1},\cdots,u_{m})$ and $\Tilde{\mathbf{U}}=(\Tilde{u}_{1},\cdots,\Tilde{u}_{m})$. Then 
\begin{enumerate}[(a)]
        \item \, $\mathbf{U}$ and $\Tilde{\mathbf{U}}$ satisfy the pairwise exchangeability condition \eqref{pwexch} if (i) the score functions are swapping-invariant with respect to $(\mathbf{T},\Tilde{\mathbf{T}})$, i.e. 
       \begin{equation}\label{principle-pw}
        \mbox{$g\left(\cdot,S_{i};(\mathbf{T},\Tilde{\mathbf{T}})_{\mathrm{swap}(\mathcal{J})},\mathbf{S}\right) = g\left(\cdot,S_{i};(\mathbf{T},\Tilde{\mathbf{T}}),\mathbf{S}\right)$ for any $\mathcal{J}\subset[m]$ }; 
        \end{equation} 
      and  (ii) $\mathbf{T}$, $\Tilde{\mathbf{T}}$ and $\mathbf{S}$ satisfy the exchangeability condition \eqref{jointexch-covariate};
        
        \item \, $\mathbf{U}$ and $\Tilde{\mathbf{U}}$ satisfy the joint exchangeability condition \eqref{jointexch-scores}, if (i) the score functions are conditionally independent of $\mathbf S$ and permutation-invariant with respect to the elements in $\{\mathbf T, \Tilde{\mathbf{T}}\}$, i.e. 
        \begin{equation}\label{principle-jt}
     \mbox{  $g\left(\cdot,S_{i};(\mathbf{T},\Tilde{\mathbf{T}}),\mathbf{S}\right)= g\left(\cdot;(\mathbf{T},\Tilde{\mathbf{T}})\right)=g\left(\cdot;(\mathbf{T},\Tilde{\mathbf{T}})_{\Pi}\right)$ };  
        \end{equation}
     and (ii) $\mathbf{T}$ and $\Tilde{\mathbf{T}}$ satisfy the exchangeability condition \eqref{jointexch}.
    \end{enumerate}
\end{thm}

We provide two remarks regarding the theorem. Firstly, parts (a) and (b) draw inspiration respectively from relevant theories presented in \citet{barber15knockoff} and \citet{marandon22mlfdr}. Nevertheless, our work differs  from these studies in terms of research objectives. Our primary focus is to provide principles for incorporating side information, a perspective that was absent in previous works. Therefore, we have repurposed and consolidated existing theories to align with our specific problem setups. Secondly, Theorem \ref{thm:exch} exclusively addresses the conventional multiple testing setup, where $F_0$ is known and thus training data $\mathbf{T}^{tr}$ is not involved. For the semi-supervised setup, we present the principle, methodology and theory with revised notations in Section \ref{app:subsub-pu-group}, and Section \ref{subsec:ssmt} of the Supplement.


\subsection{Locally adaptive estimators} 
\label{subsec:clfdr}

In this subsection, we introduce estimators for \( \pi_S \) and \( f_S(\cdot) \) by making the \emph{local smoothness} assumption, which posits that units with similar values form ``local neighborhoods''. The relational knowledge encoded within the auxiliary sequence can be conveniently represented by a weight matrix \( \mathbf{W} \equiv \mathbf{W}(\mathbf{S}) = (w_{ij})_{i,j \in [m]} \), which regulates the relative contributions from units \( j \neq i \). Let \( d: \mathcal{X} \times \mathcal{X} \to \mathbb{R} \) be a distance function that characterizes the local neighborhood within the metric space \( \mathcal{X} \). Define the weight \( w_{ij} = W\{d(S_i, S_j)\} \), where \( W(\cdot) \) is a non-negative decreasing function. Denote the vector of weights associated with unit \( i \) as \( \mathbf{W}_{S_i}(\mathbf{S}) = (w_{ij} : j \in [m]) \). To simplify the notation, we denote \( \hat{f}_{S_i}^*(t) \equiv \hat{f}^*\{t; \mathbf{T}, \mathbf{W}_{S_i}(\mathbf{S})\} \) and \( \hat{\pi}^*_{S_i} \equiv \hat{\pi}^*\{\mathbf{T}, \mathbf{W}_{S_i}(\mathbf{S})\} \), omitting the common data \( \mathbf{S} \) and \( \mathbf{T} \) shared across all units.

Let $p(T_{j})$ denote the p-value associated with hypothesis $H_{j}$. We begin by considering a class of kernel estimators for $f_{S_i}(t)$ and $\pi_{S_i}$ that allow us to incorporate the local neighborhood information through the weight matrix:
\begin{equation}
    \hat{f}_{S_i}^{*}(t)=\frac{\sum_{j=1}^{m}w_{ij}K_{h}(t-T_{j})}{\sum_{j=1}^{m}w_{ij}} \quad \text{and} \quad \hat{\pi}_{S_{i}}^{*}=1-\frac{\sum_{j=1}^{m}w_{ij}\mathbb{I}\{p(T_{j})>\lambda\}}{(1-\lambda)\sum_{j=1}^{m}w_{ij}}, \label{conventional_est}
\end{equation}
where $\lambda\in(0,1)$ is a pre-specified tuning parameter with the default choice of $\lambda=0.5$, $K_{h}(t)=h^{-1}K(t/h)$, and $K(t)$ represents a symmetric kernel function that satisfies $\int K(t)dt=1, \int tK(t)dt=0$, and $\int t^{2}K(t)dt<\infty,$ with $h$ being the bandwidth of the kernel function. The weighting strategy in \eqref{conventional_est} allows for an adaptive utilization of the available data: we borrow information from the entire sequence $\mathbf{T}$, but the units are treated differentially according to the structural information encoded in the weight matrix $\mathbf W$.

The kernel estimators \eqref{conventional_est} are intuitively appealing and encompass well-established estimators in the literature. In the scenario where $S_{i}$ represents the spatial location of unit $i$, these kernel estimators assign higher weights to units in close proximity, reflecting a local neighborhood effect. This intuition can be generalized to the case where $S_i$ is within a space $\mathcal X$ defined by a metric $d$: if $\mathcal{X}=\mathbb{R}$ and $w_{ij}=K_h(|S_i-S_j|)$, then $\hat{f}_{S_i}^{*}(t)$ recovers a variation of the bivariate density estimator proposed in \cite{CARS}, while $\hat{\pi}_{S_{i}}^{*}$ recovers the kernel estimator for the non-null proportion introduced in \cite{LAWS22}. Alternatively, when dealing with a discrete variable $S_i\in[K]$, we can set $w_{ij}=\II\{S_{i}=S_{j}\}$.  Suppose $S_i$ represents the group membership, with $K$ being the total number of groups. For $S_i=k$, $\hat{f}_{S_i}^{*}(t)\equiv \hat{f}_k^{*}(t)$ simplifies to a standard kernel density estimator constructed based on the data from the $k$th group $\{T_i: S_i=k\}$ (\citealp{cs09}). Similarly, $\hat{\pi}_{S_{i}}^{*}\equiv \hat{\pi}_{k}^{*}$ reduces to Storey's estimator (\citealp{SchSpj82, storey02}) of the non-null proportion in the $k$th group. The asymptotic properties of these kernel estimators have been investigated in \cite{CARS} and \cite{LAWS22}. 

Unfortunately, the estimators $\hat{f}_{S_i}^{*}(t)$ and $\hat{\pi}_{S_{i}}^{*}$ cannot be directly employed to construct pairwise exchangeable scores. Drawing inspiration from the strategy introduced in \citet{marandon22mlfdr}, we deliberately combine the calibration set $\tilde{\mathbf{T}}$ with the primary data $\mathbf{T}$ to ``conformalize'' the estimators presented in \eqref{conventional_est}. The conformalized version $f_{S_i}^{**}(t)$, which is explained in detail in Section \ref{app:estimate}, is defined as follows:
\begin{equation}\label{f-ds}
    \hat{f}_{S_{i}}^{**}(t) = \frac{\sum_{j=1}^{m}w_{ij}[K_{h}(t-T_{j})+K_{h}(t-\tilde{T}_{j})]}{2\sum_{j=1}^{m}w_{ij}}. 
\end{equation}
We emphasize that the bandwidth $h$ should either be pre-specified or determined using a data-driven rule that guarantees permutation-invariance with respect to $\mathbf{T}$ and $\Tilde{\mathbf{T}}$. See Appendix Section \ref{app:estimate} for more details on how the Silverman's rule \citep{silverman18density} or the Sheather-Jones method \citep{sheather91} may be tailored to select data-driven $h$ in the conformalization process. Similarly, $\hat{\pi}_{S_{i}}^{*}$ should be modified as: 
\begin{equation}\label{pi-ds} 
  \hat {\pi}^{**}_{S_{i}}=1-\frac{\sum_{j=1}^{m}w_{ij}[\II\{p(T_{j})>\lambda\}+\II\{p(\Tilde{T}_{j})>\lambda\}]}{2(1-\lambda)\sum_{j=1}^{m}w_{ij}},
\end{equation}
where $p(\Tilde{T}_{j})$ represents the p-value corresponding to $\Tilde{T}_{j}$, computed in the same way as for computing $p(T_{j})$. In Section \ref{subsec:claw_clfdr}, we prove that the conformalized estimators \eqref{f-ds} and \eqref{pi-ds} satisfy the guiding principle \eqref{principle-pw} in Theorem \ref{thm:exch} (a). Subsequently, we can construct conformity score function to emulate the Clfdr \eqref{clfdr-stat}:
\begin{equation}\label{hatclfdr**}
   \widehat{\mathrm{Clfdr}}^{**}(t,S_{i})=\frac{(1-\hat{\pi}_{S_i}^{**})f_{0}(t)}{\hat{f}_{S_{i}}^{**}(t)}, \quad \forall i\in[m]. 
\end{equation}

However, unlike previous findings (cf. Remark 2.1 and Theorem 4.1 in \citealp{marandon22mlfdr}) that suggest the optimal ranking remains unaffected during the conformalization process, the contamination of the mixture model from the calibrated data introduces systematic bias and significant complexities in the presence of side information. The next subsection focuses on the development of a strategy to effectively mitigate the systematic bias.
 
\subsection{The CLAW procedure and its validity}
\label{subsec:claw_clfdr}

To demonstrate that $\widehat{\mathrm{Clfdr}}^{**}(t,S_{i})$ can be systematically biased, we examine the large-$m$ limits of our estimators, assuming standard assumptions in kernel estimation (cf. \citealp{fan03nonlinear, CARS}). Specifically, these limits are expressed as $\hat{f}_{S_i}^{**}(t) \overset{p}{\to} \frac{1}{2} [f_{S_i}(t)+f_{0}(t)]$ and $\hat{\pi}^{**}_{S_{i}} \overset{p}{\to} \frac{1}{2}\pi_{S_{i}}$, where $\overset{p}{\to}$ indicates convergence in probability.  Consequently, we have the following relationship:
\begin{equation*}
    \widehat{\mathrm{Clfdr}}^{**}(t,S_{i}) \overset{p}{\to} \frac{(1-\pi_{S_i}/2)f_{0}(t)}{(1-\pi_{S_i}/2)f_{0}(t)+(\pi_{S_i}/2)f_{1S_i}(t)} := \mathrm{Clfdr}^{**}(t,S_{i}). 
\end{equation*}
Clearly, $\mathrm{Clfdr}^{**}(t,S_{i})$ can deviate significantly from $\mathrm{Clfdr}(t,S_{i})={(1-{\pi}_{S_{i}})f_{0}(t)}/{{f}_{S_{i}}(t)}$. 

To effectively emulate the ranking of $\mathrm{Clfdr}(t,S_{i})$, we introduce a mapping as follows: 
\begin{equation}
 \mathrm{Clfdr}(t,S_{i})=\frac{(1-\pi_{S_i})f_{0}(t)}{(1-\pi_{S_i})f_{0}(t)+\pi_{S_i}f_{1S_i}(t)} \overset{x\mapsto x/(1-x)}{\Longrightarrow}
\frac{(1-\pi_{S_i})f_{0}(t)}{\pi_{S_i}f_{1S_i}(t)}=:R(t,S_{i}). \label{or_Rt}        
\end{equation}
Since the mapping $x\mapsto x/(1-x)$ is strictly increasing on the interval $(0,1)$, the ranking of scores remains unchanged after the transformation. Thus, employing $R(t,S_{i})$ as the score function is equivalent to utilizing $\mathrm{Clfdr}(t,S_{i})$, as our prototype algorithm operates based on the relative ranks of the scores rather than their absolute values. Some elementary calculation reveals the relationship between $\mathrm{Clfdr}^{**}(t,s)$ and $R(t,s)$: 
\begin{equation*}
        \frac{2-\pi_{S_i}}{1-\pi_{S_i}}[(\mathrm{Clfdr}^{**}(t,S_{i}))^{-1}-1]=\frac{\pi_{S_i}f_{1S_i}(t)}{(1-\pi_{S_i})f_{0}(t)}=R^{-1}(t,S_{i}). 
\end{equation*}
Utilizing large-$m$ limits $\hat{\pi}_{S_i}^{**}\overset{p}{\to}\pi_{S_i}/2$ and $\widehat{\mathrm{Clfdr}}^{**}(t,S_{i})\overset{p}{\to} \mathrm{Clfdr}^{**}(t,S_{i})$, we propose to estimate $R(t,S_{i})$ by 
\begin{equation}
    \hat{R}(t,S_{i})=\frac{1/2-\hat{\pi}_{S_i}^{**}}{1-\hat{\pi}_{S_i}^{**}}\frac{\widehat{\mathrm{Clfdr}}^{**}(t,S_{i})}{1-\widehat{\mathrm{Clfdr}}^{**}(t,S_{i})}. \label{Rthat} 
\end{equation}
Consequently, under certain regularity conditions, $\hat{R}(t,S_{i})\overset{p}{\to}{R}(t,S_{i})$ as $m\rightarrow\infty$, demonstrating the efficacy of the transformation \eqref{Rthat}.

\begin{remark}\rm{
We discuss two small modifications for the quantities in \eqref{Rthat} in practical situations. First, since $\hat\pi^{**}_{S_i}\overset{p}{\to}\pi_{S_i}/2\in[0,1/2]$ in large-$m$ limits, we modify $\hat\pi^{**}_{S_i}$ to $\tilde\pi^{**}_{S_i}$ to ensure that the proportion estimator remains within the valid range of $[0, 1/2]$: 
\begin{equation}\label{pii-hat}
    \tilde\pi^{**}_{S_i}=\epsilon\II\{\hat\pi^{**}_{S_i}\leq0\} + (1/2-\epsilon)\II\{\hat\pi^{**}_{S_i}>1/2\} + \hat\pi^{**}_{S_i} \II\{0<\hat\pi^{**}_{S_i}\leq 1/2\}, 
\end{equation}
where we may set $\epsilon=0.001$. Further, as the Clfdr represents the posterior probability, the estimated Clfdr in \eqref{hatclfdr**} is modified as $\widetilde{\mathrm{Clfdr}}^{**}(t,S_{i})=\min\left\{{(1-\tilde{\pi}_{S_i}^{**})f_{0}(t)}/{\hat{f}_{S_{i}}^{**}(t)}, c\right\} $, where we may set $c=0.999$. In our numerical studies, CLAW is implemented by plugging $ \tilde\pi^{**}_{S_i}$ and $\widetilde{\mathrm{Clfdr}}^{**}(t,S_{i})$ into equation \eqref{Rthat}. 
}
\end{remark}

The newly introduced score function \( \hat{R}(t, S_{i}) \) presented in \eqref{Rthat} possesses two desirable properties. Firstly, it faithfully emulates the Clfdr ranking in the large-\( m \)-limit scenario; a theoretical justification for using the \( \mathrm{Clfdr} \) ranking (or \( R(t, s) \)) is provided in the next subsection. Secondly, the score function satisfies the guiding principle \eqref{principle-pw}. Hence, we can generate pairwise exchangeable scores \( u_{i} = \hat{R}(T_{i}, S_{i}) \) and \( \tilde{u}_{i} = \hat{R}(\tilde{T}_{i}, S_{i}) \), which in turn are employed in the prototype algorithm. The key steps of the proposed CLAW procedure are summarized in Algorithm \ref{algo:claw_clfdr}, with its theoretical properties established in Theorem \ref{thm:CLAW}.

\begin{algorithm}[!htbp]
\renewcommand{\algorithmicrequire}{\textbf{Input : }}
\renewcommand{\algorithmicensure}{\textbf{Output : }}
\caption{The CLAW procedure} \label{algo:claw_clfdr}
\begin{algorithmic}[1]
\Require The sequence of triples $({T}_{i},\Tilde T_i, S_i)_{i=1}^m$, the target FDR level $\alpha$. 
\Ensure The index set of rejected hypotheses $\mathcal{R}\subset[m]$. 
\State Construct the weight matrix $\mathbf W$ based on auxiliary covariates $\mathbf{S}$.
\ForAll{$i$ in $[m]$}
\State Compute conformalized estimators $\hat{f}^{**}_{S_i}(t)$ and $\hat{\pi}^{**}_{S_{i}}$ based on \eqref{f-ds} and \eqref{pi-ds}. 
\State Construct conformalized score functions $\widehat{\mathrm{Clfdr}}^{**}(t,S_{i})$ by \eqref{hatclfdr**}.
\State Transform $\widehat{\mathrm{Clfdr}}^{**}(t,S_{i})$ to $\hat{R}(t,S_{i})$ via \eqref{Rthat}. Obtain $u_{i}=\hat{R}(T_{i},S_{i})$ and $\Tilde{u}_{i}=\hat{R}(\Tilde{T}_{i},S_{i})$. 
\EndFor
\State Apply Algorithm \ref{algo:claw} with $u_{i}$ and $\Tilde{u}_{i}$ obtained in the previous step. Let $\mathcal{R}=\{i\in[m]: u_{i}\leq\tau\wedge\Tilde{u}_{i}\}$.
\State \textbf{Return} The rejection set $\mathcal{R}$.
\end{algorithmic}
\end{algorithm}

\begin{thm}\label{thm:CLAW}(Validity of the CLAW procedure).
 Suppose $(\mathbf{T}, \Tilde{\mathbf{T}},\mathbf{S})$ satisfy the exchangeability condition \eqref{jointexch-covariate}, and $u_i$ and $\tilde u_i$ are computed according to Algorithm \ref{algo:claw_clfdr}. Then (a)
$(u_i:{i\in[m]})$ and $(\Tilde{u}_{i}:{i\in[m]})$ satisfy the pairwise exchangeability \eqref{pwexch}.
(b) If there is no tie between $u_i$ and $\Tilde{u}_i$ almost surely, then Algorithm \ref{algo:claw_clfdr} controls the FDR at level $\alpha$.
\end{thm}

Although Algorithm \ref{algo:claw_clfdr} employs conformalized empirical Bayes estimators, the validity of our inference framework hinges solely on the pairwise exchangeability condition \eqref{pwexch}. Compared to conventional empirical Bayes FDR procedures, the theoretical guarantees of CLAW remain unaffected under model mis-specifications, which enhances its practicality and applicability. Moreover, Algorithm \ref{algo:claw_clfdr} is a specialized version of the generic Algorithm \ref{algo:claw}, and any score functions satisfying the principle outlined in Theorem \ref{thm:exch} can be effectively applied within Algorithm \ref{algo:claw}, highlighting the flexibility and applicability of our proposed framework. See Section \ref{app:subsub-pu-group} and Section \ref{subsec:ssmt} of the Supplement for alternative methods of constructing score functions.

\subsection{An optimality theory tailored for BC algorithms}

The optimality theory concerning the use of Clfdr \eqref{clfdr-stat}, as established in \cite{CARS}, cannot be directly applied to our framework. The primary issue is that existing theories focus exclusively on rejection rules of the form $\II\{g(T_i,S_i)\leq t\}$, whereas our decision rule rejects hypotheses only within the \textit{candidate rejection set}, denoted by $\mathcal{A}$, and takes the form $\II\{u_i\leq t\wedge\tilde{u}_i\}=\II\{u_i\leq t\}\II\{i\in\mathcal{A}\}$. In light of an insightful referee's comment and as elucidated in Section \ref{app:relation}, several FDR procedures, including knockoff filters \citep{barber15knockoff}, AdaPT \citep{lei18adapt}, and CLAW, fall within the class of Selective SeqStep+ algorithms, or BC algorithms \citep{barber15knockoff}. 
Concretely, knockoff filters \citep{barber15knockoff,ren23knockoff} focus exclusively on the subset of features for which the corresponding anti-symmetric statistic is positive, AdaPT \citep{lei18adapt} only considers the subset of hypotheses for which $p_i<1-p_i$, and CLAW concentrates on the subset $\mathcal{A}=\{i:u_i<\tilde{u}_i\}$. 

Therefore, we present Proposition \ref{prop:opt} to establish an optimality theory specifically tailored for BC-type algorithms. The proposition does not claim that the rule is universally the most powerful, as its optimality is confined to the restricted subset $\mathcal{A}$. However, given the significance of BC-type algorithms, this theory may hold independent interest. For further discussions, please refer to Section \ref{app:opt_or} of the Supplement. This optimality theory is developed under an oracle setup in which the parameters in the model \eqref{model:mixture} are assumed to be known. Furthermore, the marginal FDR $\mathrm{mFDR}(\mathcal{R})=\EE(|\mathcal{R}\cap\mathcal{H}_{0}|)/\EE(|\mathcal{R}|)$ has been employed in place of the conventional FDR to simplify the theoretical derivations.

\begin{proposition} \label{prop:opt} 
    Suppose $\{(T_i, S_i, \theta_i) : i \in [m]\}$ are generated from the covariate-adaptive model \eqref{model:mixture} and $\tilde{T}_i \overset{i.i.d.}{\sim}{f_0}$. 
    Assume an oracle setup in which $R(\cdot,\cdot)$ defined in \eqref{or_Rt} is known for all $i\in[m]$.  Consider oracle scores $u_i = R(T_i,S_i)$ and $\tilde{u}_i=R(\tilde{T}_i,S_i)$, and candidate set $\mathcal{A}:=\{i:u_i<\tilde{u}_i\}$. 
    Let $\mathcal{R}_{u} = \{i\in\mathcal{A} : u_i \leq t^{*}\}$ be the rejection set for some threshold $t^{*}$
    such that $\mathrm{mFDR}(\mathcal{R}_{u}) = \alpha$.
    Then for any rejection rule $\mathcal{R}\subset\mathcal{A}$ such that $\mathrm{mFDR}(\mathcal{R}) \leq \alpha$,
    we have that $\EE(|\mathcal{R}_{u}\cap\mathcal{H}_{0}^{c}|)\geq \EE(|\mathcal{R}\cap\mathcal{H}_{0}^{c}|)$.
\end{proposition}

\section{Extensions and Related Works}
\label{sec:extension} 

In this section, we first discuss how CLAW can be employed to handle the semi-supervised setup (Section \ref{app:subsub-pu-group}). Furthermore, we explore the incorporation of data collected from multiple sources within the CLAW framework (Section \ref{subsec:deran-claw}). Finally, we cast CLAW within the broader context of conformal inference and highlight its connections with related approaches (Section \ref{app:relation}).

\subsection{Semi-supervised CLAW with grouped hypotheses}
\label{sec:relation}\label{app:subsub-pu-group}

Assume that $(\mathbf T^0, \mathbf T, \mathbf S)$ satisfy the exchangeability condition \eqref{jointexch-covariate}, where $\mathbf{T}^0=\mathbf{T}^{tr}\cup\Tilde{\mathbf{T}}$\footnote{Our framework can be employed to integrate training data \(\mathbf{T}^{tr}\) that encompasses samples from various data sources (cf. Appendix \ref{app:pwexch}). However, in this section we have intentionally chosen to utilize \(\mathbf{T}^{tr}\) to denote \(\mathbf{T}^{tr0}\) as defined in Remark 1, thereby ensuring alignment with the data-splitting strategies implemented by \citet{bates23} and \citet{marandon22mlfdr}. This slight abuse of notation serves to clarify the connections with related works, particularly in illustrating how the concepts from AdaDetect can be leveraged within the CLAW framework (Section \ref{app:subsub-pu-group}) and how existing conformal methods can be adjusted to meet the requirements of pairwise exchangeability (Section \ref{app:relation}).}. 

Although the null distribution is not explicitly known, CLAW can be implemented by directly estimating the unknown $f_0$. The estimate, denoted as $\hat{f}_{0}(t)=\hat{f}_{0}(t;\mathbf{T}^{tr})$, can be obtained by applying parametric or nonparametric methods  \citep{fan03nonlinear} on $\mathbf{T}^{tr}$. Since $\hat{f}_{0}(t)$ does not involve the test data $\mathbf{T}$ and calibration data $\Tilde{\mathbf{T}}$, substituting $\hat f_0$ for $f_0$ in \eqref{hatclfdr**} and \eqref{Rthat} fulfills the principle in Theorem \ref{thm:exch}(a), thereby producing pairwise exchangeable scores. 

This direct estimation approach may be prone to inaccuracy and instability. To mitigate these issues, it is possible to enhance the method by integrating techniques from the literature on semi-supervised classification with positive and unlabeled data (PU classification, cf. \citealp{PU_NIPS2014,PU_ML2020}). PU learning methods become especially effective when dealing with high-dimensional data. In this subsection, we explore PU learning issues concerning the case when $S_i$ takes values from a discrete set $\{1,\cdots,K\}$. This corresponds to the simple yet significant scenario of multiple testing with groups \citep{efron08, cs09}, enabling us to effectively demonstrate the benefits of CLAW over existing methods. The case where $S_i$ is a continuous variable is discussed in Section \ref{subsec:ssmt} of the Supplement.

Two commonly employed strategies for testing with groups, as mentioned in \cite{efron08}, are the \emph{pooled analysis} and \emph{separate analysis}. The former simply carries out an FDR analysis as usual, discarding the grouping variable $S_i$. The latter first conducts group-wise FDR analyses separately, and subsequently combines the testing results from separate groups. We discuss two variations of the AdaDetect algorithm \citep{marandon22mlfdr} that utilize PU learning techniques for out-of-distribution testing. The first variation adopts the pooled analysis strategy, estimating the density ratio of the samples in the training set to the combined samples from the test and calibration sets. The second variation, called $\mbox{SeparateAD}$, performs group-wise analysis by following the separate analysis strategy. Specifically, $\mbox{SeparateAD}$ first constructs conformity scores using a class of functions that are group-wise permutation-invariant, i.e.
\begin{equation}\label{group-g}
g\left(\cdot, k; \cup_{i: S_i=k}\{T_i, \Tilde T_i\}, \mathbf{T}^{tr} \right)\equiv g\left(\cdot, k; \left(\cup_{i: S_i=k}\{T_i, \Tilde T_i\}\right)_{\Pi}, \mathbf{T}^{tr} \right), 
\end{equation}
where $\Pi$ represents an arbitrary permutation of the elements in the set $\cup_{i: S_i=k}\{T_i, \Tilde T_i\}$. As the null samples within the same group satisfy the exchangeable condition \eqref{jointexch-covariate}, it follows from Theorem \ref{thm:exch} (b) that the scores constructed via \eqref{group-g} are jointly exchangeable within group $k$. Next, $\mbox{SeparateAD}$ applies AdaDetect to each group separately at level $\alpha_k$ and combines the rejections to form the final rejection set. However, the finite-sample FDR theory for $\mbox{SeparateAD}$ has not been established and it remains unclear how to adjust the $\alpha_k$'s to maximize power. 

Finally, we discuss the implementation of CLAW using PU learning techniques. Following the approach outlined in \cite{marandon22mlfdr}, we estimate the ratio of the density of $\mathbf{T}^{tr}$ to that of $\cup_{i: S_i=k}\{T_i, \Tilde T_i\}$ within the $k$-th group using the class of functions specified in \eqref{group-g}: 
\begin{equation}\label{app-pu-group-dr}
    \hat{r}(\cdot,k)=\hat{r}(\cdot,k;\cup_{i: S_i=k}\{T_i, \Tilde T_i\}, \mathbf{T}^{tr}). 
\end{equation}
Next, we discuss the estimation of the non-null proportion $\pi_k$. Since the exact knowledge of $f_0$ is unavailable, we initially construct conformal p-values using the available data. These p-values are then utilized in \eqref{pi-ds} to obtain the estimate. Further details can be found in Section \ref{app:subsub-confp} of the Supplement.
Let $\widehat{\mathrm{Clfdr}}^{**}(t,k)=(1-\hat{\pi}_{k}^{**})\hat{r}(t,k)$ for $S_{i}=k$. The ranking score function $\hat{R}(t,k)$ can be derived through the transformation \eqref{Rthat}. 

The following proposition establishes the pairwise exchangeability of the conformity scores $u_i= \hat{R}(T_i, S_i)$ and $ \tilde u_i= \hat{R}(\tilde T_i, S_i)$. Therefore, the semi-supervised CLAW procedure boils to implementing Algorithm \ref{algo:claw} with the pairs $(u_i, \tilde u_i)_{i=1}^m$.
\begin{proposition}\label{prop:group}
    Consider the score functions $\hat{R}(t,k)$ defined in \eqref{Rthat} with $\widehat{\mathrm{Clfdr}}^{**}(t,k)=(1-\hat{\pi}_{k}^{**})\hat{r}(t,k)$, where $\hat{\pi}_{k}^{**}$ and $\hat{r}(t,k)$ are defined in \eqref{pi-ds} and \eqref{app-pu-group-dr}, respectively. Let $u_i= \hat{R}(T_i, S_i)$, $ \tilde u_i= \hat{R}(\tilde T_i, S_i)$. Then $u_i$ and $\tilde u_i$ are pairwise exchangeable if $\mathbf{T}$, $\Tilde{\mathbf{T}}$, $\mathbf{T}^{tr}$ and $\mathbf{S}$ satisfy \eqref{jointexch-covariate}.
    \end{proposition}

In the classical setup for multiple testing, \cite{cs09} demonstrated that both pooled and separate FDR analyses could be uniformly enhanced by the CLfdr procedure. Similarly, in the semi-supervised setup, we anticipate that AdaDetect and SeparateAD, which respectively employ pooled and separate strategies, can be improved by CLAW. This claim is substantiated by the numerical experiments described in Section \ref{subsec:group}. 



\subsection{Multi-source data aggregation via CLAW}\label{subsec:deran-claw}

This section presents an exploratory investigation into a flexible framework for data aggregation. The motivating setup revolves around a scenario where we have collected multiple auxiliary sequences $(\mathbf{S}^{(k)}: k\in[K])$, where $\mathbf{S}^{(k)}$ represents an $m$-dimensional vector encoding the covariates associated with hypotheses $(H_j)_{j=1}^m$ from the $k$th auxiliary data source, $k\in[K]$. The task of effectively utilizing all of this data poses considerable challenges, as different $\mathbf{S}^{(k)}$ may be acquired in varying formats, assessed using disparate metrics, and measured with distinct units of measurement. To tackle this, our proposed framework leverages the property that the average of e-values remains an e-value. Consequently, we execute the prototype algorithm $K$ times, each time for a specific covariate sequence, generating $K$ e-values for each $H_i$; the $K$ e-values can be aggregated to derive an averaged e-value for each $H_i$. The generic framework allows the primary data, calibration data, as well as testing procedures, to vary across $k\in[K]$. The description of our proposal is provided in Algorithm \ref{algo:derandom}.

\begin{algorithm}
\renewcommand{\algorithmicrequire}{\textbf{Input : }}
\renewcommand{\algorithmicensure}{\textbf{Output : }}
\caption{The integrative CLAW procedure} \label{algo:derandom}
\begin{algorithmic}[1]
\Require The test data $\mathbf{T}^{(k)}$, calibration data $\Tilde{\mathbf{T}}^{(k)}$, training data $\mathbf{T}^{tr(k)}$, covariates $\mathbf{S}^{(k)}$, $k=1,\cdots,K$, where $K$ is the number of testing procedures implementation, denote the procedures by $\mathcal{P}_{1},\cdots,\mathcal{P}_{K}$; the FDR level for each implementation $\alpha_{1},\cdots,\alpha_{K}$, a target FDR level $\alpha$; e-value weights $v_{1},\cdots,v_{K}$.
\Ensure A set of rejection $\mathcal{R}\subset[m]$. 

\ForAll{$k=1,2,\cdots,K$} 

\State Implement procedure $\mathcal{P}_{k}$ at FDR level $\alpha_{k}$ on the whole dataset $(\mathbf{T}^{(k)},\Tilde{\mathbf{T}}^{(k)},\mathbf{T}^{tr(k)},\mathbf{S}^{(k)})$ to construct (generalized) e-values $e_{1}^{(k)},\cdots,e_{m}^{(k)}$.

\EndFor

\State Let $\bar{e}_{i}=\sum_{k=1}^{K}v_{k}e_{i}^{(k)}/\sum_{k=1}^{K}v_{k}$ for $i\in[m]$. Denote the ordered statistics by $\bar{e}_{(1)}\geq \bar{e}_{(2)}\geq\cdots\geq \bar{e}_{(m)}$. Let $\hat{k}=\max\{i:({i\bar{e}_{(i)}}/{m})\geq({1}/{\alpha})\}$.

\State Let $\mathcal{R}=\{i\in[m]:\bar{e}_{i}\geq \bar{e}_{(\hat{k})}\}$.

\State \textbf{Return} The rejection set $\mathcal{R}$.
\end{algorithmic}
\end{algorithm}

The following theorem establishes the validity of Algorithm \ref{algo:derandom} for FDR control.

\begin{thm}\label{thm:multiev}
Consider Algorithm \ref{algo:derandom}, assume that every procedure $\mathcal{P}_k$ produces e-values $\{e_{j}^{(k)}:j\in[m]\}$ such that (\ref{cond:generalev}) is fulfilled, $k\in[K]$. Then Algorithm \ref{algo:derandom} controls the FDR at level $\alpha$. 
\end{thm}

Algorithm \ref{algo:derandom} introduces a highly flexible framework with significant ramifications across various scenarios. Firstly, in the primary scenario that motivates this framework, the prototype algorithm is implemented to each distinct auxiliary sequence $\mathbf{S}^{(k)}=(S_{i}^{(k)}:i\in[m])$, while the test dataset $\mathbf{T}$ and calibration dataset $\Tilde{\mathbf{T}}$ remain the same across different implementations. The utilization of average e-values to integrate diverse types of side information into the inferential process offers a practical and intriguing perspective (cf. \citealp{banerjee2023harnessing}). Secondly, when we have access to only one test dataset $\mathbf{T}$ and one auxiliary sequence $\mathbf S$, but a substantial number of null samples are available for calibrating or training models,  Algorithm \ref{algo:derandom}  can be implemented to utilize multiple calibration sets of null samples for improving the reliability and stability. This important perspective has been embraced by recent works \citep{ren2023derandomized, bashari2023derandomized}. 
Thirdly, if we have obtained multiple sets of test data $\mathbf{T}^{(k)}$, $k\in[K]$, from $K$ different studies, Algorithm \ref{algo:derandom} may be tailored to perform global null or partial conjunction tests to aggregate the evidence across diverse studies.  
Lastly, the framework offered by Algorithm \ref{algo:derandom} opens up possibilities for the development of new integrative tools, which can take advantage of an ensemble of machine learning models \citep{liang22integrative}. 

The discussions presented in this section are preliminary and have raised various issues that warrant further exploration. These include determining the number of implementations $K$, selecting suitable $\alpha_k$ values across different implementations, assigning proper e-value weights, and tailoring existing meta-analysis methods to perform global null or partial conjunction tests. Additionally, achieving a balance between computational efficiency, statistical power, and algorithm stability is a significant topic of interest in this field.

\subsection{Connections to conformal methods}
\label{app:relation}

Let $u(\cdot)$ denote a conformity score function. The conformal p-value \citep{vovk05, bates23} calculates the standardized rank of the score associated with $H_i$ in $\mathcal{D}^{cal}$: 
\begin{equation}\label{def:conf_pv}
   \hat p_i (T_i)=\frac{ 1+|\{k\in \mathcal D^{cal}:  u({T}_{k}^0)\leq  u(T_{i})\}| }{1+|\mathcal D^{cal}|},\quad i\in[m]. 
\end{equation}
The score function $ u(\cdot)$, which is required to satisfy certain permutation-invariant properties, can be determined \emph{a priori} \citep{mary22semi,zhao23rank,jin23selection}, learned from training data $\mathbf{T}^{tr}$ \citep{bates23}, or carefully constructed using a combination of training, calibration and test data \citep{yang21bonus,marandon22mlfdr}. 

\cite{bates23} and \cite{marandon22mlfdr} show that the null p-values $\{p_{i}:i\in\mathcal{H}_{0}\}$ defined by \eqref{def:conf_pv} are super-uniform and PRDS when the null scores $\left\{u({T}_{i}^0), i\in \mathcal D^{cal}; u(T_{j}), j\in\mathcal{H}_{0} \right\}$ are exchangeable conditional on non-null scores $\{ u(T_{j}), j\notin\mathcal{H}_{0}\}$. Therefore, following  \cite{by01}, the BH algorithm, employed with conformal p-values \eqref{def:conf_pv}, is valid for FDR control. An alternative approach to establish the FDR theory is provided by \cite{mary22semi}, which demonstrates the equivalence between the conformal BH (CBH) algorithm and the counting knockoffs algorithm \citep{weinstein17counting} that rejects $H_{i}$  if $u(T_{i})\leq\hat{t}$, where 
\begin{equation}\label{def:conf_pv_bh}
    \hat{t}=\max\left\{ t\in\{u(T_{i})\}_{i\in[m]}: Q^*(t)\equiv \frac{\frac{1}{1+|\mathcal D^{cal}|}[1+\sum_{j\in\mathcal D^{cal}}\II\{u({T}_{j}^0)\leq t\}] }{ \frac{1}{m}\sum_{j=1}^{m}\II\{u(T_{j})\leq t\}} \leq \alpha \right\}. 
\end{equation} 

Under the broader scope of conformalized multiple testing, Algorithm \ref{algo:claw} modifies the CBH algorithm \eqref{def:conf_pv_bh} in several important respects. The first notable enhancement is the incorporation of side information in the new bivariate function \(u(T_i, S_i)\), which has significant implications for both score construction and theoretical analysis. The heterogeneity in \(S_i\) leads to conformity scores that are not jointly exchangeable, rendering the conformal p-values prescribed in \eqref{def:conf_pv} invalid and necessitating the development of new principles, methodologies, and theories.
The second adjustment involves setting \(|\mathcal{D}^{cal}| = m\) and removing the factor \(\frac{m}{1+m}\) in \(Q^*(t)\). While this adjustment incurs a minor loss of power, it appears to be indispensable for establishing FDR theory in finite samples.
The third adjustment entails substituting \(\sum_{j=1}^{m}\II\{u_i \leq t\}\) with \(\sum_{j=1}^{m}\II\{u_i \leq t \wedge \tilde{u}_i\}\), thereby aligning the FDP process with the Selective SeqStep+ algorithm \citep{barber15knockoff}. The last two modifications transform the FDP process within the CBH framework into a new mirror process that eliminates the need for jointly exchangeable scores. Furthermore, the thresholding rule \(\II\{u_i \leq t \wedge \tilde{u}_i\}\), which leverages both testing and calibration scores, allows the CLAW framework to adapt to varying sparsity levels. This adaptability is particularly advantageous in scenarios where \(\pi_s\) varies with the covariate value \(s\), as it sidesteps the challenging task of estimating the null proportion. Additional explanations and illustrations can be found in Sections \ref{app:cf}-\ref{app:kn} of the Supplement.

\section{Experiments with Simulated Data}
\label{sec:simu}

This section presents simulation results that compare CLAW with existing methods. Section \ref{subsec:group} (\ref{subsec:spatial}) investigates the where $S_i$ is discrete (continuous). Supplementary numerical results are provided in Appendix \ref{app:simu}, which include additional comparisons involving multivariate data, random covariates, and correlated data. The reported results are obtained by averaging over 200 replicated experiments, with the nominal FDR level set at $\alpha=0.05$.

\subsection{Multiple testing with discrete covariates (grouped hypotheses)}
\label{subsec:group}


Consider a scenario where $S_i$ takes two values, $\{1,2\}$, dividing the testing units into two distinct groups. The data is generated using a hierarchical approach, conditioned on $S_i$, as follows: 
\begin{equation*}
\PP(\theta_{i}=1|S_{i}=k)=\pi_{k}, \quad T_{i}|(\theta_{i},S_{i}=k) \overset{\text{ind.}}{\sim} (1-\theta_{i})\mathcal{N}(0,1)+\theta_{i}F_{1k}, \quad k\in\{1,2\}, \quad i\in[m]. 
\end{equation*}
Let $m_{k}=|\{i:S_{i}=k\}|$. Our experiments have considered three settings:
\begin{enumerate}[1. ]
    \item $m_{1}=3000$, $\pi_{1}=0.2$, $F_{11}=\mathcal{N}(\mu,1)$, $\mu$ varies; $m_{2}=1500$, $\pi_{2}=0.1$, $F_{12}=\mathcal{N}(-2,0.5^2)$. 
    \item $m_{1}=3000$, $\pi_{1}=0.2$, $F_{11}=\mathcal{N}(2,1)$; $m_{2}=1500$, $\pi_{2}=p$ with $p$ varying, $F_{12}=\mathcal{N}(-4,1)$. 
    \item  $m_{1}=3000$, $\pi_{1}=0.2$, $F_{11}=\mathcal{N}(2,0.5^2)$; $m_{2}$ varies, $\pi_{2}=0.1$, $F_{12}=\mathcal{N}(-4,1)$.
\end{enumerate}

We compare CLAW, implemented by running Algorithm \ref{algo:claw_clfdr} using weights $w_{ij}=\II\{S_{i}=S_{j}\}$, with the following methods: 
\textbf{PooledBH}, which ignores $S_i$ and applies BH to all p-values; 
\textbf{SeparateBH}, which applies BH at $\alpha$ for both groups, and then outputs the rejection set by combining the rejections from both groups; 
\textbf{PooledAD}, which ignores $S_i$ and applies AdaDetect with kernel estimators; 
\textbf{SeparateAD}, which applies AdaDetect with kernel estimators at $\alpha$ for both groups, and then outputs the rejection set by combining the rejections from both groups;
\textbf{IHW}, which implement IHW \citep{Ignatiadis21IHW} by the R package \texttt{IHW}. 
CLAW, PooledAD, and SeparateAD all utilize the same calibration data, which are simulated as $\Tilde{T}_{i}\overset{i.i.d.}{\sim}\mathcal{N}(0,1)$ for $i\in[m]$. When estimating the mixture density, a Gaussian kernel with the bandwidth chosen by Silverman's rule \citep{silverman18density} is employed. We vary the parameters $\mu$, $\pi$ and $m_2$, and present in Figure \ref{fig:group1} the corresponding levels of FDR and AP.

\begin{figure}[!htbp]
    \centering
    \includegraphics[width=0.9\linewidth,height=4in]{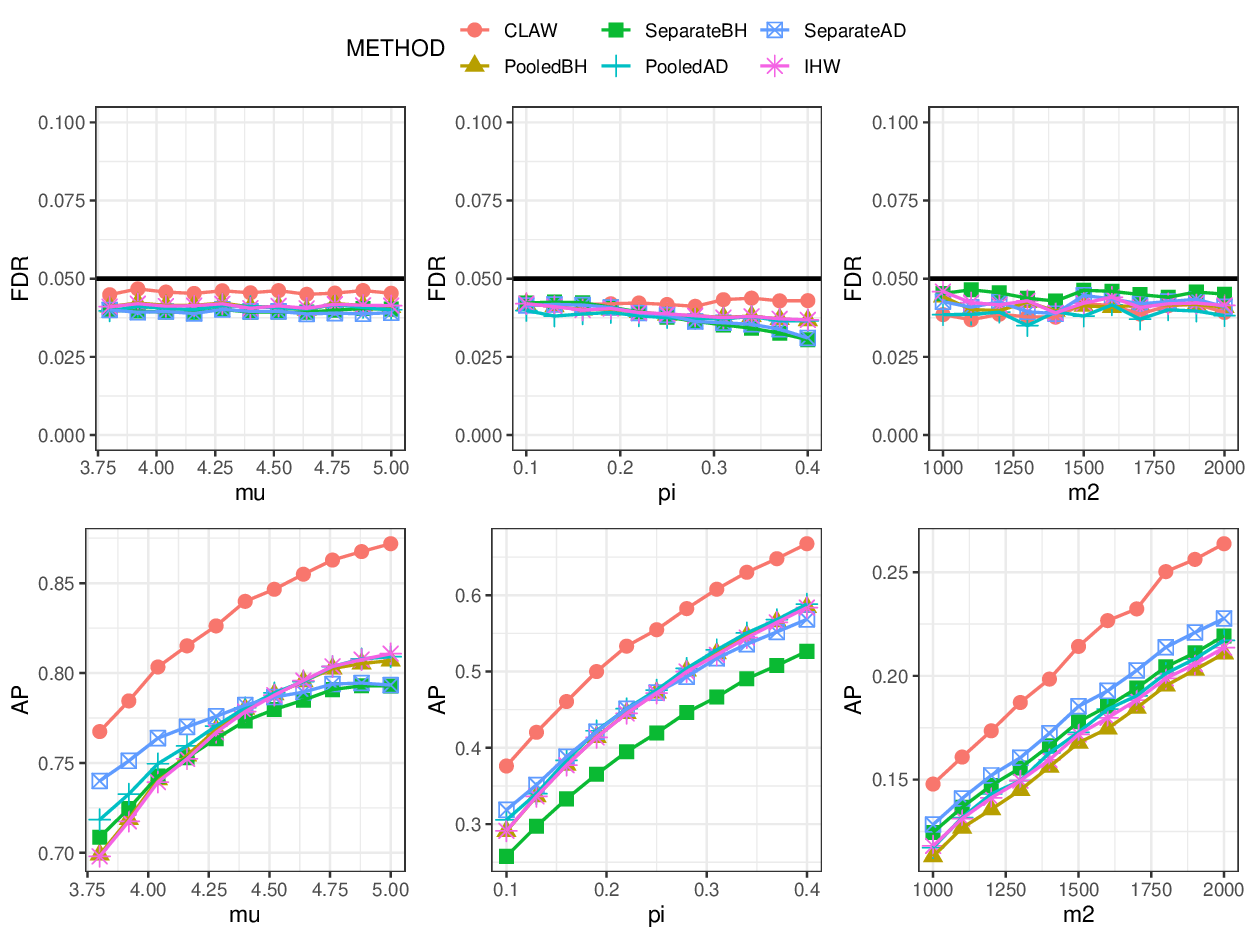}
    \caption{FDR and AP comparison for grouped multiple testing at $\alpha=0.05$. The left, middle and right columns are corresponding to setting 1, 2 and 3, respectively.}\label{fig:group1}
\end{figure}

The following patterns regarding the strengths and limitations of different methods can be observed. Firstly, all methods control the FDR below the nominal level, with CLAW consistently displaying the highest AP across all settings. Secondly, SeparateAD (PooledAD) outperforms SeparateBH (PooledBH), which can be attributed to the  utilization of the test data in constructing more efficient scores. Thirdly, the comparison between the pooled and separate strategies does not yield a definitive conclusion. Specifically, the bottom left panel demonstrates that  SeparateAD outperforms PooledAD for small values of $\mu$, while the opposite is observed for large values. This discrepancy arises due to the combined impact of (a) the ranking within the groups and (b) the allocation of ``$\alpha$-wealth'' across the groups. Similar patterns are observed with the SeparateBH and PooledBH methods. These patterns highlight the superiority of CLAW, as it constructs effective scores via locally adaptive weighting, which addresses both ranking and $\alpha$-wealth allocation issues within a unified framework.


\subsection{Multiple testing with ordinal covariates}
\label{subsec:spatial}


Consider a scenario where the primary statistics $T_i$ are observed along an ordered sequence. We utilize the natural order in the sequence $S_i=i$ as the covariate. The case where $S_i$ encodes the location of higher-dimensional spatial regions is considered in Section \ref{subsec:spatial-2D} of the supplement.   

The data are generated following a hierarchical model specified as follows: 
\begin{equation*}
\PP(\theta_{i}=1|S_{i}=s)=\pi_{s}, \quad  T_{i}|(\theta_{i},S_{i}=s) \overset{\text{ind.}}{\sim} (1-\theta_{i})\mathcal{N}(0,1)+\theta_{i}F_{1s}, 
\end{equation*}
where $i=1, \cdots, 3000$. We explore the following three settings in which the data exhibit a ``smoothness pattern’’, characterized by the similarity in values between $\pi_i$ and $\pi_j$, as well as between $F_{1i}$ and $F_{1j}$, when $i$ and $j$ are in close proximity.
\begin{enumerate}[1. ]
    \item  $F_{1s}\equiv F_{1}=\mathcal{N}(\mu,1)$; $\pi_{s}=0.6$ for $s\in[201,350]\cup[1501,1650]$, $\pi_{s}=0.3$ for $s\in[801,1000]\cup[2101,2300]$, and $\pi_{s}=0.02$ otherwise. 
    \item  $F_{1s}=\mathcal{N}(-2.5,1)$ if $s\leq1500$ and $F_{1s}=\mathcal{N}(3.6,1.5^2)$ otherwise; $\pi_{s}=2\pi$ for $s\in[201,350]\cup[1501,1650]$, $\pi_{s}=\pi$ for $s\in[801,1000]\cup[2101,2300]$, and $\pi_{s}=0.02$ otherwise. 
    \item  $F_{1s}=\mathcal{N}(\mu+0.15\sin(0.6s),1)$; $\pi_{s}=0.4(1+\sin(0.02s))$ for $s\in[201,500]\cup[801,1100]\cup[1501,1800]\cup[2101,2400]$ and $\pi_{s}=0.02$ otherwise. 
\end{enumerate}

One possible idea for testing hypotheses along an ordered sequence is to partition the sequence into multiple groups. However, this approach results in significant loss of information, and determining the optimal number of groups and optimizing the grouping procedure remain a complicated issue. Therefore, we exclude IHW, SeparateBH and SeparateAD from comparison as they both involve a potentially intractable grouping step. Instead, we expand the comparison to include several other well-established methods for covariate-assisted multiple testing, namely AdaPT \citep{lei18adapt}, SABHA \citep{SABHA19} and LAWS \citep{LAWS22}. In our simulation studies, AdaPT is implemented using the R package \texttt{adaptMT}. LAWS and SABHA are implemented using estimated proportions given by \eqref{conventional_est}, while the proportions in CLAW are estimated using \eqref{pi-ds}. These estimators employ the same weights $w_{ij}=\phi(|S_{i}-S_{j}|/150)$, where $\phi(x)$ represents the density function of a standard Gaussian variable. Both CLAW and AdaDetect utilize the same calibration data, generated as $\Tilde{T}_{i}\overset{i.i.d.}{\sim}\mathcal{N}(0,1)$ for $i\in[3000]$. 

\begin{figure}[!htbp]
    \centering
    \includegraphics[width=0.9\linewidth, height=4in]{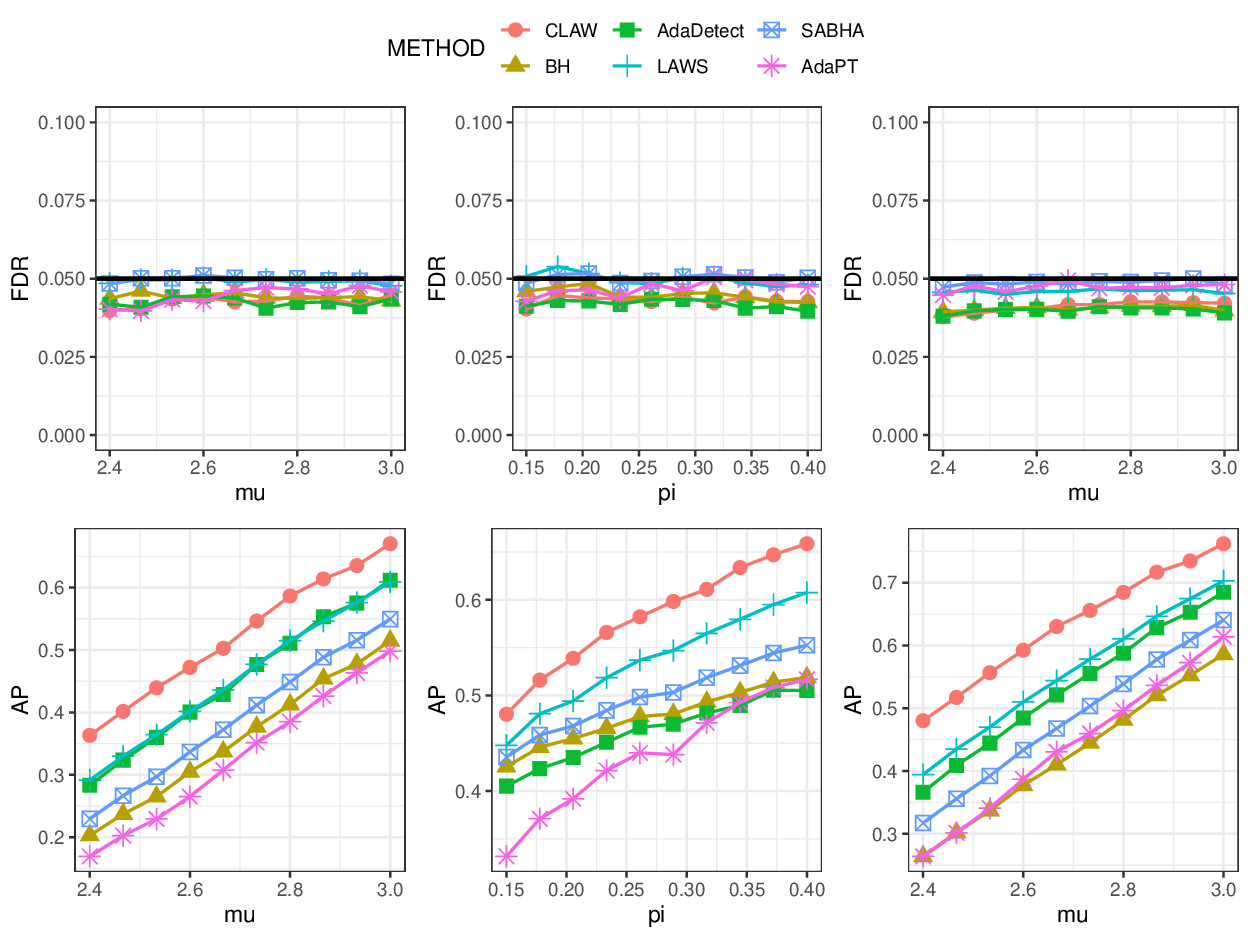}
    \caption{FDR and AP comparison for multiple testing for ordered sequences at $\alpha=0.05$. The left, middle and right columns correspond to settings 1-3, respectively.}\label{fig:spatial_1D}
\end{figure}

The simulation results from Settings 1 to 3 are summarized in corresponding columns of Figure \ref{fig:spatial_1D}, revealing several important patterns. Firstly, all methods control the FDR reasonably well. However, the FDR levels of LAWS and SABHA sometimes exceed the nominal level of $\alpha=0.05$; this is consistent with the theory as both methods only offer asymptotic control of the FDR. Secondly, the disparity in power between methods without side information, such as AdaDetect and BH, and methods that adapt to the side information, such as CLAW, AdaPT, LAWS and SABHA, becomes more pronounced as the side information becomes more informative, as can be seen in the bottom middle panel. Lastly, CLAW is superior in comparison to LAWS and SABHA in terms of both finite sample validity and higher statistical power. The power gain stems from CLAW's capability of incorporating the structural information encoded in both $\pi_s$ and $f_{1s}$. By contrast, LAWS and SABHA merely leverage the sparsity structure captured by $\pi_s$. Section \ref{app:simu} of the Supplement provides further numerical illustration of the effectiveness of CLAW in a broad range of settings where existing methods may fail to control the FDR.

\section{Experiments with Real Data}
\label{sec:application}
\label{subsec:mnist}


\subsection{Application to MNIST dataset}


We consider the novelty detection task based on the MNIST dataset \citep{lecun2010mnist}, a benchmark dataset consisting of handwritten digit images widely used for evaluating and comparing various image classification algorithms. It consists of 70,000 grayscale images, each representing a handwritten digit from 0 to 9. The images are formatted as 28$\times$28 pixel matrices, with each pixel intensity value ranging from 0 (white) to 255 (black). In our experiment, we design two settings with grouped images: the images labeled with the digit ``0'' are regarded as the inliers, while the images labeled with other digits correspond to outliers. 
\begin{enumerate}[1. ] 
    \item $\mathcal D^{\rm test1}$: Group 1 has 980 ``0''s and 120 ``6''s; Group 2 has 1500 ``0''s and 500 ``9''s.
    \item $\mathcal D^{\rm test2}$: Group 1 has 1080 ``0''s and 120 ``8''s, and Group 2 has 1500 ``0''s and 500 ``6''s.
\end{enumerate}
In above datasets, the covariate $S_i\in\{1,2\}$ represents the group memberships. Due to the high-dimensionality of the image data, accurate estimation of the working model \eqref{model:mixture} is challenging, where $\pi_{S_i}$ and $F_{1}(\cdot|S_i)$ exhibit variations across the different groups.  Additionally, precise knowledge of the null distribution is lacking. Instead, after sampling the test data from the MNIST dataset, the remaining instances of the digit ``0'' are gathered to form a dataset comprising null samples, which is denoted as $\mathbf{T}^0$. 

We apply semi-supervised CLAW, PooledAD, and SeparateAD (described in Section \ref{app:subsub-pu-group}) for detecting outliers with FDR control. Two strategies have been employed to implement these methods. The first strategy involves constructing scores using kernel density (KD) estimation methods, where the densities of both the training data and the mixture of test data and calibration data are estimated separately. The density ratio is then calculated by dividing these estimated densities. The second strategy involves the direct estimation of the density ratio using random forests (RF), as done in \citet{marandon22mlfdr}. We split the null dataset $\mathbf{T}^0$ into the calibration data $\Tilde{\mathbf{T}}$, which has the same size as the test data, and the training data $\mathbf{T}^{tr}$. In all the methods, the same calibration dataset has been utilized. 

Table \ref{table1} of the Supplement provides a summary of the experimental results obtained from Settings 1 and 2, with the FDR level set at $\alpha=0.05$. The numbers of discoveries and true discoveries (in parentheses) are presented for each method. Upon direct calculations, we can see that all methods control the FDR below the nominal level. It is evident that the KD-based methods fail to yield any discoveries. This can be attributed to the limitations of kernel methods in handling high-dimensional scenarios. In contrast, the scores generated through RF prove to be effective. Employing the RF-based scores, PooledAD outperforms SeparateAD in Setting 1, while underperforming SeparateAD in Setting 2. Furthermore, CLAW dominates AdaDetect, including both PooledAD and SeparateAD, in both Settings 1 and 2.

\subsection{Application to proteomics data}

We illustrate the application of the CLAW procedure using a proteomics dataset previously analyzed by \citet{lei18adapt} and \citet{Ignatiadis21IHW} with the AdaPT and IHW methods, respectively. The dataset, collected by \citet{dephoure2012hyper}, comprises temporal abundance profiles for 2,666 yeast proteins obtained from a quantitative mass spectrometry experiment conducted under two treatment conditions: rapamycin and dimethyl sulfoxide (DMSO). In prior analyses, p-values were first computed using Welch’s t-test, followed by the application of multiple testing procedures to identify proteins exhibiting differential abundance in yeast cells between the two conditions. Following the strategy outlined in \citet{lei18adapt}, we incorporate the logarithm of the total number of peptides as side information (covariate). The inclusion of this covariate has been shown to provide critical insights that enhance both the power and interpretability of FDR analyses. 

In the analyses conducted by \citet{Ignatiadis21IHW} and \citet{lei18adapt}, the following assumptions were made: (a) the null p-values are independent and follow a uniform distribution \(U(0,1)\); and (b) the null p-values are independent of the covariates. These assumptions appear to be well-suited for the proteomics dataset, ensuring that the joint exchangeability conditions specified in \eqref{jointexch-covariate} hold, which also implies that the condition in \eqref{jointexch} is satisfied. Consequently, we have implemented AdaDetect (which requires \eqref{jointexch}) and CLAW (which requires \eqref{jointexch-covariate}) for this analysis. The calibration data for both conformal methods are obtained by drawing iid samples from \(U(0, 1)\). Given that both the test statistics and covariates are continuous, we applied CLAW with the augmentation strategy described in Appendix \ref{subsub-pu-aug}. We compare the number of discoveries across different methods, including CLAW, BH, AdaDetect, LAWS, SABHA, AdaPT, and IHW, under a grid of nominal FDR levels \(\alpha = 0.045, 0.05, 0.055, 0.06\). 

The results are summarized in Figure \ref{fig:yeast} of Appendix \ref{app:yeast}. We can see that significant power enhancement can be achieved through the employment of methods such as AdaPT and CLAW, which effectively incorporate side information. Notably, by constructing conformity scores to emulate the Clfdr statistic, CLAW exhibits the greatest power among all methods considered. These observed patterns are consistent with our findings from the simulation studies.

\section{Discussion}\label{sec:discuss}

We conclude this article by highlighting several open issues for future exploration. First, Algorithm \ref{algo:claw} requires a minimum of \( m \) null samples for effective implementation. However, challenges may arise when null samples are limited or unavailable. Addressing these challenges may necessitate investigating how the test data can be leveraged to estimate the empirical null distribution, as advocated by \citet{efron04null}. Secondly, CLAW employs a mirror process that is restricted to testing sharp null hypotheses. A significant area for future research involves designing new FDP processes that utilize more efficient conformity scores to effectively handle composite nulls. Thirdly, Algorithm \ref{algo:claw_clfdr} employs a working model inspired by the NEB framework, which demonstrates effectiveness in independent settings with low-dimensional covariates. Nonetheless, when faced with complex data-generating processes that involve dependencies and higher-dimensional covariates, it becomes crucial to explore alternative methodologies for constructing powerful score functions. Fourthly, as low statistical power can indicate the trustworthiness of a predictive model, it is of great interest to enhance the conformal framework to incorporate power analysis. Particularly, in critical application areas such as aviation, medical screening, and cybersecurity, the risks associated with false negatives can greatly exceed those of false positives. Consequently, statistical guarantees on power or the missed discovery rate, as discussed in \cite{abraham2023sharp}, provide an important direction for future research. Fifthly, Algorithm \ref{algo:derandom} offers a flexible framework for integrating auxiliary data from diverse sources by leveraging the properties of generalized e-values. An intriguing avenue for future research involves the dynamic allocation of \(\alpha_k\) to prioritize the most informative covariates from specific sources. Finally, there is significant interest in extending the CLAW framework to tackle closely related problems. Key tasks include developing innovative methods for online testing, selective classification, and set or interval prediction in scenarios where side information is available.

\section*{Acknowledgement}

We are grateful to the Associate Editor and the referees for their constructive feedback, which has greatly enhanced the clarity, presentation, and theoretical framework of our manuscript. Special appreciation also goes to Matteo Sesia and Asaf Weinstein for their valuable suggestions.


\setlength{\bibsep}{0pt}

\bibliographystyle{chicago}

\bibliography{reference}

\renewcommand{\appendixname}{Appendix~\Alph{section}}

\setcounter{equation}{0}
\renewcommand{\theequation}{\thesection.\arabic{equation}}
\counterwithin{equation}{section}
\counterwithin{figure}{section} 
\counterwithin{table}{section}

\appendix


\begin{center}\bf\Large
	Appendix
\end{center}

The supplement provides further details on methodological developments (Section \ref{app-sec:method-develop}), proofs of the primary theory (Section \ref{app:proof}), proofs of auxiliary theories (Section \ref{app:proof-4-app}), connections to existing work (Section \ref{app:discuss}), and supplementary numerical results (Section \ref{app:simu}). 

\section{Details in Methodological Developments}
\label{app-sec:method-develop}

\subsection{Derivation of conformalized estimators}
\label{app:estimate}

\subsubsection{The density estimator} \label{app:subsub-density}

We begin by discussing the process of ``conformalizing’’ the estimator of $f_{S_i}(t)$, so that it processes the swapping-invariance property with respect to $T_{j}$ and $\Tilde{T}_{j}$. In equation \eqref{conventional_est}, the estimator $\hat{f}_{S_i}^{*}(t)$ can be expressed as:
\begin{equation*}
   \hat{f}_{S_i}^{*}(t) = \sum_{j=1}^{m} \frac{w_{ij}}{\sum_{k=1}^{m}w_{ik}}K_{h}(t-T_{j}),
\end{equation*}
which represents a weighted sum of $K_{h}(t-T_{j})$. In order to create equality between the roles of $T_{j}$ and $\Tilde{T}_{j}$ in the estimator, our proposal involves introducing an additional term $K_{h}(t-\Tilde{T}_{j})$ with an equal weight to $K_{h}(t-T_{j})$. Consequently, the modified estimator is formulated as:
\begin{equation}
    \hat{f}_{S_{i}}^{**}(t)=\frac{\sum_{j=1}^{m}w_{ij}[K_{h}(t-T_{j})+K_{h}(t-\tilde{T}_{j})]}{2\sum_{j=1}^{m}w_{ij}}.\label{f**}
\end{equation}
Note that to ensure that $\hat{f}_{S_i}^{**}(t)$ remains a proper density function with an integral value of 1, we have introduced the multiplier of 2 to the denominator in  \eqref{f**}.

In practical applications, the kernel bandwidth $h$ is a critical parameter that should be either determined prior to fitting the density function, or be chosen as a data-driven quantity in a principled way. When selecting a data-driven bandwidth $h$, we recommend employing well-established techniques, such as Silverman's rule \citep{silverman18density},  Sheather-Jones method \citep{sheather91}, or  Lepski’s method \citep{gl11bandwidth} for the combined dataset \((\mathbf{T},\Tilde{\mathbf{T}})\). This ensures that $h\left((\mathbf{T},\Tilde{\mathbf{T}})_{\Pi}\right) = h\left(\mathbf{T},\Tilde{\mathbf{T}} \right)$ holds for any permutation $\Pi$ of the elements in the vector $(\mathbf{T},\Tilde{\mathbf{T}}) = (T_1, \cdots, T_m, \Tilde T_1, \cdots, \Tilde T_m)$. The permutation invariance property guarantees that the estimator \eqref{f**} is swapping invariant, and consequently the resulting conformity scores $(u_i, \Tilde u_i)$ satisfy the exchangeability condition \eqref{pwexch}.


\subsubsection{The proportion estimator}\label{app:subsub-proportion}

Before presenting the conformalized proportion estimator, we first explain the basic steps involved in deriving the estimator $\hat{\pi}^{*}_{S_i}$ \eqref{conventional_est}. Consider the quantity $m_{i}=\sum_{j=1}^{m}w_{ij}$, which represents the cumulative ``mass’’ (or ``effective number of observations’’) within the vicinity of unit $i$. For instance, in the simple case of grouped multiple testing where $w_{ij}=\II\{S_{i}=S_{j}\}$, $m_i$ denotes the size of the group to which $T_{i}$ belongs. In the more complicated case where $S_i$ is continuous, we leverage the local structure to compute $m_{i}$ by borrowing strength from points close to $S_{i}$ while assigning lesser weight to distant points. To provide intuitions for deriving the local adaptive estimator, we assume that the null p-values are  uniformly distributed on $[0,1]$.

Suppose our objective is to determine the number of null p-values that exceed the threshold $\lambda$. A conservative estimate of the empirical count  $\sum_{j\in\mathcal{H}_{0}}w_{ij}\II \{p(T_{j})>\lambda\}$ is given by
\begin{equation}\label{em:counts}
 \sum_{j=1}^{m}w_{ij}\II \{p(T_{j})>\lambda\}.
\end{equation}
The quantity provides a reasonably good approximation when $\lambda$ is large, as we expect that most non-null p-values will be relatively small. On the other hand, the expected count is given by:
\begin{equation}\label{ex:counts}
\EE \Big[\sum_{_{j\in\mathcal{H}_{0}}}w_{ij}\II \{p(T_{j})>\lambda\} \Big|\mathbf{S} \Big] = (1-\pi_{S_{i}})(1-\lambda)\sum_{j=1}^{m}w_{ij}.
\end{equation}
Consequently, we can recover the estimator for the non-null proportion presented in \eqref{conventional_est}: 
$$
\hat{\pi}^{*}_{S_{i}}= 1-\frac{\sum_{j=1}^{m}w_{ij}\II \{p(T_{j})>\lambda\}} {(1-\lambda)\sum_{j=1}^{m}w_{ij}}.
$$ 

To conformalize $\hat{\pi}^{**}_{S_{i}}$, we mix the calibration data and test data when computing the empirical counts \eqref{em:counts}, giving rise to $\sum_{j=1}^{m}w_{ij}[\II\{p(T_{j})>\lambda\}+\II\{p(\Tilde{T}_{j})>\lambda\}]$. This guarantees that the resulting estimator maintains swapping-invariance. Correspondingly, the expected counts \eqref{ex:counts} will be adjusted by a factor of 2. Setting the expected counts and empirical counts equal, our proposed estimator is given by:
\begin{equation}\label{pii-tilde} 
\Tilde{\pi}^{**}_{S_{i}}=1-\frac{\sum_{j=1}^{m}w_{ij}[\II\{p(T_{j})>\lambda\}+\II\{p(\Tilde{T}_{j})>\lambda\}]}{2(1-\lambda)\sum_{j=1}^{m}w_{ij}}. 
\end{equation}
Here, $p(\Tilde{T}_{j})$ represents the p-value of $\Tilde{T}_{j}$ calculated in the same manner as $p(T_{j})$. The estimator \eqref{pii-tilde} is subsequently adjusted using \eqref{pii-hat} to guarantee that its value remains within the feasible range of $[0,1/2]$.

\subsection{Semi-supervised CLAW via PU learning}
\label{subsec:ssmt}

This section extends the CLAW procedure to the semi-supervised multiple testing scenario. We propose a class of novel conformity scores, constructed through carefully designed PU learning algorithms, that satisfy the pairwise exchangeability property \eqref{pwexch}. Achieving this involves relaxing the exchangeability notion (Section \ref{app:pwexch}), modifying estimators for the local sparsity level (Section \ref{app:subsub-confp}), and devising new strategies for estimating density ratios (Sections \ref{sub:grouping} and \ref{subsub-pu-aug}).

\subsubsection{Pairwise exchangeability between samples}
\label{app:pwexch}

We begin the discussion by relaxing the joint exchangeability condition \eqref{jointexch-covariate} to a pairwise exchangeability between the data points:

Suppose we have labeled samples \(\mathbf{T}^0 = (T^0_i : i \in \mathcal{D}_0)\). Consider the partitioning \(\mathcal{D}_0 = \mathcal{D}^{tr} \cup \mathcal{D}^{cal}\), with \(\mathcal{D}^{tr} \cap \mathcal{D}^{cal} = \emptyset\). Let \(\mathbf{T}^{tr} = (T^0_i : i \in \mathcal{D}^{tr})\) and \(\tilde{\mathbf{T}} = (\tilde{T}_i : i \in [m]) := (T_i^0 : i \in \mathcal{D}^{cal})\) denote the training and calibration datasets. The pairwise exchangeability condition is given by:
\begin{equation}\label{data_pwexch}
\left( (\mathbf{T}, \tilde{\mathbf{T}})_{\mathrm{swap}(\mathcal{J})} \big| \mathbf{T}^{tr}, \mathbf{S} \right) \overset{d}{=} \left( \mathbf{T}, \tilde{\mathbf{T}} \big| \mathbf{T}^{tr}, \mathbf{S} \right), \quad \forall \mathcal{J} \subset \mathcal{H}_{0}.
\end{equation}
We highlight important distinctions between \eqref{jointexch-covariate} and \eqref{data_pwexch}, along with their implications:

\begin{enumerate}
\item Assumption \eqref{data_pwexch} allows the null distribution to depend on the covariates. 
Hence the data generation process can be represented as follows:
\begin{eqnarray*}
S_j  \sim  G(\cdot), \quad 
(\theta_{j}|S_{j}=s)  \sim  \mbox{Bernoulli}(\pi_s), \quad
(T_j|S_j, \theta_j)  \sim  (1-\theta_j) F_0(\cdot|S_j) + \theta_j F_1(\cdot|S_j). 
\end{eqnarray*}
Likewise, the calibration data is allowed to be generated as $\Tilde{T}_{j}\sim F_0(\cdot|S_j)$ for $j\in[m]$. 
This flexibility facilitates the modeling of complex correlation structures between calibration and test data through their relationships with the covariates, as illustrated in Examples 2 and 3 of this section.

\item Assumption \eqref{jointexch-covariate} is stronger than assumption \eqref{data_pwexch}, as the swapping-invariant property directly follows from the permutation-invariant property. Moreover, Assumption \eqref{data_pwexch} imposes no constraints on the dependency structure of \(\mathbf{T}\), thereby significantly relaxing the requirement for the equal correlation structure among the null samples as dictated by assumption \eqref{jointexch-covariate}. Example 4 in this section demonstrates that this flexibility allows for the accommodation of some complex dependence structures. 

\item Assumption \eqref{data_pwexch} eliminates the requirement for the training data \(\mathbf{T}^{tr}\) to be exchangeable with the null samples in the calibration and test sets [\(\tilde{\mathbf{T}}\) and \((T_i : i \in \mathcal{H}_{0})\)]. This flexibility allows for the use of integrative and transfer learning algorithms to leverage labeled outliers or external data from related source domains (as explored by \citealp{liang22integrative}), facilitating the development of more powerful predictive models. 
\end{enumerate}

The next theorem, delineating principles for constructing conformity scores within the semi-supervised framework, extends Theorem \ref{thm:exch} under the less stringent condition \eqref{data_pwexch}.

\begin{thm} \label{appthm:exch}
Consider a class of score functions in the form of $g(\cdot,S_{i})=g(\cdot,S_{i};(\mathbf{T},\Tilde{\mathbf{T}}),\mathbf{T}^{tr},\mathbf{S})$ for $i\in[m]$. Define $u_{i}=g(T_{i},S_{i})$ and $\Tilde{u}_{i}=g(\Tilde{T}_{i},S_{i})$. Let $\mathbf{U}=(u_{1},\cdots,u_{m})$ and $\Tilde{\mathbf{U}}=(\Tilde{u}_{1},\cdots,\Tilde{u}_{m})$. 
    \begin{enumerate}[(a)]
        \item $\mathbf{U}$ and $\Tilde{\mathbf{U}}$ satisfy the pairwise exchangeability \eqref{pwexch} if (i) the score functions are swapping invariant, i.e.
       \begin{equation}\label{app:principle-pw}
        \mbox{$g\left(\cdot,S_{i};(\mathbf{T},\Tilde{\mathbf{T}})_{\mathrm{swap}(\mathcal{J})},\mathbf{T}^{tr},\mathbf{S}\right) = g\left(\cdot,S_{i};(\mathbf{T},\Tilde{\mathbf{T}}),\mathbf{T}^{tr},\mathbf{S}\right)$ for any $\mathcal{J}\subset[m]$ }; 
        \end{equation} 
      and  (ii) $\mathbf{T}$, $\Tilde{\mathbf{T}}$, $\mathbf{T}^{tr}$ and $\mathbf{S}$ satisfy the  pairwise exchangeability condition \eqref{data_pwexch};
        
        \item $\mathbf{U}$ and $\Tilde{\mathbf{U}}$ satisfy the joint exchangeability \eqref{jointexch-scores}, if (i) the score functions are permutation-invariant, i.e.
        \begin{equation}\label{app:principle-jt}
     \mbox{  $g\left(\cdot,S_{i};(\mathbf{T},\Tilde{\mathbf{T}}),\mathbf{T}^{tr},\mathbf{S}\right)= g\left(\cdot;(\mathbf{T},\Tilde{\mathbf{T}}),\mathbf{T}^{tr}\right)=g\left(\cdot;(\mathbf{T},\Tilde{\mathbf{T}})_{\Pi},\mathbf{T}^{tr}\right)$ }; 
        \end{equation}
     and (ii)  $\mathbf{T}$  and $\mathbf{T}^0=\Tilde{\mathbf{T}}\cup\mathbf{T}^{tr}$ satisfy the joint exchangeability condition \eqref{jointexch}.
    \end{enumerate}
\end{thm}

The remaining part of this subsection provides examples that demonstrate how assumptions \eqref{jointexch-covariate} and \eqref{data_pwexch} can accommodate a diverse array of covariate types. Specifically, covariates can be random (Examples 1 \& 3) or non-random (Examples 2 \& 4), and they can be continuous (Examples 1 \& 3) or discrete (Example 2). Furthermore, we demonstrate that assumption \eqref{data_pwexch} allows for complex correlation structures, both between the null samples and the covariates (Examples 2 \& 3) and among the null samples themselves (Example 4).

\begin{example}[Two-sample sparse inference; \citealp{CARS}] \rm{In contrast to existing works that derive side information from external sources, this example demonstrates the application of the CLAW framework for handling covariates constructed from the dataset at hand. 

\medskip

Let \(\{X_{ij}: 1 \leq j \leq n_{x}\}\) and \(\{Y_{ij}: 1 \leq j \leq n_{y}\}\) denote independent copies of \(X_i \sim \mathcal{N}(\mu_{xi}, \sigma_{xi}^{2})\) and \(Y_i \sim \mathcal{N}(\mu_{yi}, \sigma_{yi}^{2})\),  \(i \in [m]\). Consider the following two-sample multiple testing problem:
\[
H_{i,0}: \mu_{xi} = \mu_{yi} \quad \text{versus} \quad H_{i,1}: \mu_{xi} \neq \mu_{yi}, \quad i \in [m].
\]
Define \(n = n_{x} + n_{y}\), \(\gamma_{x} = \frac{n_{x}}{n}\), \(\gamma_{y} = \frac{n_{y}}{n}\), \(\Bar{X}_{i} = \frac{1}{n_{x}} \sum_{j=1}^{n_x} X_{ij}\), and \(\Bar{Y}_{i} = \frac{1}{n_{y}} \sum_{j=1}^{n_y} Y_{ij}\). The following statistics can be constructed to summarize the information in the data:
\begin{equation}\label{T-S}
(T_{i}, S_{i}) = \sqrt{\frac{n_{x} n_{y}}{n}} \left( \frac{\Bar{X}_{i} - \Bar{Y}_{i}}{\sigma_{pi}}, \frac{\Bar{X}_{i} + \kappa_{i} \Bar{Y}_{i}}{\sqrt{\kappa_{i}} \sigma_{pi}} \right), \quad i \in [m],
\end{equation}
where \(\sigma_{pi}^{2} = \gamma_{yi} \sigma_{xi}^{2} + \gamma_{xi} \sigma_{yi}^{2}\) and \(\kappa_{i} = \frac{\gamma_{yi} \sigma_{xi}^{2}}{\gamma_{xi} \sigma_{yi}^{2}}\). Unlike traditional  methods that rely solely on the primary statistics \((T_i: i \in [m])\), the CARS procedure \citep{CARS} proposes to integrate auxiliary covariates \((S_i: i \in [m])\) as side information into the inferential process to enhance statistical power.

According to the construction, \(T_{i}\) and \(S_{i}\) are independent (given that they are uncorrelated Gaussian variables). Moreover, the pairs \((T_{i}, S_{i})\) are mutually independent across the \(m\) units. Since \(T_{i} | H_{i,0} \sim \mathcal{N}(0,1)\), we can independently draw data points from \(\mathcal{N}(0,1)\) to create a calibration dataset \(\Tilde{\mathbf{T}} = (\Tilde{T}_{i}: i \in [m])\). It follows that the test data \(\mathbf{T} = (T_{i}: i \in [m])\), the calibration data \(\Tilde{\mathbf{T}} = (\Tilde{T}_{i}: i \in [m])\), and the auxiliary covariates \(\mathbf{S} = (S_{i}: i \in [m])\) satisfy  the joint exchangeability condition \eqref{jointexch-covariate}. Consequently, we can implement the CLAW procedure utilizing the triplet \((\mathbf{T}, \Tilde{\mathbf{T}}, \mathbf{S})\), which can be regarded as a conformalized adaptation of CARS.

Lastly, we emphasize that both \citet{CARS} and our analysis are predicated on two idealized assumptions: (a) the variances \(\sigma_{xi}^{2}\) and \(\sigma_{yi}^{2}\) are known, and (b) both \(X_i\) and \(Y_i\) are Gaussian. 
In cases where the variances must be estimated or in instances where \(X_i\) and \(Y_i\) are non-Gaussian, \(T_{i}\) and \(S_{i}\) constructed via \eqref{T-S} would be correlated, leading to potential violations of both exchangeability conditions \eqref{jointexch-covariate} and \eqref{data_pwexch}. Therefore, the generalization of CLAW with finite-sample FDR theory, without relying on these idealized assumptions, presents an important avenue for future research. \qed
}

\end{example}

\begin{example}[Multi-class outlier detection]\label{example:multi-class} \rm{This example involves discrete covariates encoding side information regarding (nonrandom) group memberships. In this scenario, the triplet \((\mathbf{T}, \tilde{\mathbf{T}}, \mathbf{S})\) satisfies the pairwise exchangeability condition \eqref{data_pwexch} but not the joint exchangeability condition \eqref{jointexch-covariate}.

Our analysis focuses on a semi-supervised multiple testing framework in which test samples \(\mathbf{T}\) can be divided into \(K\) groups. Let \(\mathcal{D}^{test} = \bigcup_{i=1}^{K}\mathcal{D}_k\) represent the index set of all test samples, \((S_j \in [K]: j \in \mathcal{D}^{test})\) the set of covariates indicating group memberships, and \(\mathcal{D}_k = \{j \in \mathcal{D}^{test} : S_j = k\}\). Denote the test samples by \(\mathbf{T} = (\mathbf{T}_{1}, \ldots, \mathbf{T}_{K}) \coloneqq (T_1, \ldots, T_m)\), where \(\mathbf{T}_{k} = (T_j : j \in \mathcal{D}_k)\). In the above notation, \(\mathbf{S} = (\mathbf{S}_{1}, \ldots, \mathbf{S}_{K}) \coloneqq (S_1, \ldots, S_m)\) denotes the covariate sequence encoding grouping information, where \(\mathbf{S}_i\) is a vector of length \(|\mathcal{D}_i|\) with all elements equal to \(i\), \(i \in [K]\). Denote \(\mathcal{D}^0\) as the index set of all labeled null samples (inliers). The inliers corresponding to each group are given by \(\mathbf{T}^0_k = (T^0_i : i \in \mathcal{D}_{0,k})\), where \(\mathcal{D}_{0,k}\) denotes the index set of labeled samples from class \(k\), \(k \in [K]\).

Consider an outlier detection problem in medical image classification. Suppose we have collected brain images from a large cohort of healthy individuals (labeled null samples), and the objective is to identify abnormal images in new subjects. The covariate \(S\) may represent demographic characteristics such as gender or race. For example, brain images from healthy males and females can exhibit significant differences, suggesting that the exchangeability condition may only apply within distinct groups. 
Specifically, let \(\mathcal{H}_{0,k}\subset\mathcal D_k\) denote the index set of inliers from class \(k\) within the test data, i.e., $j\in\mathcal{H}_{0,k}$ if and only if $T_{j}$ is an inlier of class $k$, and $\mathcal{H}_0=\cup_{i=1}^{K}\mathcal{H}_{0,i}$. A fundamental and intuitive assumption underpinning our analysis is:
\begin{equation}\label{data_pwexch3}
(T_i, i \in \mathcal{H}_{0,k}; T_{j}^{0}, j \in \mathcal{D}_{0,k}) \text{ are exchangeable conditional on } \left(T_{i} : i \notin \mathcal{H}_{0,k}\right) \cup \left(T_{j}^{0} : j \notin \mathcal{D}_{0,k}\right).
\end{equation}

To implement CLAW, we partition the set \(\mathcal{D}_{0,k}\) into two subsets: a calibration set \(\mathcal{D}^{cal}_{k}\) of size $|\mathcal{D}^{cal}_{k}|=|\mathcal{D}_{k}|$, and a training set \(\mathcal{D}^{tr}_{k}\). Let \(\tilde{\mathbf{T}}_{k} = ({T}_i^0 : i \in \mathcal{D}^{cal}_{k})\) and the whole calibration dataset be $\tilde{\mathbf{T}}=(\tilde{\mathbf{T}}_{1},\cdots,\tilde{\mathbf{T}}_{K}):=(\tilde{T}_1,\cdots,\tilde{T}_m)$. The training dataset is defined as \(\mathbf{T}^{tr} = \{T_i^0 : i \in \bigcup_{k=1}^{K} \mathcal{D}^{tr}_{k}\}\).

Note that the inliers from different groups do not share the same (null) distribution; therefore, the triplet \((\mathbf{T}, \tilde{\mathbf{T}}, \mathbf{S})\) fails to satisfy the joint exchangeability condition \eqref{jointexch-covariate}. However, for every \(i \in \mathcal{H}_0\), since both \(T_i\) and \(\tilde{T}_i\) are inliers of class \(k = S_i\), it follows from condition \eqref{data_pwexch3} that 
$$(T_i, \tilde{T}_i, \mathbf{T}_{-i}, \tilde{\mathbf{T}}_{-i} | \mathbf{S}, \mathbf{T}^{tr}) \overset{d}{=} (\tilde{T}_i, T_i, \mathbf{T}_{-i}, \tilde{\mathbf{T}}_{-i} | \mathbf{S}, \mathbf{T}^{tr}).$$ Hence, the pairwise exchangeability assumption \eqref{data_pwexch} holds. Consequently, we can apply the CLAW procedure with \((\mathbf{T}, \tilde{\mathbf{T}}, \mathbf{S})\) for outlier detection, following the steps outlined in Section \ref{sec:relation} and Sections \ref{app:subsub-confp}-\ref{sub:grouping}. This modified version of CLAW effectively utilizes both labeled null samples and the structural information encoded in group memberships to enhance detection efficiency while maintaining effective control over the FDR in finite samples. \qed
}

\end{example}

\begin{example}[Multiple testing under heteroscedasticity; \citealp{HART}] \rm{This example involves continuous covariates that encode side information related to the heteroscedasticity present among the testing units. We illustrate how to construct a triplet \((\mathbf{T}, \tilde{\mathbf{T}}, \mathbf{S})\) from the raw observations that satisfies the pairwise exchangeability condition \eqref{data_pwexch} and can therefore be implemented within the CLAW framework.

Suppose that we collect \(n_i\) repeated measurements \((X_{ij}: j \in [n_i])\) from testing unit \(i\), where \(i \in [m]\). The observations \((X_{ij}: j \in [n_i])\) are independent across units and obey the following hierarchical model:  
\[
X_{ij} | \mu_i, \sigma_{i} \overset{ind}{\sim} F(\cdot | \mu_i, \sigma_i^2), \quad 
\mu_i | \sigma_{i} \overset{i.i.d.}{\sim} (1-\pi)\delta_{0} + \pi G(\cdot | \sigma_i), \quad \sigma_i \overset{i.i.d.}{\sim} V(\cdot), \quad j \in [n_i], \quad i \in [m],
\]
where \(F(\cdot | \mu, \sigma^2)\) represents a distribution with mean \(\mu\) and variance \(\sigma^2\); \(G(\cdot | \sigma)\) denotes an unspecified distribution parameterized by \(\sigma\); \(V(\cdot)\) refers to another unspecified distribution; and \(\pi = \mathbb{P}(\mu_i = 0)\) indicates the sparsity level. Moreover, the unobserved parameters \(\mu_i\) and \(\sigma_{i}^2\) are allowed to exhibit correlation. For each testing unit \(i \in [m]\), we assume the availability of a null dataset denoted by \(\{Y_{ij}: j \in [N_i]\}\), \(N_i \geq n_i\), which are independently drawn from \(F(\cdot | \mu_i = 0, \sigma_i^2)\). The objective is to simultaneously test \(m\) null hypotheses: \(H_i: ~\mu_i = 0\), for \(i \in [m]\).

The classical multiple testing frameworks, which utilize standardized statistics such as p-values or z-values, may result in information loss, as the heterogeneity in variances provides critical structural information. \citet{HART} demonstrated that a heteroscedasticity-adjusted ranking and thresholding (HART) procedure, which incorporates sample variances as side information, can effectively enhance the power of existing FDR methods. However, the asymptotic theory in \citet{HART} relies on Gaussian assumptions and consistent estimates of model parameters. Below, we outline the key steps for utilizing the CLAW framework to conformalize the HART procedure. The primary technical tool is Theorem 5 in Section A.2 of the Supplement, which provides guidelines for constructing test statistics and covariates from raw observations that satisfy the pairwise exchangeability condition \eqref{data_pwexch}.

Consider test statistics \(T_i = \frac{1}{n_i} \sum_{j=1}^{n_i} X_{ij}\), with calibration statistics \(\tilde{T}_i = \frac{1}{n_i} \sum_{j=1}^{n_i} Y_{ij}\) for \(i \in [m]\). The remaining null data points, denoted \(\mathbf{T}^{tr}_{i} = (Y_{ij}: j = n_{i}+1, \ldots, N_i)\), will be used as training data to estimate the unknown variances:
$$
S_i = \frac{1}{N_i - n_i} \sum_{j=n_i+1}^{N_i} \left(Y_{ij} - \frac{\sum_{j=n_i+1}^{N_i} Y_{ij}}{N_i - n_i + 1}\right)^2.
$$
Since \(\mathbf{T}^{tr}_{i}\) is independent of \((T_i, \tilde{T}_i)\), and \(S_i\) is measurable with respect to \(\mathbf{T}^{tr}_{i}\), we have
$$
(T_i, \tilde{T}_i | \mathbf{T}^{tr}_{i}, S_i) \overset{d}{=} (T_i, \tilde{T}_i) \overset{d}{=} (\tilde{T}_i, T_i) \overset{d}{=} (\tilde{T}_i, T_i | \mathbf{T}^{tr}_{i}, S_i), \quad i \in \mathcal{H}_0.
$$
If \(n_i = N_i\), there are no additional labeled samples for estimating the variance. In this case, we define
$$
S_i^* = \frac{1}{2n_i - 2} \sum_{j=1}^{n_i} \left[(X_{ij} - T_i)^2 + (Y_{ij} - \tilde{T}_i)^2\right].
$$

Next, we verify the pairwise exchangeability \((T_i, \tilde{T}_i | S_i^*) \overset{d}{=} (\tilde{T}_i, T_i | S_i^*)\) for \(i \in \mathcal{H}_0\). Let \(A = \sum_{j=1}^{n_i} (X_{ij} - T_i)^2\) and \(B = \sum_{j=1}^{n_i} (Y_{ij} - \tilde{T}_i)^2\). By construction, for \(i \in \mathcal{H}_0\), we have 
$(T_i, A) \overset{d}{=} (\tilde{T}_i, B)$. Moreover, $(T_i,A)$ is independent of $(\tilde{T}_i, B)$. Thus, 
$$
\left(T_i,\tilde{T}_i, A \vee B, A \wedge B\right) \overset{d}{=} \left(\tilde{T}_i, T_i, A \vee B, A \wedge B\right).
$$
The pairwise exchangeability follows since \(S_i^*\) is a symmetric function of \(\{A, B\} = \{A \vee B, A \wedge B\}\). Finally, as the raw data from different test units are mutually independent, the assumption \eqref{data_pwexch} holds.

This example highlights the effectiveness of the CLAW framework in three key aspects. First, the conformalized HART procedure diverges from conventional FDR methods by eliminating the Gaussian assumption. This relaxation broadens the applicability of the method across a wider range of data distributions. Second, the extension of FDR validity to finite samples enhances the asymptotic theory in existing works. Finally, the shift from joint exchangeability to pairwise exchangeability further increases the applicability of the CLAW framework. Specifically, while the correlation of \(S_i^*\) with both \(T_i\) and \(\tilde{T}_i\) presents challenges for existing conformal methods that depend on joint exchangeability assumptions, the CLAW framework is well-equipped to address these complexities effectively. 
\qed
}
\end{example}

\begin{example}[Correlated and non-exchangeable null test samples]

\rm{This example demonstrates the capability of the relaxed assumption \eqref{data_pwexch} to effectively address complex correlation structures that conventional FDR methods may struggle to accommodate. We will begin by outlining the background context from which the problems of interest may arise and subsequently discuss the implementation of the CLAW framework as a solution to the problem at hand.

In various signal processing applications, such as wireless sensor networks, communication systems, and biomedical monitoring, the challenge of outlier detection arises when multiple receivers are employed to capture signals from a common source. Consider a scenario where a single source produces an original signal \((y_{i}: i \in [m])\), which is received by two different signal receivers. Under normal conditions, both receivers accurately record the true signal along with inherent noise, which obeys a specified distribution \(F_{\epsilon}\). However, discrepancies may occur when one of the receivers becomes faulty. For instance, the first receiver, which operates correctly, outputs the reliable records (calibration samples) \(\tilde{\mathbf{T}}=(\tilde{T}_{i}: i \in [m])\). In contrast, the second receiver, experiencing malfunction or contamination, outputs records \(\mathbf{T}=(T_{i}: i \in [m])\) (test samples) that deviate significantly from the expected values, following a different distribution. In this framework, practitioners can leverage the trustworthy data from the properly functioning receiver \((\tilde{T}_{i}: i \in [m])\) to calibrate and identify specific time points at which the other receiver exhibits outlier behavior. 

Let \((y_{i}: i \in [m])\) symbolize an underlying stochastic process. For illustrative purposes, we can consider \((y_{i}: i \in [m])\) as a stationary AR(1) process, where \(\mathrm{cor}(y_{i}, y_{j}) = \rho^{|i-j|}\), \(\rho\in[-1, 1]\), and \(y_i\) obeys a marginal distribution \(F_{y}\). Assume that \(\mathbf{T}\) and \(\tilde{\mathbf{T}}\) obey the following model:
\[
T_{i}|(\theta_{i} = 0, S_{i} = s) = y_{i} + \epsilon_{i}, \quad T_{i}|(\theta_{i} = 1, S_{i} = s) \sim F_{1s}, \quad \tilde{T}_{i} = y_{i} + \epsilon_{m+i},
\]
where \(\PP(\theta_{i} = 1 | S_{i} = s) = \pi_{s}\) and \(\{\epsilon_{i}: i \in [2m]\}\) represent i.i.d. noises following distribution \(F_{\epsilon}\). The non-null observations \((T_{i}: i \notin \mathcal{H}_{0})\), which are sampled from \(F_{1s}\) conditioned on \(S_{i}=s\), are assumed to be independent of the null samples \((T_{i}: i \in \mathcal{H}_{0}) \cup (\tilde{T}_{i}: i \in [m])\). As anomalies tend to appear in clusters, we adopt the sequential order as side information, i.e., \(S_i = i\) for \(i \in [m]\).

The problem of detecting abnormal signals can be framed within the multiple testing framework: \(H_i: \theta_i = 0, i \in [m]\). To implement CLAW, we need to verify the pairwise exchangeability. Observe that for every \(i \in \mathcal{H}_0\),
\[
\left( T_i, \tilde{T}_i, \mathbf{T}_{-i}, \tilde{\mathbf{T}}_{-i} \,|\, (y_{j}: j \in [m]) ,\mathbf{S} \right) \overset{d}{=} \left( \tilde{T}_i, T_i, \mathbf{T}_{-i}, \tilde{\mathbf{T}}_{-i} \,|\, (y_{j}: j \in [m]),\mathbf{S} \right).
\]
This equality highlights that the joint distribution of the test and calibration data is invariant to the swapping of \(T_i\) and \(\tilde{T}_i\) conditional on \((y_{j}: j \in [m])\), provided that  $T_i$ is an inlier. By construction, the randomness in both \(\mathbf{T}\) and \(\tilde{\mathbf{T}}\) comes from \(\{\epsilon_{i}\}\), conditioned on \((y_{i}: i \in [m])\). The desired pairwise exchangeability \eqref{data_pwexch} can be established by integrating out \((y_{i}: i \in [m])\).   \qed
    
}
 \end{example}

\subsubsection{Estimating the non-null proportion $\pi_{S_i}$ under the semi-supervised setup}
\label{app:subsub-confp}

Our proposed estimator for the non-null proportion $\pi_{S_i}$ (or local sparsity level) in \eqref{pii-tilde} relies on the availability of p-values. In the classical setting, these p-values can be computed directly based on the null distribution $F_0$. However, in the semi-supervised scenario, $F_0$ is unknown. Therefore, we propose an alternative approach to address this challenge by first constructing conformal p-values through the following steps:
\begin{enumerate}[1.]
    \item Split the training set $\mathcal{D}^{tr}$ into $\mathcal{D}^{tr}_1$ and $\mathcal{D}^{tr}_2$, denote $\mathbf{T}^{tr1}=\{T_{i}^0:i\in\mathcal{D}^{tr}_1\}$ and $\mathbf{T}^{tr2}=\{T_{i}^0:i\in\mathcal{D}^{tr}_2\}$.
    \item Learn some conformity score function $s(t)=s(t;\mathbf{T},\Tilde{\mathbf{T}},\mathbf{T}^{tr1},\mathbf{T}^{tr2})$ based on $(\mathbf{T},\Tilde{\mathbf{T}},\mathbf{T}^{tr1},\mathbf{T}^{tr2})$, where $s$ is chosen such that
    \begin{equation}\label{st-class}
    s\left(t;\mathbf{T},\Tilde{\mathbf{T}},\mathbf{T}^{tr1},\mathbf{T}^{tr2}\right)=s\left(t;(\mathbf{T},\Tilde{\mathbf{T}},\mathbf{T}^{tr1})_\Pi,\mathbf{T}^{tr2}\right),
    \end{equation}
    for any permutation $\Pi$ on $(\mathbf{T},\Tilde{\mathbf{T}},\mathbf{T}^{tr1})$.
    \item Calculate the conformity scores, and define the conformal $p$-values by
    \begin{equation}\label{pu-conf-pv}
        \hat{p}(T_{i})=\frac{1+|\{k\in\mathcal{D}^{tr}_{1}:s(T_{k}^0)\leq s(T_{i})\}| }{1+|\mathcal{D}^{tr}_1|}, \quad \hat{p}(\Tilde{T}_{i})=\frac{1+|\{k\in\mathcal{D}^{tr}_{1}:s(T_{k}^0)\leq s(\Tilde{T}_{i})\}| }{1+|\mathcal{D}^{tr}_1|}, \quad i\in[m]. 
    \end{equation}
\end{enumerate}

We now establish the exchangeability properties of the conformal p-values \eqref{pu-conf-pv}.

\begin{property}\label{app-prop:confpv}
    Consider the conformal p-values $\hat{p}(T_j)$ and $\hat{p}(\Tilde{T}_j)$ constructed by \eqref{pu-conf-pv} using score function $s(\cdot)$ satisfying \eqref{st-class}. Then we have:
    \begin{enumerate}[(a)]
        \item If $\mathbf{T}$ and $\mathbf{T}^0=\Tilde{\mathbf{T}}\cup\mathbf{T}^{tr}$ satisfy the joint exchangeability \eqref{jointexch}, then the null p-values 
    \begin{equation}\label{ssmt-exch1}
    \left(\hat{p}(\Tilde{T}_1), \cdots, \hat{p}(\Tilde{T}_m), \hat{p}(T_i), i\in\mathcal{H}_0 \right)
    \end{equation}
    are jointly exchangeable.        
        \item If $\mathbf{T}$, $\Tilde{\mathbf{T}}$ and $\mathbf{T}^{tr}$ satisfy the pairwise exchangeability \eqref{data_pwexch}, then the null p-values are pairwise exchangeable:
        \begin{equation}\label{ssmt-exch2}
        \big( \hat{p}(T_i),\hat{p}(\Tilde{T}_i) \big) \overset{d}{=} \big( \hat{p}(\Tilde{T}_i),\hat{p}(T_i) \big) \text{ conditional on } \big( \hat{p}(T_j),\hat{p}(\Tilde{T}_j):j\neq i \big) \mbox{  for $i\in\mathcal{H}_0$.}
        \end{equation}
    \end{enumerate}
\end{property}

\begin{remark}\rm{
The methodology and theoretical framework presented in this subsection for constructing conformal p-values using $(\mathbf{T}, \Tilde{\mathbf{T}}, \mathbf{T}^{tr})$ is closely related to but departs from existing approaches \citep{mary22semi,marandon22mlfdr,bates23} due to the incorporation of new exchangeability conditions \eqref{ssmt-exch1} and \eqref{ssmt-exch2} that involve null p-values in both the test and calibration sets.}
\end{remark}

The construction of the conformalized estimator for the non-null proportion within the semi-supervised framework entails replacing conventional p-values in \eqref{pii-tilde} with conformal $p$-values presented in \eqref{pu-conf-pv}:
\begin{equation}
   \tilde{\pi}^{**}_{S_{i}}=1-\frac{\sum_{j=1}^{m}w_{ij}[\II\{\hat{p}(T_{j})>\lambda\}+\II\{\hat{p}(\Tilde{T}_{j})>\lambda\}]}{2(1-\lambda)\sum_{j=1}^{m}w_{ij}}. \label{pu-p**}
\end{equation}

\subsubsection{Density ratio estimation when $S_i$ is discrete}\label{sub:grouping}

The next two subsections explore the extension of the PU learning strategy for estimating the density ratio $\hat{r}(t, S)$  in the presence of side information, including the grouping strategy (discrete $S_i$) and augmentation strategy (continuous $S_i$).

In Section \ref{app:subsub-pu-group}, we have proposed \eqref{app-pu-group-dr} for estimating $\hat{r}(t, S)$ when the covariates indicate group membership, and further provided Proposition \ref{prop:group} to justify the pairwise exchangeability condition. Next we establish a property to consolidate Proposition \ref{prop:group}.

\begin{property}\label{app-prop:group}
    Consider the ranking score function $\hat{R}(t,k)$ gained by the transformation \eqref{Rthat} with $\widehat{\mathrm{Clfdr}}^{**}(t,k)=(1-\hat{\pi}_{k}^{**})\hat{r}(t,k)$, where $\hat{\pi}_{k}^{**}$ is defined by \eqref{pu-p**} and $\hat{r}(t,k)$ is deduced by \eqref{app-pu-group-dr}. If $\mathbf{T}$, $\Tilde{\mathbf{T}}$, $\mathbf{T}^{tr}$ and $\mathbf{S}$ satisfy the conditional pairwise exchangeability \eqref{data_pwexch}, then the scores 
    $u_i= \hat{R}(T_i, S_i)$, $ \tilde u_i= \hat{R}(\tilde T_i, S_i)$ satisfy the pairwise exchangeability \eqref{pwexch}.
\end{property}

Although the conclusions of Proposition \ref{prop:group} and Property \ref{app-prop:group} are identical, the conditions in Property \ref{app-prop:group} are weaker because: (a) the conformal p-value is allowed to depend on the training data beyond a known non-random function $F_0$, and (b) the exchangeability assumption \eqref{jointexch-covariate} is relaxed to pairwise exchangeability \eqref{data_pwexch}. Hence, in Section \ref{app:proof-4-app}, we only verify Property \ref{app-prop:group}, from which Proposition \ref{prop:group} directly follows as a corollary.

\subsubsection{Density ratio estimation when $S_i$ is continuous}\label{subsub-pu-aug}

Consider the working model \eqref{model:mixture}. Let $q(s)$ denote the marginal density of $S$, $f(t,s)$ denote the joint probability density of $(T,S)$, and $f_{0}(t,s)=f_{0}(t)q(s)$ denote the joint probability density of $(T,S)$ under the null. The conditional independence between $T$ and $S$ under the null implies the relationship:
$
{f_{0}(t)}/{f_{s}(t)}={f_{0}(t,s)}/{f(t,s)}.
$
This observation serves as motivation to augment both the test data and corresponding calibration data with the covariate. The data augmentation process consists of three steps.

In Step 1, we create augmented data $T_{i}^{+}=(T_{i},S_{i})$ and $\Tilde{T}_{i}^{+}=(\Tilde{T}_{i},S_{i})$, $i\in[m]$, for both the test and calibration sets. 

In Step 2, we randomly pair each $S_{i}$ with one training sample $T_{i}^{tr}\in \mathbf{T}^{tr}=(T_j^0:j\in\mathcal{D}^{tr}\subset\mathcal{D}_{0})$ to obtain augmented training data $\{T_{i}^{tr+}\}_{i\in\mathcal{D}^{tr}}=\{(T_{i}^{tr},S_{i})\}_{i\in\mathcal{D}^{tr}}$.

In Step 3, we apply a PU learning algorithm, which is permutation invariant to the unordered set $\cup_{i\in[m]}\{T_i^+, \Tilde T_i^+\}$, to estimate the ratio of the density of $\{T_{i}^{tr+}\}$ to the density of $\{T_{i}^{+}\}\cup\{\Tilde{T}_{i}^{+}\}$. This ratio is denoted as:
\begin{equation}\label{app-pu-augm-dr}
    \hat{r}(t,s) = \hat{r}\left(t,s; \cup_{i\in[m]}\{T_i^+,\Tilde{T}_i^+\},\{T_i^{tr+}:i\in[m]\}\right).
\end{equation}
By applying transformation \eqref{Rthat} to $\widehat{\mathrm{Clfdr}}^{**}(t,s)=(1-\hat{\pi}_{s}^{**})\hat{r}(t,s)$, the conformity score function $\hat{R}(t,s)$ can be obtained. The pairwise exchangeability between conformity scores is established in the next property. 

\begin{property}\label{app-prop:augm}
Consider the conformity score function $\hat{R}(t,s)$ calculated by following Steps 1-3. If $\mathbf{T}$, $\Tilde{\mathbf{T}}$, $\mathbf{T}^{tr}$ and $\mathbf{S}$ satisfy the conditional pairwise exchangeability \eqref{data_pwexch}, then 
    $u_i= \hat{R}(T_i, S_i)$, $ \tilde u_i= \hat{R}(\tilde T_i, S_i)$ satisfy the pairwise exchangeability condition \eqref{pwexch}.
\end{property}

\begin{remark}\rm{In Step 2, if $|\mathcal{D}^{tr}|<m$, we may sample $T_{i}^{tr}$ from $\mathbf{T}^{tr}$ with replacement. On the other hand, when $|\mathcal{D}^{tr}|>m$, we have two strategies to make optimal use of the null samples. 
The first strategy involves sampling from $\mathbf{S}$ with replacement. Such strategies are valid because resampling $\mathbf{T}^{tr}$ and $\mathbf{S}$ still ensures that the pairwise exchangeability condition \eqref{data_pwexch} holds.
Alternatively, the second strategy involves implementing a derandomized procedure to enhance reliability and efficiency. This can be achieved by leveraging the e-values obtained from Algorithm \ref{algo:claw_clfdr}, as discussed in Section \ref{subsec:deran-claw}. 
 }
\end{remark}

\section{Proofs for Primary Theory}
\label{app:proof}
This section proves the primary theories in the main text.

\subsection{Proof of Proposition \ref{thm:ev}}
\label{app:proof-ev}

We first state and proof a lemma that is instrumental for establishing the finite-sample FDR theory concerning Algorithm \ref{algo:claw}. It delineates the method and theory on utilizing Algorithm \ref{algo:claw} to construct generalized e-values, paving the way for employing the e-BH theory in \cite{wang22ev} for our problem. 

\begin{lemma}\label{lemma1}
    Suppose that the scores $(u_1,\cdots,u_m)$ and $(\Tilde{u}_1,\cdots,\Tilde{u}_m)$ satisfy \eqref{pwexch}. Let $\tau$ be the  threshold output by Algorithm \ref{algo:claw}. If there is no ties between $u_i$ and $\Tilde{u}_i$ almost surely, then
    \begin{equation*}
        \EE \left[ \frac{\sum_{j\in\mathcal{H}_{0}} \II \{u_{j}\leq\tau\wedge\Tilde{u}_{j}\} }{1+\sum_{j\in\mathcal{H}_{0}} \II \{\Tilde{u}_{j}\leq\tau\wedge u_{j}\}}  \right] =
        \EE \left[ \frac{\sum_{j\in\mathcal{H}_{0}} \II\{u_{j}<\Tilde{u}_{j}\}\II \{u_{j}\leq\tau\} }{1+\sum_{j\in\mathcal{H}_{0}} \II\{\Tilde{u}_{j}< u_{j}\}\II \{\Tilde{u}_{j}\leq\tau\}}  \right] 
        \leq1.
    \end{equation*} 
\end{lemma}

\begin{proof}[Proof of Lemma \ref{lemma1}.]

We first present an equivalent expression of Algorithm \ref{algo:claw}. Let $\nu_{i}=u_{i}\wedge\Tilde{u}_{i}$ and $\eta_{i}=\II\{u_{i}<\Tilde{u}_{i}\}$. Since $\PP(u_{i}=\Tilde{u}_{i})=0$, we have $\II\{\Tilde{u}_{i}<u_{i}\}=1-\eta_{i}$ almost surely for all $i\in[m]$. As $Q(t)$ only jumps at the points in the set $\{\nu_i: i\in[m]\}$, the threshold $\tau$ output by Algorithm \ref{algo:claw} can be narrowed down within the set $\{\nu_i\}_{i=1}^m$, i.e.  
$\tau = \max\{ t\in\mathcal{U}\cup\Tilde{\mathcal{U}}: Q(t) \leq \alpha \} = \max\{ t\in\{\nu_i\}_{i=1}^m: Q(t)\leq \alpha \}$. Let $\nu_{(1)}\leq\cdots\leq\nu_{(m)}$ be the order statistics. We have $\tau=\nu_{(\hat{k})}$, where
\begin{equation}\label{anotherstop}
    Q(t)= \frac{1+\sum_{j=1}^m (1-\eta_{j})\II\{\nu_j\leq t\} }{ [\sum_{j=1}^m \eta_{j} \II\{\nu_j\leq t\}]\vee1 } \mbox{ and }   \hat{k}=\max\{i\in[m]:Q(\nu_{(i)})\leq\alpha\}.
\end{equation}

We first claim that:
    \begin{equation}\label{flipcoin}
        (\eta_{i}:i\in\mathcal{H}_0)\overset{i.i.d.}{\sim}B(1,1/2) \text{ conditional on } (\nu_1,\cdots,\nu_m).
    \end{equation}
To prove the claim, we first state a useful lemma without proof. 

\begin{lemma}\label{lem:barbercandes}(\citealp{barber15knockoff})
For any anti-symmetric function $h(x,y)$ satisfying $h(x,y)=-h(y,x)$, if the scores $(u_1,\cdots,u_m)$ and $(\Tilde{u}_1,\cdots,\Tilde{u}_m)$ are pairwise exchangeable under the null, i.e., \eqref{pwexch} holds, then $(\mathrm{sign}(h(u_i,\Tilde{u}_i)):i\in\mathcal{H}_0)$ are i.i.d. coin flips conditional on $(|h(u_i,\Tilde{u}_i)|:i\in[m])$.
\end{lemma}
    
    We start by considering the sign of $u_i-\Tilde{u}_i$, indicated by $\{-1, 1\}$, which can also be expressed as $1-2\eta_i$. Define the following anti-symmetric function $h(x,y)$
    $$h(x,y)=\mathrm{sign}(x-y)(x\wedge y).$$
    Then $\nu_i=|h(u_i,\Tilde{u}_i)|$ and $1-2\eta_i=\mathrm{sign}(u_i-\Tilde{u}_i)=\mathrm{sign}(h(u_i,\Tilde{u}_i))$. 
    By Lemma \ref{lem:barbercandes}, we have that $((1-2\eta_i:i\in\mathcal{H}_0)$ are i.i.d. coin flips conditional on $(\nu_i:i\in[m])$, and $\PP(\eta_{i}=0|\nu_1,\cdots,\nu_m)=\PP(1-2\eta_i=1|\nu_1,\cdots,\nu_m)=1/2$, $\PP(\eta_{i}=1|\nu_1,\cdots,\nu_m)=\PP(1-2\eta_{i}=-1|\nu_1,\cdots,\nu_m)=1/2$. Equivalently, \eqref{flipcoin} holds.

    Let $\mathcal{G}=\sigma\left( (\nu_{i}:i\in[m]), \{\eta_{i}:i\notin\mathcal{H}_0\} \right)$. Consider the filtration $\mathcal{F}=(\mathcal{F}_k:k\in[m])$ generated by
    \begin{equation*}
        \mathcal{F}_{k}=\sigma\left( \mathcal{G} \cup \sigma(V_{j},\Tilde{V}_{j}:k\leq j \leq m) \right),
    \end{equation*}
    where $V_{j}=\sum_{l\in\mathcal{H}_0}\eta_{l}\II\{\nu_{l}\leq\nu_{(j)}\},\quad \Tilde{V}_j = \sum_{l\in\mathcal{H}_0}(1-\eta_{l})\II\{\nu_{l}\leq\nu_{(j)}\}.$
    It is easy to check that $\mathcal{F}_{i+1}\subset\mathcal{F}_{i}$. Define the following random process
    $$M_{i}=\frac{V_i}{1+\Tilde{V}_i},\quad i=1,\cdots,m.$$
Following the arguments, for example, in \citet{barber15knockoff} or \citet{zhao2023plis}, we can show that $(M_i:i\in[m])$ is a backward discrete-time super-martingale with respect to $\mathcal{F}$, i.e.,
    \begin{equation}\label{tower-exp}
        \EE[M_{i}|\mathcal{F}_{i+1}] \leq M_{i+1} ,\quad \forall i\in[m-1].
    \end{equation}
Moreover, $\hat{k}$ is an $\mathcal{F}$-stopping time, as knowing $\{\eta_{j}:j\notin\mathcal{H}_{0}\}$, $(\nu_{j}:j\in[m])$ and $\{V_{i}, \Tilde{V}_{i}: K\leq i\leq m\}$ is sufficient to determine whether the event $\{\hat{k}=K\}$ occurs. Therefore, we can apply Doob's optional stopping theorem on $(M_{i}:i\in[m])$ and $\hat{k}$ to establish that
    $$\EE[M_{\hat{k}}]\leq\EE[M_m]=\EE\left[ \frac{\sum_{j\in\mathcal{H}_{0}} \eta_{j}}{1+\sum_{j\in\mathcal{H}_{0}} (1-\eta_{j})} \right].$$
The above expectation can be computed through various methods (cf. \citealp{barber15knockoff, weinstein17counting}). Alternatively, we introduce a novel and more generic approach that capitalizes on the pairwise exchangeability property inherent in our problem setup. This calculation entails the application of the following simple yet useful lemma.

\begin{lemma}\label{lem:ex}
For non-negative random variables $X$, $Y$ and $Z$ satisfying $(X,Y,Z)\overset{d}{=}(Y,X,Z)$, we have
$$ 
\mathbb{E}\left[ \frac{X}{X+Y+Z} \right] = \mathbb{E}\left[ \frac{Y}{X+Y+Z} \right].
$$
\end{lemma}
\begin{proof}[Proof of Lemma \ref{lem:ex}.]
Since $(X,Y,Z)\overset{d}{=}(Y,X,Z)$, we have that $(X,Y,X+Y+Z)\overset{d}{=}(Y,X,Y+X+Z),$ which implies $$\frac{X}{X+Y+Z} \overset{d}{=} \frac{Y}{X+Y+Z},$$ and the lemma follows.  
\end{proof}

    By \eqref{flipcoin}, we have that $(\eta_{j}:j\in\mathcal{H}_0)\overset{i.i.d.}{\sim} B(1,1/2)$ conditional on $(\nu_i:i\in[m])$, which implies that 
    \begin{equation*}
    \begin{split}
        & \Big( \eta_i,1-\eta_i \Big| \sum_{j\in\mathcal{H}_0,j\neq i}\eta_{j},(\nu_i:i\in[m]) \Big) \overset{d}{=} \Big( \eta_i,1-\eta_i \Big| (\nu_i:i\in[m]) \Big)  \\
        \overset{d}{=}& \Big( 1-\eta_i,\eta_i \Big| (\nu_i:i\in[m]) \Big) \overset{d}{=} \Big( 1-\eta_i,\eta_i \Big| \sum_{j\in\mathcal{H}_0,j\neq i}\eta_{j},(\nu_i:i\in[m]) \Big).
    \end{split}
    \end{equation*}

Integrating $(\nu_i:i\in[m])$ out, the following pairwise exchangeability holds:
    \begin{equation}
    \Big( \eta_i,1-\eta_i , \sum_{j\in\mathcal{H}_0,j\neq i}\eta_{j} \Big) \overset{d}{=} \Big( 1-\eta_i,\eta_i ,\sum_{j\in\mathcal{H}_0,j\neq i}\eta_{j} \Big).
    \end{equation}
By  Lemma \ref{lem:ex}, we have that
\begin{equation*}
\begin{split}
   & \EE\left[ \frac{\sum_{j\in\mathcal{H}_{0}} \eta_{j}}{1+\sum_{j\in\mathcal{H}_{0}} (1-\eta_{j})} \right] 
    \leq  \sum_{i\in\mathcal{H}_{0}} \EE\left[ \frac{ \eta_{i}}{\eta_{i}+(1-\eta_i)+\sum_{j\in\mathcal{H}_{0},j\neq i} (1-\eta_{j})} \right] \\
    = &  \sum_{i\in\mathcal{H}_{0}} \EE\left[ \frac{ 1-\eta_{i}}{\eta_{i}+(1-\eta_i)+\sum_{j\in\mathcal{H}_{0},j\neq i} (1-\eta_{j})} \right] \text{   (apply Lemma \ref{lem:ex})}\\
    \leq& \sum_{i\in\mathcal{H}_{0}} \EE\left[ \frac{ 1-\eta_{i}}{(1-\eta_i)+\sum_{j\in\mathcal{H}_{0},j\neq i} (1-\eta_{j})} \right] 
    = 1.
\end{split}
\end{equation*}
The proof of Lemma \ref{lemma1} is complete by noting that
    \begin{equation*}
  \EE \left[ \frac{\sum_{j\in\mathcal{H}_{0}} \II\{u_{j}<\Tilde{u}_{j}\}\II \{u_{j}\leq\tau\} }{1+\sum_{j\in\mathcal{H}_{0}} \II\{\Tilde{u}_{j}< u_{j}\}\II \{\Tilde{u}_{j}\leq\tau\}}  \right]
           \leq \EE [M_{m}]
            =\EE\left[ \frac{\sum_{j\in\mathcal{H}_{0}} \eta_{j}}{1+\sum_{j\in\mathcal{H}_{0}} (1-\eta_{j})} \right] \leq 1.
    \end{equation*}       
\end{proof}

\begin{proof}[Proof of Proposition \ref{thm:ev}.]

First, we can see that $\mathbb{E} \big[ \sum_{j \in \mathcal{H}_{0} }e_{j} \big]\leq m$ according to Lemma \ref{lemma1}. Therefore, the e-BH is valid for such a set of generalized e-values. Let $R=|\mathcal{R}|$. By the definition of  $\tau$, we have $\frac{1+\sum_{j=1}^{m} \II \{\Tilde{u}_{j}\leq\tau\wedge u_{j}\}}{R}\leq\alpha$, so for $j\in\mathcal{R}$,
\begin{equation*}
    e_{j} = \frac{m \II \{u_{j}\leq\tau\wedge\Tilde{u}_{j}\} }{1+\sum_{i=1}^{m} \II \{\Tilde{u}_{i}\leq\tau\wedge u_{i}\}} \geq \frac{m}{\alpha R}.
\end{equation*}
Therefore, $\hat{k} = \max\{i:e_{(i)}\geq\frac{m}{\alpha i}\} \geq R$. Since only the largest $R$ e-values are non-zero, we have that $e_j\geq e_{(R)}\geq e_{(\hat k)}$, indicating $j\in \mathcal{R}_{ebh}$. Conversely, if $j\notin \mathcal{R}$, $e_{j}=0$, which means that $j$ cannot be selected by the e-BH procedure, then $j\notin \mathcal{R}_{ebh}$. In conclusion, $\mathcal{R}=\mathcal{R}_{ebh}$.
\end{proof}

\subsection{Proof of Theorem \ref{thm:fdr}}

The theorem can be established as a corollary of Proposition \ref{thm:ev}: the e-BH procedure with generalized e-values defined in (\ref{newev}) is equivalent to Algorithm \ref{algo:claw}. Hence the conclusion follows from the e-BH theory \citep{wang22ev}. For readers interested in an alternative proof, we offer one directly utilizing Lemma \ref{lemma1}. Note that

\begin{equation*}
        \begin{split}
            \mathrm{FDP}(\mathcal{R}) &= \frac{\sum_{i\in\mathcal{H}_{0}} \II\{u_{i}\leq\tau\wedge\Tilde{u}_i\} }{(\sum_{i=1}^{m} \II\{u_{i}\leq\tau\wedge\Tilde{u}_i\}) \vee 1 } \\
            &= \frac{1+\sum_{i=1}^{m} \II\{\Tilde{u}_{i}\leq\tau\wedge u_{i}\} }{(\sum_{i=1}^{m} \II\{u_{i}\leq\tau\wedge\Tilde{u}_{i}\}) \vee 1 } \cdot \frac{1+\sum_{i\in\mathcal{H}_{0}} \II\{u_{i}\leq\tau\wedge\Tilde{u}_{i}\} }{1+\sum_{i=1}^{m} \II\{\Tilde{u}_{i}\leq\tau\wedge u_{i}\} }\\
            &= Q(\tau) \cdot \frac{1+\sum_{i\in\mathcal{H}_{0}} \II\{\Tilde{u}_{i}\leq\tau\wedge u_i\} }{1+\sum_{i=1}^{m} \II\{\Tilde{u}_{i}\leq\tau\wedge u_i\} } \cdot \frac{\sum_{j\in\mathcal{H}_{0}} \II \{u_{j}\leq\tau\wedge\Tilde{u}_{j}\} }{1+\sum_{j\in\mathcal{H}_{0}} \II \{\Tilde{u}_{j}\leq\tau\wedge u_{j}\}}\\
            &\leq \alpha \cdot 1 \cdot \frac{\sum_{j\in\mathcal{H}_{0}} \II \{u_{j}\leq\tau\wedge\Tilde{u}_{j}\} }{1+\sum_{j\in\mathcal{H}_{0}} \II \{\Tilde{u}_{j}\leq\tau\wedge u_{j}\}}.
        \end{split}
    \end{equation*}
    The last inequality holds because of the definition of $\tau$ and the trivial fact $\mathcal{H}_{0}\subset[m]$. The desired result follows by taking expectations on the both sides:
    \begin{equation*}
        \mathrm{FDR}=\EE[\mathrm{FDP}(\mathcal{R})] \leq \alpha \EE \left[ \frac{\sum_{j\in\mathcal{H}_{0}} \II \{u_{j}\leq\tau\wedge\Tilde{u}_{j}\} }{1+\sum_{j\in\mathcal{H}_{0}} \II \{\Tilde{u}_{j}\leq\tau\wedge u_{j}\}}  \right] \leq\alpha. \qed
    \end{equation*}

\subsection{Proof of Theorem \ref{thm:exch}} 
\label{app:thm2}

\begin{proof}[Proof of part (a).]
Let $\psi(x,y)$ be a vector-valued symmetric function satisfying $\psi(x,y)=\psi(y,x)$. Consider two random elements $X$ and $Y$ that are pairwise exchangeable, i.e. $(X,Y)\overset{d}{=}(Y,X)$. Then we have
\begin{equation}\label{pw_sym_func}
(X,Y,\psi(X,Y)) \overset{d}{=} (Y,X,\psi(Y,X))  {=} (Y,X,\psi(X,Y)).
\end{equation}

Suppose we are interested in utilizing function $g(t,S_{j};(\mathbf{T},\Tilde{\mathbf{T}}),\mathbf{S})$ to construct conformity scores. The swapping-invariance property \eqref{principle-pw} implies that $g$ is fully determined by the unordered pairs $\{T_1,\Tilde{T}_1\},\cdots,\{T_m,\Tilde{T}_m\}$ and the covariate sequence $\mathbf{S}$. To emphasize that $g$ it is invariant when swapping $T_i$ and $\Tilde{T}_i$, we adopt the notation $g(t,S_{j};\{T_i,\Tilde{T}_i\},(\mathbf{T}_{-i},\Tilde{\mathbf{T}}_{-i}),\mathbf{S})$, where $\{T_i,\Tilde{T}_i\}$ represents the unordered set of $T_i$ and $\Tilde{T}_i$. The corresponding scores are 
\begin{equation}\label{uj-tuj}
u_j=g(T_j,S_{j};\{T_i,\Tilde{T}_i\},(\mathbf{T}_{-i},\Tilde{\mathbf{T}}_{-i}),\mathbf{S}), \quad
\Tilde{u}_j=g(\Tilde{T}_j,S_{j};\{T_i,\Tilde{T}_i\},(\mathbf{T}_{-i},\Tilde{\mathbf{T}}_{-i}),\mathbf{S}).
\end{equation}

Let $\mathbf{G}_i \equiv(u_1,\cdots,u_{i-1},u_{i+1},\cdots,u_m,\Tilde{u}_1,\cdots,\Tilde{u}_{i-1},\Tilde{u}_{i+1},\cdots,\Tilde{u}_{m},T_i\vee\Tilde{T}_{i},T_i\wedge \Tilde{T}_{i})$. The vector $\mathbf{G}_i$ comprises two components. The first part encompasses scores from units excluding $i$: $$(u_1,\cdots,u_{i-1},u_{i+1},\cdots,u_m,\Tilde{u}_1,\cdots,\Tilde{u}_{i-1},\Tilde{u}_{i+1},\cdots,\Tilde{u}_{m})\coloneqq(\mathbf{U}_{-i},\Tilde{\mathbf{U}}_{-i}),$$ while the second part $(T_i\vee\Tilde{T}_{i},T_i\wedge \Tilde{T}_{i})$ provides the values of the unordered set $\{T_i, \Tilde T_i\}$. 

Note that $(T_i\vee\Tilde{T}_{i},T_i\wedge \Tilde{T}_{i})=(\Tilde{T}_i\vee{T}_{i},\Tilde{T}_i\wedge {T}_{i})$, and the scores $u_i$ and $\Tilde{u}_i$ are swapping invariant [cf. \eqref{uj-tuj}]. Given $(\mathbf{T}_{-i},\Tilde{\mathbf{T}}_{-i},\mathbf{S})$, the following mapping
$$
(T_i,\Tilde{T}_i) \mapsto \mathbf{G}_i \equiv(u_1,\cdots,u_{i-1},u_{i+1},\cdots,u_m,\Tilde{u}_1,\cdots,\Tilde{u}_{i-1},\Tilde{u}_{i+1},\cdots,\Tilde{u}_{m},T_i\vee\Tilde{T}_{i},T_i\wedge \Tilde{T}_{i})
$$
represents a (vector-valued) bivariate function that is symmetric with respect to $(T_i,\Tilde{T}_i)$.

Since permutation invariance implies swapping invariance, a direct consequence of condition \eqref{jointexch-covariate} is 
$$
(T_{i},\Tilde{T}_{i}|\mathbf{T}_{-i},\Tilde{\mathbf{T}}_{-i},\mathbf{S}) \overset{d}{=} (\Tilde{T}_{i},T_i|\mathbf{T}_{-i},\Tilde{\mathbf{T}}_{-i},\mathbf{S})
$$ for $i\in\mathcal{H}_0$. Applying \eqref{pw_sym_func}, we have
\begin{equation}\label{proof-condi-ex-id}
    \Big( T_i, \Tilde{T}_i \Big| \mathbf{G}_{i} , \mathbf{T}_{-i},\Tilde{\mathbf{T}}_{-i},\mathbf{S} \Big) \overset{d}{=} \Big( \Tilde{T}_i, T_i \Big|\mathbf{G}_{i}, \mathbf{T}_{-i},\Tilde{\mathbf{T}}_{-i},\mathbf{S} \Big).
\end{equation}

As the function $g(t,S_{i}; \{T_i,\Tilde{T}_i\},(\mathbf{T}_{-i},\Tilde{\mathbf{T}}_{-i}),\mathbf{S})$ is nonrandom with respect to $$\sigma(\{T_i,\Tilde{T}_i\},(\mathbf{T}_{-i},\Tilde{\mathbf{T}}_{-i}),\mathbf{S}) \subset \sigma(\mathbf{G}_{i}, {\mathbf{T}}_{-i}, \Tilde{\mathbf{T}}_{-i},\mathbf{S})$$ it follows from  \eqref{proof-condi-ex-id} that
$$
\Big( u_i, \Tilde{u}_i \Big| \mathbf{G}_{i},\mathbf{T}_{-i},\Tilde{\mathbf{T}}_{-i},\mathbf{S} \Big) \overset{d}{=} \Big( \Tilde{u}_i, u_i \Big|\mathbf{G}_{i}, \mathbf{T}_{-i},\Tilde{\mathbf{T}}_{-i},\mathbf{S} \Big), \quad \mbox{for $i\in\mathcal{H}_0$.}
$$
Finally, by integrating out $(\mathbf{T}_{-i},\Tilde{\mathbf{T}}_{-i},\mathbf{S})$ and $(T_i\vee\Tilde{T}_{i},T_i\wedge \Tilde{T}_{i})$, we arrive at the desired conclusion \eqref{pwexch}.
\end{proof}

\begin{remark}\rm
    In the proof of part (a), we have implicitly used the following fact: if $T_i$ and $\Tilde{T}_i$ are exchangeable, then $T_i$ and $\Tilde{T}_i$ are still exchangeable conditional on their order statistics $(T_i\vee\Tilde{T}_{i},T_i\wedge \Tilde{T}_{i})$, or the unordered set $\{T_i,\Tilde{T}_i\}$. The conclusion can be naturally extended to exchangeable variables sequence with arbitrary finite length (cf. \citealp{papadatos2022order}), which has been utilized in \citet{marandon22mlfdr}.
\end{remark}

\begin{proof}[Proof of part (b).]
We start by introducing the notations below for conciseness:
\begin{equation*}
\begin{split}
    \mathbf{A} &= (A_1,\cdots,A_{m+|\mathcal{H}_0|})=(\Tilde{T}_1,\cdots,\Tilde{T}_{m},T_i:i\in\mathcal{H}_0), \quad
    \mathbf{B} = (T_i:i\notin\mathcal{H}_0), \\
    \mathcal{C} &= \{T_1,\cdots,T_m,\Tilde{T}_1,\cdots,\Tilde{T}_m\}, \mbox{ an unordered multiset of the elements in $(\mathbf{T},\Tilde{\mathbf{T}})$}, \\
    U_i &= g(A_i;(\mathbf{T},\Tilde{\mathbf{T}})) = g(A_i;\mathcal{C}),\quad i\in\{1,\cdots,m+|\mathcal{H}_0|\}.
\end{split}
\end{equation*}

The permutation-invariance property \eqref{principle-jt} implies that $g(t;(\mathbf{T},\Tilde{\mathbf{T}})) = g(t;\mathcal{C})$. 
By condition \eqref{jointexch}, $$({A}_{\Pi_0(1)},\cdots,{A}_{\Pi_0(m+|\mathcal{H}_0|)},\mathbf{B}) \overset{d}{=} ({A}_1,\cdots,A_{m+|\mathcal{H}_0|},\mathbf{B}),$$ for any permutation $\Pi_0$ of $\{1,\cdots,m+|\mathcal{H}_0|\}$. According to \cite{papadatos2022order}, if a set of variables are exchangeable, then the variables remain exchangeable conditional on their unordered multiset. We conclude that
$$
({A}_{\Pi_0(1)},\cdots,{A}_{\Pi_0(m+|\mathcal{H}_0|)}| \mathbf{B} ,\mathcal{C}) \overset{d}{=} ({A}_1,\cdots,A_{m+|\mathcal{H}_0|} | \mathbf{B} ,\mathcal{C}).
$$

Since $g(t;(\mathbf{T},\Tilde{\mathbf{T}})) = g(t;\mathcal{C})$ is nonrandom conditional on $\mathcal{C}$, we have
$$\Big( g(A_1;\mathcal{C}),\cdots, g(A_{m+|\mathcal{H}_0|};\mathcal{C}) \Big| \mathbf{B},\mathcal{C} \Big) \overset{d}{=} \Big( g(A_{\Pi_{0}(1)};\mathcal{C}),\cdots, g(A_{\Pi_{0}(m+|\mathcal{H}_0|)};\mathcal{C}) \Big| \mathbf{B},\mathcal{C} \Big). 
$$
Note that $u_i=g(T_i;(\mathbf{T},\Tilde{\mathbf{T}}))=g(T_i;\mathcal{C})$ for $i\notin\mathcal{H}_0$ are deterministic given $(\mathbf{B},\mathcal{C})$, we have 
$$\Big( U_1,\cdots, U_{m+|\mathcal{H}_0|} \Big| \mathbf{B},\mathcal{C},(u_i:i\notin\mathcal{H}_0) \Big) \overset{d}{=} \Big( U_{\Pi_{0}(1)},\cdots, U_{\Pi_{0}(m+|\mathcal{H}_0|)} \Big|  \mathbf{B},\mathcal{C},(u_i:i\notin\mathcal{H}_0) \Big), $$
where $U_j=g(A_j;\mathcal{C})$, $j=1,\cdots,m+|\mathcal{H}_0|$. The desired result follows by integrating out $(\mathbf{B},\mathcal{C})$. 
\end{proof}

\subsection{Proof of Proposition \ref{prop:opt}}
\label{app:prop1}

Consider \( R_{i}(t) \) defined in \eqref{or_Rt}. Under model \eqref{model:mixture}, the elements in the set \(\{(T_{i},S_{i},\theta_{i}):i\in[m]\}\) are independent with each other. It follows that \(\PP(\theta_{i}=0|\mathbf{T},\mathbf{S})=\frac{(1-\pi_{S_{i}})f_{0}(T_{i})}{f_{S_i}(T_{i})} \coloneqq \mathrm{Clfdr}_{i}\). Let \( t^{OR} = \frac{t^{*}}{1+t^{*}} \). The monotonicity of the transformation \( x \mapsto \frac{x}{1+x} \) implies that the following two decisions are equivalent:
\[
\II \{i \in \mathcal{A}, R_{i}(T_{i}) \leq t^{*}\} = \II \{i \in \mathcal{A}, \mathrm{Clfdr}_i \leq t^{OR}\}, \quad \forall i \in [m].
\]
In the optimality theories presented in \cite{CARS} and \cite{marandon22mlfdr}, the oracle rules are based solely on the scores constructed from test data; hence, the decision rules are measurable with respect to \( (\mathbf{T},\mathbf{S}) \). However, the decision rule \(\pmb{\delta} = (\delta_i: i \in [m]) = (\II\{i \in \mathcal{R}\}: i \in [m])\) considered in our scenario is only measurable with respect to \( (\tilde{\mathbf{T}},\mathbf{T},\mathbf{S}) \) [particularly, $\pmb\delta$ is not measurable given only \( (\mathbf{T},\mathbf{S}) \)]. Thus, a more careful argument is necessary.

In what follows, the expectation is taken over the calibration, test, and auxiliary data \(\{\tilde{\mathbf{T}}, \mathbf{T}, \mathbf{S}\}\). The expected number of false positives can be calculated as:
\begin{equation}
	\begin{split}
		\mathbb{E}\Big[ \sum_{j\in\mathcal{H}_{0}} \delta_j \Big] =& \mathbb{E}\Big[ \sum_{j\in[m]} \delta_j \II\{\theta_j=0\} \Big] 
		= \mathbb{E} \left\{ \mathbb{E} \Big[\sum_{j=1}^{m} \II\{\theta_j=0\} \delta_j \Big| \tilde{\mathbf{T}}, \mathbf{T}, \mathbf{S} \Big] \right\} \\
		\overset{(i)}{=} & \mathbb{E} \left\{ \sum_{j=1}^{m} \delta_j \mathbb{E} \Big[\II\{\theta_j=0\} \Big| \tilde{\mathbf{T}}, \mathbf{T}, \mathbf{S} \Big] \right\} 
		\overset{(ii)}{=} \mathbb{E} \left\{ \sum_{j=1}^{m} \delta_j \mathbb{E} \Big[\II\{\theta_j=0\} \Big| \mathbf{T}, \mathbf{S} \Big] \right\} \\
		=& \mathbb{E} \left\{ \sum_{j=1}^{m} \delta_j \mathrm{Clfdr}_j \right\}.
	\end{split}
\end{equation}
Equality (i) holds because \((\delta_i : i \in [m])\) are measurable with respect to \((\tilde{\mathbf{T}}, \mathbf{T}, \mathbf{S})\), while Equality (ii) holds due to the independence between \((\theta_i : i \in [m])\) and \(\tilde{\mathbf{T}}\).

Let \(u_i = R_i(T_i)\) and \(\tilde{u}_i = R_i(\tilde{T}_i)\). Denote \(\mathrm{Clfdr}_{i}\) and \(\widetilde{\mathrm{Clfdr}}_i\) as the corresponding Clfdr values transformed from \(u_i\) and \(\tilde{u}_i\). Let \(\mathcal{A} = \{i \in [m] : u_i < \tilde{u}_i\} \equiv \{i \in [m] : \mathrm{Clfdr}_{i} < \widetilde{\mathrm{Clfdr}}_i\}\). As the proposition is trivially true if \(\mathcal{A} = \emptyset\), without loss of generality, we assume that \(\mathcal{A} \neq \emptyset\). According to our assumption, the rejection set is given by 
\[
\mathcal{R}_{u} = \{i : u_{i} \leq t^{*} \wedge \tilde{u}_i\} = \{i : \mathrm{Clfdr}_{i} \leq t^{OR} \wedge \widetilde{\mathrm{Clfdr}}_i\} = \{i \in \mathcal{A} : \mathrm{Clfdr}_{i} \leq t^{OR}\}.
\]
This rejection set \(\mathcal{R}_{u}\) has an mFDR of exactly \(\alpha\), which leads to
\begin{equation}\label{totalexpect} \small
    \mathbb{E} \left\{ \sum_{i=1}^{m} (\mathrm{Clfdr}_{i} - \alpha) \mathbb{I}\{ \mathrm{Clfdr}_{i} \leq t^{OR} \wedge \widetilde{\mathrm{Clfdr}}_i \} \right\} = \mathbb{E} \left\{ \sum_{i \in \mathcal{A}} (\mathrm{Clfdr}_{i} - \alpha) \mathbb{I}\{ \mathrm{Clfdr}_{i} \leq t^{OR} \} \right\} = 0.
\end{equation}
Thus, if \(\alpha > t^{OR}\), the sum in \eqref{totalexpect} would be negative, leading to the conclusion that \(\alpha \leq t^{OR}\).

Define
$Q_{OR}(t) = \frac{\sum_{j\in\mathcal{H}_{0}\cap\mathcal{A}} \mathbb{I}\{ \mathrm{Clfdr}_{j}\leq t \}}{\sum_{j\in\mathcal{A}} \mathbb{I}\{ \mathrm{Clfdr}_{j}\leq t \}}.$
Let $Q_{OR}(t_{j})=\alpha_{j}$ for $j=1,2$. By \eqref{totalexpect}, we have $\alpha_j\leq t_j$ and
\begin{equation}\label{totalexpect2}
    \EE \left[\sum_{i\in\mathcal{A}} (\mathrm{Clfdr}_{i}-\alpha_{j})  \mathbb{I} \{\mathrm{Clfdr}_{i}\leq t_{j}\} \right]=0.
\end{equation}

We claim that \(Q_{OR}(t)\) is monotone in \(t\). We only need to show that \(\alpha_{1} \leq \alpha_{2}\) if \(t_{1} < t_{2}\) and shall prove this by contradiction. Assume instead that \(\alpha_{1} > \alpha_{2}\) for \(t_{1} < t_{2}\). Then we can write:
\[
\begin{split}
    &(\mathrm{Clfdr}_{i} - \alpha_{2}) \mathbb{I}(\mathrm{Clfdr}_{i} \leq t_{2}) = (\mathrm{Clfdr}_{i} - \alpha_{2}) \mathbb{I}(\mathrm{Clfdr}_{i} \leq t_{1}) + (\mathrm{Clfdr}_{i} - \alpha_{2}) \mathbb{I}(t_{1} < \mathrm{Clfdr}_{i} \leq t_{2}) \\
    =& (\mathrm{Clfdr}_{i} - \alpha_{1}) \mathbb{I}(\mathrm{Clfdr}_{i} \leq t_{1}) + (\alpha_{1} - \alpha_{2}) \mathbb{I}(\mathrm{Clfdr}_{i} \leq t_{2}) + (\mathrm{Clfdr}_{i} - \alpha_{1}) \mathbb{I}(t_{1} < \mathrm{Clfdr}_{i} \leq t_{2}),
\end{split}
\]
where we have 
$
\mathbb{E}[(\alpha_{1} - \alpha_{2}) \mathbb{I}(\mathrm{Clfdr}_{i} \leq t_{2}) + (\mathrm{Clfdr}_{i} - \alpha_{1}) \mathbb{I}(t_{1} < \mathrm{Clfdr}_{i} \leq t_{2})] > 0
$
for all \(i \in \mathcal{A}\). Consequently, it follows that 
\[
\mathbb{E}\left[\sum_{i \in \mathcal{A}} (\mathrm{Clfdr}_{i} - \alpha_{2}) \mathbb{I}(\mathrm{Clfdr}_{i} \leq t_{2})\right] > 0,
\]
which contradicts the condition outlined in \eqref{totalexpect2}. 
Thus, we conclude that our initial assumption must be incorrect, claiming that \(Q_{OR}(t)\) is indeed monotone in \(t\).

Let \(\mathcal{R}' \subset \mathcal{A}\) be a rejection set satisfying \(\mbox{mFDR} \leq \alpha\). The corresponding individual decisions are defined as \(\delta_{i}' = 1\) for \(i \in \mathcal{R}'\) and \(\delta_{i}' = 0\) otherwise. Using similar arguments as in (\ref{totalexpect}), we have 
$$
\mathbb{E}\left[\sum_{i=1}^{m}(\mathrm{Clfdr}_{i}-\alpha)\delta_{i}'\right] = \mathbb{E}\left[\sum_{i \in \mathcal{A}}(\mathrm{Clfdr}_{i}-\alpha)\delta_{i}'\right] \leq 0.
$$
Note that \(\mathbb{I}\{ \mathrm{Clfdr}_{i} \leq t^{OR} \} = \mathbb{I}\left\{ \frac{\mathrm{Clfdr}_{i}-\alpha}{1 - \mathrm{Clfdr}_{i}} \leq \lambda_{OR}\right\}\). Let \(\lambda_{OR} = \frac{t^{OR}-\alpha}{1 - t^{OR}}\). Since \(\frac{x - \alpha}{1 - x}\) is increasing in \(x\) for \(\alpha < x < 1\), we have 
\begin{equation}\label{ineq1}
    \mathbb{E}\left[\sum_{i \in \mathcal{A}}(\mathrm{Clfdr}_{i} - \alpha)\left(\mathbb{I}\{ \mathrm{Clfdr}_{i} \leq t^{OR} \} - \delta_{i}'\right)\right] \geq 0.
\end{equation}
It follows that, for all \(i \in \mathcal{A}\):
\begin{eqnarray*}
\mathrm{Clfdr}_{i} - \alpha - \lambda_{OR}(1 - \mathrm{Clfdr}_{i}) \leq 0 \text{ if } \mathbb{I}\{ \mathrm{Clfdr}_{i} \leq t^{OR} \} > \delta_{i}', 
\\
\mathrm{Clfdr}_{i} - \alpha - \lambda_{OR}(1 - \mathrm{Clfdr}_{i}) > 0 \text{ if } \mathbb{I}\{ \mathrm{Clfdr}_{i} \leq t^{OR} \} < \delta_{i}'.
\end{eqnarray*}
We conclude that for all $i\in\mathcal{A}$, 
$$
(\mathbb{I}\{ \mathrm{Clfdr}_{i}\leq t^{OR} \}-\delta_{i}')[\mathrm{Clfdr}_{i}-\alpha-\lambda_{OR}(1-\mathrm{Clfdr}_{i})]\leq 0.
$$ 
Summing over $i$ and taking expectation, we have
\begin{equation}\label{ineq2}
    \mathbb{E}\left\{\sum_{i\in\mathcal{A}} (\mathbb{I}\{ \mathrm{Clfdr}_{i}\leq t^{OR} \}-\delta_{i}')[\mathrm{Clfdr}_{i}-\alpha-\lambda_{OR}(1-\mathrm{Clfdr}_{i})] \right\} \leq 0.
\end{equation}
Combining (\ref{ineq1}) and (\ref{ineq2}), we have
\begin{equation*}
  \lambda_{OR}\mathbb{E}\left\{\sum_{i\in\mathcal{A}} (\mathbb{I}\{ \mathrm{Clfdr}_{i}\leq t^{OR} \}-\delta_{i}')(1-\mathrm{Clfdr}_{i}) \right\} 
      \geq \mathbb{E}\left\{\sum_{i\in\mathcal{A}} (\mathbb{I}\{ \mathrm{Clfdr}_{i}\leq t^{OR} \}-\delta_{i}')(\mathrm{Clfdr}_{i}-\alpha) \right\} \geq 0.
\end{equation*}

Finally, noting that \(\lambda_{OR} > 0\), the expected number of true discoveries for any rejection rule \(\mathcal{R} \subset \mathcal{A} \subset [m]\) is given by 
$$
\mathbb{E}\left\{\sum_{i=1}^{m}\mathbb{I}\{i \in \mathcal{R}\}(1 - \mathrm{Clfdr}_{i})\right\} = \mathbb{E}\left\{\sum_{i \in \mathcal{A}}\mathbb{I}\{i \in \mathcal{R}\}(1 - \mathrm{Clfdr}_{i})\right\}.
$$ 
The proof is completed by noting that
\begin{eqnarray*}
&& \mathbb{E}\left\{\sum_{i \in [m]}\mathbb{I}\{\mathrm{Clfdr}_{i} \leq t^{OR} \wedge \widetilde{\mathrm{Clfdr}}_{i}\}(1 - \mathrm{Clfdr}_{i})\right\} \\ 
& = & \mathbb{E}\left\{\sum_{i \in \mathcal{A}}\mathbb{I}\{\mathrm{Clfdr}_{i} \leq t^{OR}\}(1 - \mathrm{Clfdr}_{i})\right\} \\ 
& \geq & \mathbb{E}\left\{\sum_{i \in \mathcal{A}}\delta_{i}'(1 - \mathrm{Clfdr}_{i})\right\} \\ 
& = & \mathbb{E}\left\{\sum_{i \in [m]}\delta_{i}'(1 - \mathrm{Clfdr}_{i})\right\}.  \qed
\end{eqnarray*}

\subsection{Proof of Theorem \ref{thm:CLAW}}

As Algorithm \ref{algo:claw_clfdr} is a special case of Algorithm \ref{algo:claw}, we only need to verify that our conformity scores are pairwise exchangeable. Consider the score function 
$$
g(t,S_{i};(\mathbf{T},\Tilde{\mathbf{T}}),\mathbf{S})=\frac{1/2-\hat{\pi}_{S_i}^{**}((\mathbf{T},\Tilde{\mathbf{T}}),\mathbf{S})}{1-\hat{\pi}_{S_i}^{**}((\mathbf{T},\Tilde{\mathbf{T}}),\mathbf{S})}\frac{\widehat{\mathrm{Clfdr}}_{i}^{**}(t;(\mathbf{T},\Tilde{\mathbf{T}}),\mathbf{S})}{1-\widehat{\mathrm{Clfdr}}^{**}(t,S_{i};(\mathbf{T},\Tilde{\mathbf{T}}),\mathbf{S})}.
$$ 
It follows from the constructions of $\hat{\pi}_{S_i}^{**}$ [\eqref{pi-ds}] and $\hat{f}_{S_i}^{**}(t)$ [\eqref{f-ds}] that
$$
  \hat{\pi}^{**}_{S_i}((\mathbf{T},\Tilde{\mathbf{T}})_{\mathrm{swap}(\mathcal{J})},\mathbf{S})=\hat{\pi}^{**}_{S_i}((\mathbf{T},\Tilde{\mathbf{T}}),\mathbf{S}) \mbox{ and }
  \hat{f}_{S_i}^{**}(t;(\mathbf{T},\Tilde{\mathbf{T}})_{\mathrm{swap}(\mathcal{J})},\mathbf{S}) = \hat{f}_{S_i}^{**}(t;(\mathbf{T},\Tilde{\mathbf{T}}),\mathbf{S})\;\; \forall \mathcal{J}\subset[m].
$$ 
Hence $\widehat{\mathrm{Clfdr}}^{**}(t,S_{i};(\mathbf{T},\Tilde{\mathbf{T}}),\mathbf{S})$, and consequently $g(t,S_{i};(\mathbf{T},\Tilde{\mathbf{T}}),\mathbf{S})$, satisfies condition (i) of Theorem \ref{thm:exch} (a). Finally, the pairwise exchangeability between conformity scores follow from  \eqref{jointexch-covariate} and Theorem \ref{thm:exch} (a), completing the proof. \qed 

\subsection{Proof of Theorem \ref{thm:multiev}}
\label{app:ev}

First recall that if $\EE\left[\sum_{i\in\mathcal{H}_{0}}e_{i}^{(k)} \right] \leq m$, then $e_{1}^{(k)},\cdots,e_{m}^{(k)}$ constitute a set of generalized e-values. Moreover, the work by \citet{wang22ev} demonstrates that by utilizing a set of generalized e-values, the e-BH procedure effectively controls the FDR at the nominal level. Note that 
\begin{equation*}
    \EE [\sum_{i\in\mathcal{H}_{0}}\bar{e}_{i} ] = \EE\left[ \sum_{i\in\mathcal{H}_{0}}\frac{1}{\sum_{k=1}^{K}v_{k}}\sum_{k=1}^{K}v_{k}e_{i}^{(k)} \right]= \frac{1}{\sum_{k=1}^{K}v_{k}} \sum_{k=1}^{K}v_{k}\EE\left[\sum_{i\in\mathcal{H}_{0}}e_{i}^{(k)} \right] \leq m.
\end{equation*}
We conclude that $\bar{e}_{1},\cdots,\bar{e}_{m}$ represent a set of generalized e-values, establishing the validity of Algorithm \ref{algo:derandom}. \qed

\section{Auxiliary Theoretical Results} 
\label{app:proof-4-app}
This section provides proofs for the theorems and properties developed in Section \ref{app-sec:method-develop}, and verification of exchangeability conditions in applications.

\subsection{Proof of Theorem \ref{appthm:exch}}

This theorem will be proved as an extension of Theorem \ref{thm:exch}, but some differences will occur because of the introduction of the additional training data set $\mathbf{T}^{tr}$ and a less stringent condition \eqref{data_pwexch}. Hence, some details may be omitted in the proofs of this section.

\begin{proof}[Proof of part (a)]
For each $i\in[m]$, consider the following vector-valued bivariate function 
$$
(T_i,\Tilde{T}_i) \mapsto \mathbf{G}_i \equiv(u_1,\cdots,u_{i-1},u_{i+1},\cdots,u_m,\Tilde{u}_1,\cdots,\Tilde{u}_{i-1},\Tilde{u}_{i+1},\cdots,\Tilde{u}_{m},T_i\vee\Tilde{T}_{i},T_i\wedge \Tilde{T}_{i}),
$$
where 
$u_j=g(T_j,S_{j};(\mathbf{T},\Tilde{\mathbf{T}}),\mathbf{T}^{tr},\mathbf{S}),\quad  \Tilde{u}_j=g(\Tilde{T}_j,S_{j};(\mathbf{T},\Tilde{\mathbf{T}}),\mathbf{T}^{tr},\mathbf{S}).$ 
Given $(\mathbf{T}_{-i},\Tilde{\mathbf{T}}_{-i},\mathbf{T}^{tr},\mathbf{S})$, the vector 
$\mathbf{G}_i$ 
is a symmetric bivariate function of $(T_i,\Tilde{T}_i)$. For $i\in\mathcal{H}_0$, we have 
$( T_i, \Tilde{T}_i, \mathbf{G}_{i} | \mathbf{T}_{-i},\Tilde{\mathbf{T}}_{-i},\mathbf{T}^{tr},\mathbf{S} ) \overset{d}{=} ( \Tilde{T}_i, T_i, \mathbf{G}_{i} | \mathbf{T}_{-i},\Tilde{\mathbf{T}}_{-i},\mathbf{T}^{tr},\mathbf{S} )$ by condition \eqref{data_pwexch} and claim \eqref{pw_sym_func}.
Equivalently, we have that
\begin{equation*}
    \Big( T_i, \Tilde{T}_i \Big| \mathbf{G}_{i} ,\mathbf{T}_{-i},\Tilde{\mathbf{T}}_{-i},\mathbf{T}^{tr},\mathbf{S} \Big) \overset{d}{=} \Big( \Tilde{T}_i, T_i \Big|\mathbf{G}_{i}, \mathbf{T}_{-i},\Tilde{\mathbf{T}}_{-i},\mathbf{T}^{tr},\mathbf{S} \Big).
\end{equation*}
Consider $g(t,S_{i};(\mathbf{T},\Tilde{\mathbf{T}}),\mathbf{T}^{tr},\mathbf{S})=g(t,S_{i};         \{T_i,\Tilde{T}_i\},(\mathbf{T}_{-i},\Tilde{\mathbf{T}}_{-i}),\mathbf{T}^{tr},\mathbf{S}),$ which is nonrandom with respect to $ \sigma(\{T_i,\Tilde{T}_i\},(\mathbf{T}_{-i},\Tilde{\mathbf{T}}_{-i}),\mathbf{T}^{tr},\mathbf{S}) \subset \sigma(\mathbf{G}_{i}, \mathbf{T}_{-i},\Tilde{\mathbf{T}}_{-i},\mathbf{T}^{tr},\mathbf{S})$, we have 
$$
\Big( u_i, \Tilde{u}_i \Big| \mathbf{G}_{i},\mathbf{T}_{-i},\Tilde{\mathbf{T}}_{-i},\mathbf{T}^{tr},\mathbf{S} \Big) \overset{d}{=} \Big( \Tilde{u}_i, u_i \Big|\mathbf{G}_{i},\mathbf{T}_{-i},\Tilde{\mathbf{T}}_{-i},\mathbf{T}^{tr},\mathbf{S} \Big),
$$
for $i\in\mathcal{H}_0$. The desired result follows  by integrating out $(\mathbf{T}_{-i},\Tilde{\mathbf{T}}_{-i},\mathbf{T}^{tr},\mathbf{S})$ and $(T_i\vee\Tilde{T}_{i},T_i\wedge \Tilde{T}_{i})$.
\end{proof}

\begin{proof}[Proof of part (b).]
We first modify the definitions of the following key quantities:
\begin{equation*}
\begin{split}
    \mathbf{A} &= (A_1,\cdots,A_{m+|\mathcal{H}_0|})=(\Tilde{T}_1,\cdots,\Tilde{T}_{m},T_i:i\in\mathcal{H}_0), \\
    \mathbf{B} &= (T_i:i\notin\mathcal{H}_0), \\
    \mathcal{C} &= (\mathbf{T}^{tr},\{T_1,\cdots,T_m,\Tilde{T}_1,\cdots,\Tilde{T}_m\}),\\
    U_i &= g(A_i;(\mathbf{T},\Tilde{\mathbf{T}}),\mathbf{T}^{tr}) = g(A_i;\mathcal{C}),\quad i\in\{1,\cdots,m+|\mathcal{H}_0|\}.
\end{split}
\end{equation*}
The rest of the proof follows the same lines as the proof of Theorem \ref{thm:exch} (b), and is omitted. 
\end{proof}

\subsection{Verification for Properties \ref{app-prop:confpv}-\ref{app-prop:augm}}

\begin{proof}[Verification of Property \ref{app-prop:confpv}.]
 Consider conformal p-values \eqref{pu-conf-pv} calculated through
    \begin{equation}\label{app-proof-pvfunc}
        \hat{p}(t;\mathbf{T},\Tilde{\mathbf{T}},\mathbf{T}^{tr})=\frac{1+|\{k\in\mathcal{D}^{tr}_{1}:s(T_{k}^0;\mathbf{T},\Tilde{\mathbf{T}},\mathbf{T}^{tr1},\mathbf{T}^{tr2})\leq s(t;\mathbf{T},\Tilde{\mathbf{T}},\mathbf{T}^{tr1},\mathbf{T}^{tr2})\}| }{1+|\mathcal{D}^{tr}_1|},
    \end{equation}
where $\mathbf{T}^{tr}=\mathbf{T}^{tr1}\cup\mathbf{T}^{tr2}$, and $s(t)$ satisfies
    $$
    s(t;\mathbf{T},\Tilde{\mathbf{T}},\mathbf{T}^{tr1},\mathbf{T}^{tr2})=s(t;(\mathbf{T},\Tilde{\mathbf{T}},\mathbf{T}^{tr1})_\Pi,\mathbf{T}^{tr2}).
    $$
This implies that $\hat{p}(t;\mathbf{T},\Tilde{\mathbf{T}},\mathbf{T}^{tr})$ fulfills the permutation-invariance condition \eqref{app:principle-jt} with respect to $(\mathbf{T},\Tilde{\mathbf{T}},\mathbf{T}^{tr1})$ and its generalized swapping-invariance condition \eqref{app:principle-pw} with respect to $(\mathbf{T},\Tilde{\mathbf{T}})$. According to Theorem \ref{appthm:exch}, part (a) of Property \ref{app-prop:confpv} follows from \eqref{jointexch}, and part (b) follows from \eqref{data_pwexch}. 
\end{proof}

\smallskip

\begin{proof}[Verification of Property \ref{app-prop:group}.]
    According to Theorem \ref{appthm:exch} (a), we only need to show that $\hat{R}(t,k)$ is swapping-invariant with respect to $(\mathbf{T},\Tilde{\mathbf{T}})$. As $\hat{R}(t,k)$ is constructed via $\hat{\pi}^{**}_k$ [\eqref{pu-p**}] and $\hat{r}(t,k)$ [\eqref{app-pu-group-dr}], we only need to verify condition \eqref{app:principle-pw} for these two estimators. First, since the conformal p-value function \eqref{app-proof-pvfunc} is permutation-invariant with respect to $(\mathbf{T},\Tilde{\mathbf{T}})$, we have
$\hat{\pi}^{**}_{k}((\mathbf{T},\Tilde{\mathbf{T}})_{\mathrm{swap}(\mathcal{J})},\mathbf{T}^{tr},\mathbf{S})=\hat{\pi}^{**}_{k}((\mathbf{T},\Tilde{\mathbf{T}}),\mathbf{T}^{tr},\mathbf{S})$
for any $\mathcal{J}\subset[m]$, establishing \eqref{app:principle-pw} for $\hat{\pi}^{**}_{k}$. Moreover, note that $\hat{r}(t,k)=\hat{r}(t,k;\cup_{i: S_i=k}\{T_i, \Tilde T_i\}, \mathbf{T}^{tr})$
is determined by $\mathbf{S}$, $\mathbf{T}^{tr}$ and the union of the unordered sets $\{T_i,\Tilde{T}_i\}$ for $S_i=k$, establishing \eqref{app:principle-pw} for $\hat{r}(t,k)$. 
\end{proof}

\smallskip

\begin{proof}[Verification of Property \ref{app-prop:augm}. ]
    According to Theorem \ref{appthm:exch} (a), we only need to show the score function $\hat{R}(t,k)$ is swapping-invariant with respect to $(\mathbf{T},\Tilde{\mathbf{T}})$. As $\hat{R}(t,k)$ is derived by $\hat{\pi}^{**}_k$ in \eqref{pu-p**} and $\hat{r}(t,k)$ in \eqref{app-pu-augm-dr}, we need to justify \eqref{app:principle-pw} for these two estimators. In the proof of Property \ref{app-prop:group}, we have verified \eqref{app:principle-pw} for $\hat{\pi}^{**}_k$ in \eqref{pu-p**}. Now we turn to the density ratio estimator \eqref{app-pu-augm-dr}. Since 
$$
\hat{r}(t,s) = \hat{r}(t,s;\mathbf{T}^+,\tilde{\mathbf{T}}^+,\mathbf{T}^{tr+}  )= \hat{r}(t,s; \cup_{i\in[m]}\{T_i^+,\Tilde{T}_i^+\},\{T_i^{tr+}:i\in[m]\} )
$$ is determined by the unordered sets $\cup_{i\in[m]}\{T_i^+,\Tilde{T}_i^+\}$ and $\{T_i^{tr+}:i\in[m]\}$, we have that
$$\hat{r}(t,s;(\mathbf{T}^+,\tilde{\mathbf{T}}^+)_{\rm{swap}(\mathcal{J})},\mathbf{T}^{tr+}) = \hat{r}(t,s;\mathbf{T}^+,\tilde{\mathbf{T}}^+,\mathbf{T}^{tr+}  ),$$
for any $\mathcal{J}\subset[m]$. By the construction of the augmented data, swapping $T_i^+=(T_i,S_i)$ and $\Tilde{T}_i^+=(\Tilde{T}_i,S_i)$ is equivalent to swapping $T_i$ and $\Tilde{T}_i$ given $\mathbf{S}$, implying that \eqref{app:principle-pw} holds for $\hat{r}(t,k)$ in \eqref{app-pu-augm-dr}.
\end{proof}

\section{Further details for comparison with existing work}
\label{app:discuss}

\subsection{CLAW versus PLIS}
\label{subsec:plis}

The PLIS procedure, proposed by \citet{zhao2023plis}, offers an assumption-lean approach for multiple testing in structured probabilistic models such as the hidden Markov models (HMM). PLIS begins by constructing baseline data and subsequently computes the conformity scores using a user-specified working model. PLIS guarantees finite-sample FDR control under the pairwise exchangeability condition, and exhibits substantial power improvement when the underlying data-generating process can be well represented by the chosen working model. 

In this section, we begin by comparing the theoretical frameworks of CLAW and PLIS. Subsequently, we present numerical results to illustrate the strengths and limitations of both methods across various practical scenarios.

\subsubsection{Theoretical comparisons}

We outline several significant distinctions between PLIS and CLAW below. 

\begin{itemize}

\item Firstly, the two methods serve different purposes. PLIS is designed to integrate the \emph{dependency structure} among hidden states, while CLAW aims to capture the local \emph{smoothness structure} inherent in the covariates. As we will demonstrate shortly, each method has its own merits and limitations, making them suitable for different scenarios.

\item Secondly, PLIS requires specifying a class of \emph{parametric} working models, such as the Hidden Markov Model (HMM), to effectively capture the underlying dependency structure. However, this model specification may be impractical within our problem framework, where structural information is encoded by a covariate sequence $(S_i)_{i=1}^m$. In contrast, CLAW adopts a \emph{nonparametric} approach to constructing a bivariate score function, eliminating the need for specifying a parametric working model, thereby offering a more flexible framework during the modeling phase. 

\item Finally, the methodologies for constructing conformity scores and the underlying theory of pairwise exchangeability differ between PLIS and CLAW. PLIS primarily focuses on the construction of baseline data, which may systematically deviate from the optimal rule. In contrast, CLAW utilizes novel techniques across three key steps: (a) designing covariate-adaptive weights, (b) learning swapping-invariant score functions, and (c) implementing monotone transformations. These innovative techniques enable the CLAW procedure to effectively emulate the optimal rule. 
\end{itemize}

\subsubsection{Numerical comparisons}

We conduct a numerical experiment to illustrate the distinct advantages of CLAW and PLIS. The results are presented in Figure \ref{fig:CLAW_vs_PLIS}. In our experiment, PLIS is implemented using a parametric HMM as its working model, while the working model for CLAW is nonparametric. The implementation details of PLIS and CLAW can be found in \citet{zhao2023plis} and Section \ref{subsec:spatial}, respectively. 

The data are generated according to the following model:
\begin{equation*}
    T_{i}|(\theta_{i},S_{i}=s) \overset{ind.}{\sim} (1-\theta_{i})\mathcal{N}(0,1)+\theta_{i}\mathcal{N}(0,\mu_{s}),\quad i=1,\cdots,3000,
\end{equation*}
where $\theta_{i}=0$ ($1$) denotes $H_i$ is true (false). We consider two settings in our illustrations.
\begin{enumerate}[I.]
    \item HMM setting: The hidden states $(\theta_i)_{i=1}^{3000}$ form a binary Markov chain. The transition probabilities are $a_{00}=0.95$ and $a_{11}=0.5$, where $a_{ij}=\PP(\theta_{t+1}=j|\theta_{t}=i)$. The signal amplitude $\mu_{s}=\mu$ does not change over $s$. 
    \item Covariate-adaptive model setting: $\theta_{i}|(S_{i}=s)\overset{ind.}{\sim} Bin(1,\pi_{s})$, where $\pi_{s}=0.4(1+\sin(0.2s))$ for $s\in[201,500]\cup[801,1100]\cup[1501,1800]\cup[2101,2400]$ and $\pi_{s}=0.02$ otherwise. The signal amplitude $\mu_{s}=\mu+0.2\sin(0.6s)$ varies as a function of $s$.
\end{enumerate}
In both settings, the null samples are generated as i.i.d. $\mathcal{N}(0,1)$ variables to implement both PLIS and CLAW. 

\begin{figure}[!htbp]
    \centering
    \includegraphics[width=0.9\linewidth,height=0.5\linewidth]{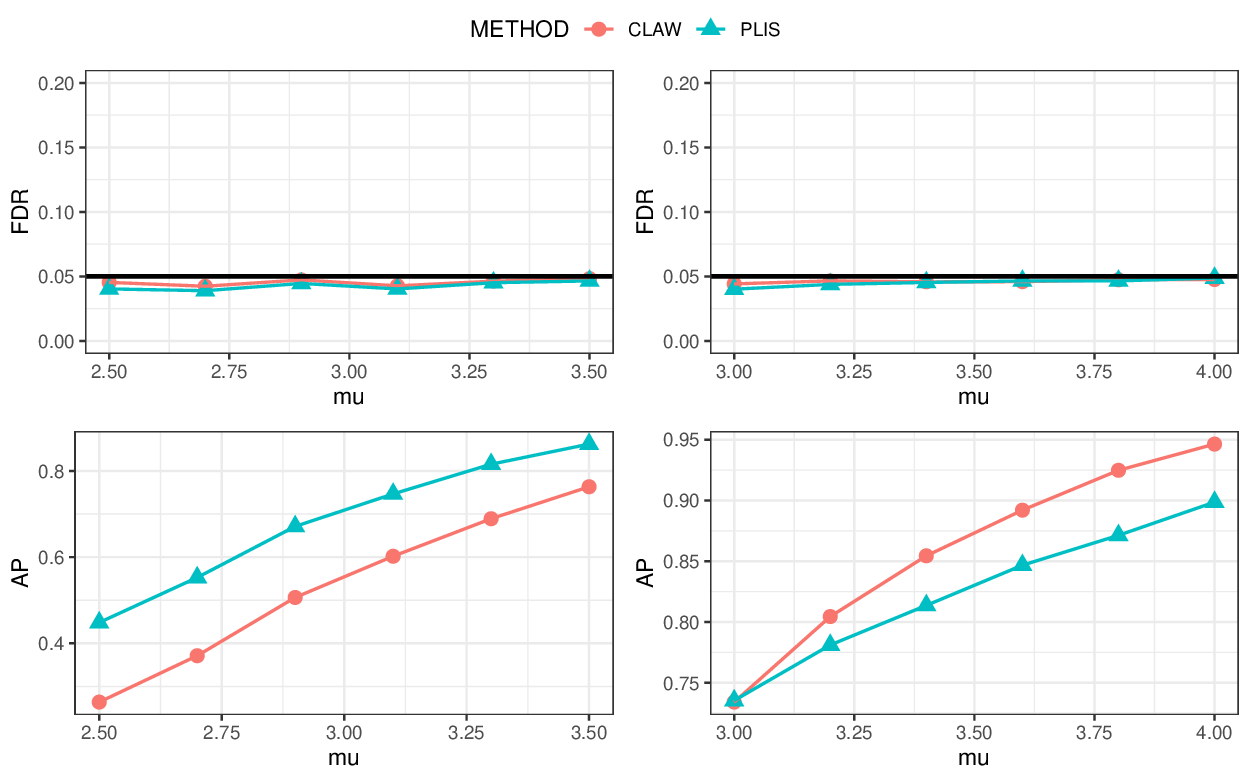}
    \caption{Comparison for CLAW and PLIS at FDR level $\alpha=0.05$. The left column shows the results when the underlying data generation process is an HMM (setting I), while the right column considers a generic covariate-adaptive model that deviates from the HMM (setting II).}\label{fig:CLAW_vs_PLIS}
\end{figure}

Figure \ref{fig:CLAW_vs_PLIS} demonstrates that both CLAW and PLIS control the FDR at the nominal level, though their powers vary significantly. When the underlying model is an HMM, PLIS exhibits superior performance compared to CLAW. Conversely, in the context of a generic covariate-adaptive model that diverges from the HMM, CLAW demonstrates a clear advantage.

\subsection{CLAW versus conformal methods}
\label{app:cf}

In Section \ref{app:relation}, we have illustrated how the decision process \eqref{confq} can be derived by modifying the CBH procedure \eqref{def:conf_pv_bh}. We now provide additional illustrations to explain why CBH exhibits conservativeness and how  CLAW can overcome this issue. First, we can rewrite the numerator of \eqref{def:conf_pv_bh} as:
$$ 
1+\sum_{j\in\mathcal{H}_{0}}\II\{\Tilde{u}\leq t\}+\sum_{j\in[m]\setminus\mathcal{H}_{0}}\II\{\Tilde{u}\leq t\}.
$$
In this expression, the term $\sum_{j\in\mathcal{H}_{0}}\II\{\Tilde{u}\leq t\}$ represents the estimated number of false discoveries when the threshold is $t$. The term ``+1'' is necessary for establishing martingale properties and ensuring the super-uniformity of conformal p-values \eqref{def:conf_pv} under the null. However, the term $\sum_{j\in[m]\setminus\mathcal{H}_{0}}\II\{\Tilde{u}\leq t\}$ is redundant and contributes to the conservativeness, which can be effectively addressed by introducing the new decision process \eqref{confq}. For $j\in[m]\setminus\mathcal{H}_{0}$, $u_j$ tends to be smaller than $\Tilde{u}_j$ with high probability. This is because $u_j$ is calculated based on a non-null sample $T_j$, while $\Tilde{u}_j$ is calculated based on a null sample $\Tilde{T}_j$. Consequently, when $F_0$ and $F_{1s}$ are well distinguished, we have:
$$
\sum_{j\in[m]\setminus\mathcal{H}_{0}}\II\{\Tilde{u}_{j} \leq t\wedge u_{j}\}\approx 0.
$$
This approximation holds with high probability. Finally, regarding the denominator, we have:
$$
\sum_{j\in[m]\setminus\mathcal{H}_{0}}\II\{{u}_{j} \leq t\wedge \Tilde{u}_{j}\}\approx \sum_{j\in[m]\setminus\mathcal{H}_{0}}\II\{{u}_{j} \leq t\}
$$
This approximation holds with high probability for moderate $t$. This indicates that we will not ``lose'' too many correct counts compared to the unmodified method \eqref{def:conf_pv_bh}, contributing to the effectiveness of \eqref{confq}.

\subsection{Comparison with Storey-BH type methods}
\label{app:storey}

In Section \ref{subsec:clfdr}, we proposed an estimator \eqref{pi-ds} to assess the signal's proportion $\pi_{S_i}$ in the EB working model \eqref{model:mixture}. As illustrated in Section \ref{app:subsub-proportion}, \eqref{pi-ds} serves as a conformalized version of the locally adaptive estimator $\hat{\pi}_{S_i}^{*}$ in \eqref{conventional_est}, which generalizes the standard Storey's estimator \citep{storey02}:
\begin{equation}\label{pistorey}
    \hat{\pi}^{Storey}=1-\frac{\sum_{j=1}^{m}\II\{p(T_{j})>\lambda\}}{1-\lambda},
\end{equation}
for evaluating the sparsity level that varies depending on $S_i$.

Next, we would like to summarize the differences and connections between \eqref{pistorey} and the proposed estimator \eqref{pi-ds} in the following three points:
\begin{enumerate}
    \item \emph{The estimands corresponding to the two estimators are different.} Unlike Storey's estimator, which is concerned with the \emph{global sparsity parameter} \(\pi\), our estimator (18) focuses on the \emph{local sparsity level} \(\pi_s\), which can depend on covariate values. Furthermore, the methodologies for estimating these two parameters differ substantially: Storey's estimator, as utilized in AdaDetect, relies solely on the conformal p-values derived from the test samples. In contrast, our estimator (18) employs a more sophisticated screening scheme and leverages p-values from both the test and calibration samples.

\item \emph{The two estimators play different roles in FDR analysis.} An FDR procedure typically involves two critical steps: ranking and thresholding. The Storey estimator, employed in conjunction with AdaDetect, provides a global correction that adjusts the nominal FDR level \(\alpha\) to \(\alpha/(1-\hat{\pi})\). This estimator operates solely within the thresholding step and has no influence on the ranking process. Conversely, our estimator (18) is instrumental in constructing conformity scores, as it utilizes side information to improve ranking through covariate-adaptive weights, thereby enhancing the overall efficiency of the FDR analysis. 

\item \emph{The two estimators operate within distinct classes of base algorithms.} The BH-type methods, exemplified by counting knockoffs \citep{weinstein17counting} and AdaDetect \citep{marandon22mlfdr}, achieve the nominal FDR level through Storey's correction. In contrast, BC-type methods, including knockoff filters and CLAW, are capable of attaining the nominal FDR level adaptively -- \textbf{without relying on Storey's correction} -- when the signals are sufficiently strong. This phenomenon was initially noted in Appendix B of \citet{barber15knockoff}. The FDR level of CLAW can be very close to the nominal level in many settings; this capability is attributable to the adaptivity of the BC algorithm (instead of Storey's correction). As we mentioned in the previous point, the contribution of our estimator (18) lies in enhancing efficiency in the \textbf{ranking step} through covariate-adaptive weights, whereas the Storey estimator's role in AdaDetect is to assist in achieving the nominal FDR level in the \textbf{thresholding step} by adjusting the target FDR level.
\end{enumerate}

\subsection{A numerical study comparing CLAW, BH, BH-Storey, AdaDetect and AdaDetect-Storey}

Next, we present a comparison of the numerical performance of CLAW, BH, BH-Storey, AdaDetect and AdaDetect-Storey. The data are generated from the following model  
    \begin{equation*}
        T_{i}|(\theta_{i},S_{i}=s) \sim (1-\theta_{i})\mathcal{N}(0,1)+\theta_{i}F_{1s},\quad i=1,\cdots,m,
    \end{equation*}
    where $\pi_{s}=\PP(\theta_i=1|S_i=s)$, and the calibration samples are i.i.d. $\mathcal{N}(0,1)$ variables. The following settings are considered:
    \begin{enumerate} 
        \item $m=4500$. For $i=1,\cdots,3000$, $S_i=1$, $F_{1s}=\mathcal{N}(\mu,1)$, $\pi_{s}=0.2$; For $i=3001,\cdots,4500$, $S_i=2$, $F_{1s}=\mathcal{N}(-2,0.5^2)$, $\pi_{s}=0.1$. Let $\mu$ vary.
        \item $m=3000$. $S_i=i$; $F_{1s}=F_{1}=\mathcal{N}(\mu,1)$; $\pi_{s}=0.6$ for $s\in[201,350]\cup[1501,1650]$, $\pi_{s}=0.3$ for $s\in[801,1000]\cup[2101,2300]$, and $\pi=0.02$ otherwise. Let $\mu$ vary.
        \item $m=4500$. For $i=1,\cdots,3000$, $S_i=1$, $F_{1s}=\mathcal{N}(3.6,1.5^2)$, $\pi_{s}=\pi$; For $i=3001,\cdots,4500$, $S_i=2$, $F_{1s}=\mathcal{N}(-2.5,1)$, $\pi_{s}=0.1$. Let $\pi$ vary.
    \end{enumerate}

    \begin{figure}[!htbp]
        \centering
        \includegraphics[width=0.9\linewidth,height=0.5\linewidth]{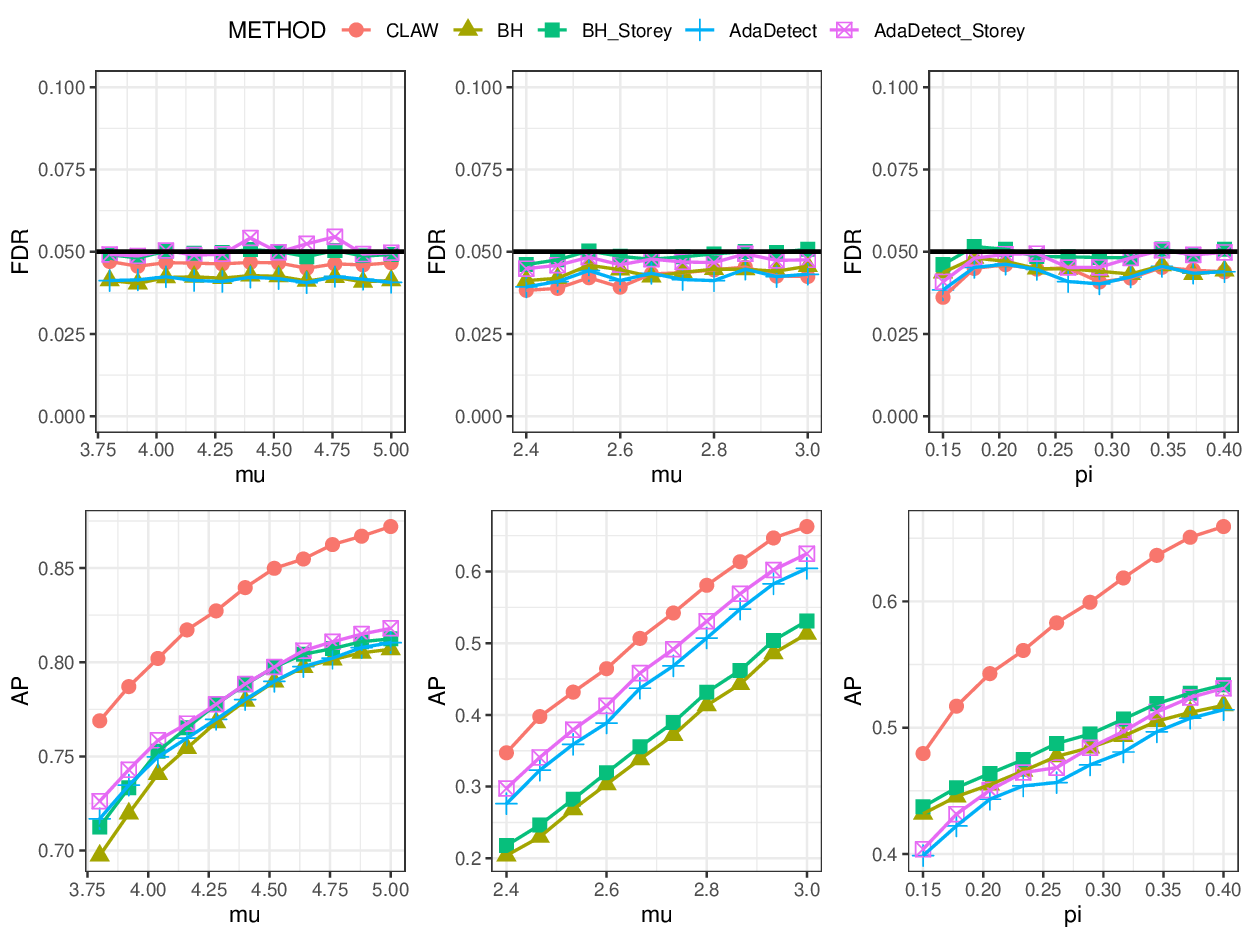}
        \caption{FDR and AP comparison for CLAW, BH, BH-Storey, AdaDetect and AdaDetect-Storey methods. The left, middle and right columns are corresponding to settings 1, 2 and 3, respectively.}\label{fig:adad-storey}
    \end{figure}
    
    The results are provided in Figure \ref{fig:adad-storey}. With Storey's correction, the BH-type methods including conventional BH and AdaDetect show some power improvements and their empirical FDR levels
    are more closer to $\alpha=0.05$ compared to CLAW. However, their power improvements are negligible because the rankings of p-values or density ratios
    (which is used to construct conformal p-values by AdaDetect) are suboptimal when side information is helpful in inference.

    In general, there are two basic step in all FDR procedures: ranking and thresholding. Although both steps can contribute to power improvement,
    constructing better ranked statistics or scores is usually more effective than simply adjusting the threshold to achieve the nominal FDR level. 
    While BH-typed methods (such as counting knockoffs \citep{weinstein17counting} and AdaDetect) can achieve the nominal FDR level via Storey's correction, the BC type methods, such as Knockoffs \citep{barber15knockoff} and 
    CLAW, can achieve it adaptively without such corrections if the signals are strong enough (see Section \ref{app:cf} for further discussion). 
    The power improvement of CLAW lies in both building more efficient scores and the adaptivity in achieving the nominal FDR level. 

    To better illustrate this point, especially the adaptivity of CLAW in achieving the nominal FDR level, we consider the following settings slightly different from those in Figure \ref{fig:adad-storey}:
    \begin{enumerate}[1']
        \item $m=4500$. For $i=1,\cdots,3000$, $S_i=1$, $F_{1s}=\mathcal{N}(\mu,1)$, $\pi_{s}=0.5$; For $i=3001,\cdots,4500$, $S_i=2$, $F_{1s}=\mathcal{N}(-2,0.5^2)$, $\pi_{s}=0.1$. Let $\mu$ vary.
        \item $m=3000$. $S_i=i$; $F_{1s}=F_{1}=\mathcal{N}(\mu,1)$; $\pi_{s}=0.9$ for $s\in[201,350]\cup[1501,1650]$, $\pi_{s}=0.6$ for $s\in[801,1000]\cup[2101,2300]$, and $\pi=0.02$ otherwise. Let $\mu$ vary.
        \item $m=4500$. For $i=1,\cdots,3000$, $S_i=1$, $F_{1s}=\mathcal{N}(3.6,1)$, $\pi_{s}=\pi$; For $i=3001,\cdots,4500$, $S_i=2$, $F_{1s}=\mathcal{N}(-2,0.5^2)$, $\pi_{s}=0.1$. Let $\pi$ vary.
    \end{enumerate}

    The simulation results are summarized in Figure \ref{fig:adaptivity}.
    In the first two columns, we can see that the FDR of BH-Storey and AdaDetect-Storey remains close to the nominal level, 
    while BH and AdaDetect exhibit a conservative behavior. 
    Our proposed method, CLAW, demonstrates conservativeness when $\mu$ is small but adaptively achieves the nominal FDR level as $\mu$ increases.

    In the third column, as the signals' proportion becomes larger, the conservativeness of BH and AdaDetect becomes increasingly prominent, 
    eventually leading to a decrease in power of AdaDetect compared to BH-Storey. Remarkably, our proposed method, CLAW, consistently outperforms other methods in all situations. 
    It showcases two important advantages: first, CLAW addresses the information loss issue in both BH and AdaDetect (and their null proportion adaptive versions) 
    by using the efficient ranking derived from the scores integrating side information; 
    second, CLAW alleviates the conservativeness of BH-based methods via adaptively achieving the nominal FDR level.

    \begin{figure}[!htbp]
        \centering
        \includegraphics[width=0.9\linewidth,height=0.5\linewidth]{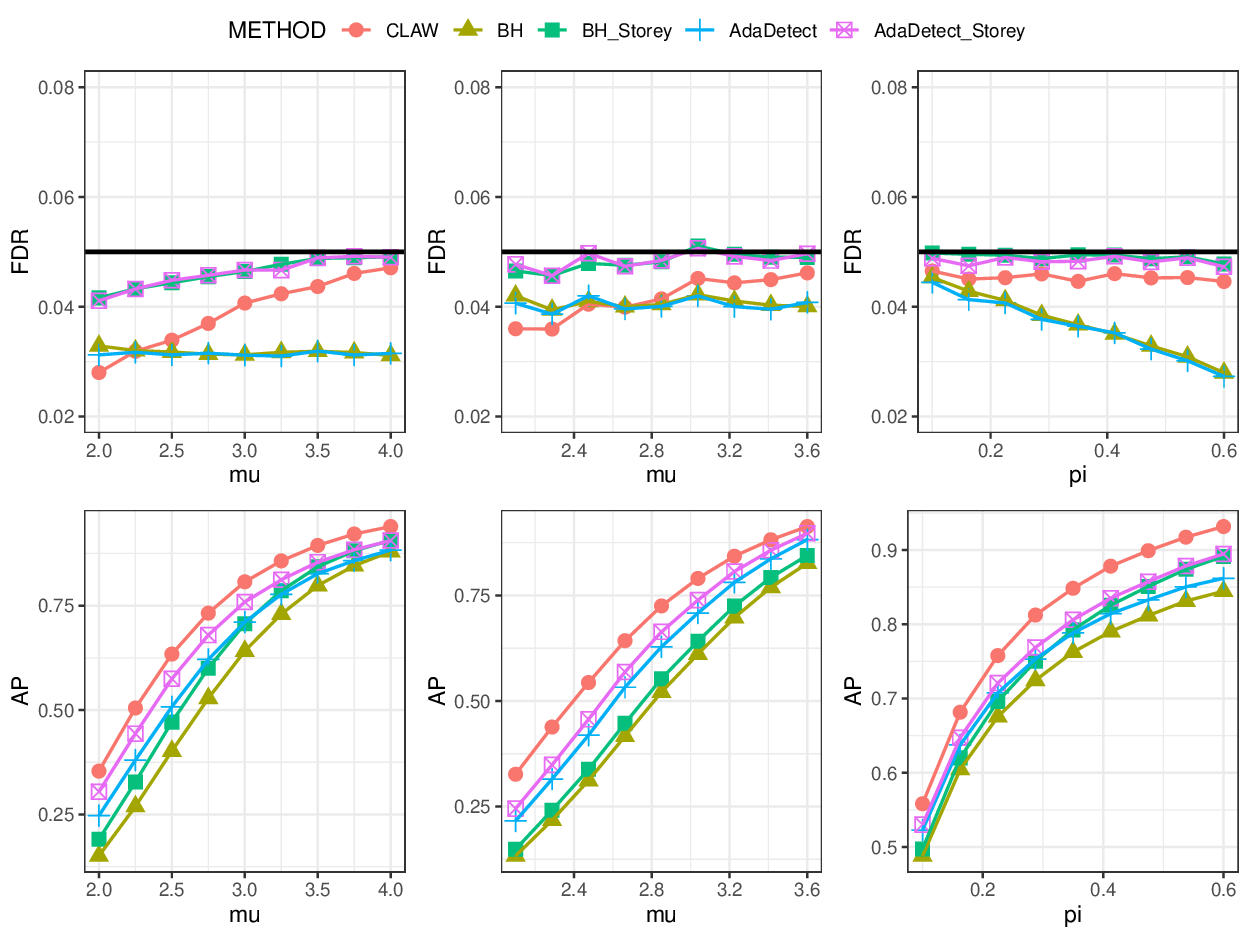}
        \caption{FDR and AP comparison or CLAW, BH, BH-Storey, AdaDetect and AdaDetect-Storey
        methods. The left, middle and right columns are corresponding to settings 1', 2' and 3', respectively.}\label{fig:adaptivity}
    \end{figure}

\subsection{CLAW and BC-based methods}\label{app:kn}

While the mirror process \eqref{confq} can be motivated from a conformal BH perspective, it can also be conceptualized as a symmetrized inference procedure closely related to the selective SeqStep+ algorithm \citep{barber15knockoff}, which we will refer to here as the Barber-Candes (BC) algorithm. 
In this section, we first prove that the FDP process is equivalent to the Selective SeqStep+ algorithm in the form of with a carefully designed anti-symmetric statistic suggested by an insightful referee, and then delve into the connections and differences between CLAW and existing BC-cased multiple testing methods, including knockoff filters for variable selection and other multiple testing procedures. 

\subsubsection{The proof that CLAW is a BC-type algorithm}

\begin{proof} To start with, define the following class of anti-symmetric statistics:
\begin{equation}\label{S_xy}
       T^S_{j}= T^S(u_{j},\tilde{u}_{j}) = \mathrm{sign}(\tilde{u}_{j}-u_{j})\cdot [g(u_j)\vee g(\tilde{u}_{j})],\quad \forall j\in[m],
    \end{equation}
where $g(\cdot):\mathbb{R}_{\geq0}\to\mathbb{R}_{\geq0}$ is a non-random strictly decreasing function. Consider the following mirror process
\begin{equation*}
    Q^{S}(t)=\frac{1+\sum_{j\in[m]} \II\{T^S_{j}\leq -t\}  }{(\sum_{j\in[m]} \II\{T^S_{j}\geq t\})\vee1}, \quad t>0.
\end{equation*}
Define $\tau'= \inf\{t\in\mathcal{T}^S:Q^{S}(t)\leq \alpha\}$, where $\mathcal{T}^S=\{|T^S_{j}|:j\in[m]\}$.
Consider a decision rule ${\pmb\delta}^\prime=\{\delta_j^\prime: j\in[m]\}$, where $\delta_j^\prime=\II\{T^S_{j}\geq \tau'\}$, then ${\pmb\delta}^\prime$ is equivalent to $\pmb{\delta}=\{\delta_j: j\in[m]\}$ output by Algorithm 1.

For a non-random strictly decreasing function $g(\cdot)$ defined on $\mathbb{R}_{\geq0}\to\mathbb{R}_{\geq0}$, the value $g(u_{i})$ can be interpreted as a non-conformity score, with a higher value indicating stronger evidence against  $H_{i}$.  As such $g(\cdot)$ is bijective, we have 
\begin{eqnarray*}
\mathcal{G}^{r} & = & \{i:u_{i}<\tilde{u}_{i}\}=\{i:g(u_{i})>g(\tilde{u}_{i})\}, \mbox{ and} \\
\mathcal{G}^{c} & = & \{i:\tilde{u}_{i}<u_i\}=\{i:g(\tilde{u}_{i})>g(u_i)\}.
\end{eqnarray*}
Consider the decision $\delta_{i}$ output by Algorithm 1. We have 
$\delta_{i}=\II\{u_{i}\leq\tau\}\II\{i\in\mathcal{G}^r\} =\II\{g(u_{i})\geq\tau'\}\II\{g(u_{i})>g(\tilde{u}_{i})\}$, where
$$
\tau'=\inf\left\{t\in\{g(u_i)\}_{i\in\mathcal{G}^{r}}\cup\{g(\tilde{u}_i)\}_{i\in\mathcal{G}^{c}}: \frac{1+\sum_{j\in[m]} \II\{g(\tilde{u}_j)\geq t\}\II\{\tilde{u}_j<u_j\} }{(\sum_{j\in[m]} \II\{g(u_j)\geq t\}\II\{u_j<\tilde{u}_j\})\vee1 } \leq \alpha\right\}.
$$
This holds because $g(\cdot)$ is strictly decreasing, and that the function 
$$
\frac{1+\sum_{j\in[m]} \II\{g(\tilde{u}_j)\geq t\}\II\{\tilde{u}_j<u_j\} }{(\sum_{j\in[m]} \II\{g(u_j)\geq t\}\II\{u_j<\tilde{u}_j\})\vee1 }
$$ 
only jumps at points within the set $\{g(u_i)\}_{i\in\mathcal{G}^{r}}\cup\{g(\tilde{u}_i)\}_{i\in\mathcal{G}^{c}}$. By the definition of $T^S_{j}$ in \eqref{S_xy}, we have that, for any $t>0$:
\begin{eqnarray*}
T^S_{j}\geq t & \Longleftrightarrow  & u_j<\tilde{u}_j  \mbox{ and }  g(u_j)\geq t, \\
T^S_{j}\leq -t & \Longleftrightarrow & \tilde{u}_j<u_j  \mbox{ and } g(\tilde{u}_j)\geq t.
\end{eqnarray*}
It follows that
$$
\frac{1+\sum_{j\in[m]} \II\{g(\tilde{u}_j)\geq t\}\II\{\tilde{u}_j<u_j\} }{(\sum_{j\in[m]} \II\{g(u_j)\geq t\}\II\{u_j<\tilde{u}_j\})\vee1 } = \frac{1+\sum_{j\in[m]} \II\{T^S_{j}\leq -t\}  }{(\sum_{j\in[m]} \II\{T^S_{j}\geq t\})\vee1} = Q^{S}(t), \quad t>0.
$$
It is easy to see that $\mathcal{T}^S=\{g(u_i)\}_{i\in\mathcal{G}^{r}}\cup\{g(\tilde{u}_i)\}_{i\in\mathcal{G}^{c}}$. Therefore we have $\tau'=\inf\{t\in\mathcal{T}^S:Q^{S}(t)\leq\alpha\}$ and
$\delta_{i}=\II\{u_i\leq\tau\}\II\{i\in\mathcal{G}^r\}=\II\{T^S_{i}\geq\tau'\}=\delta_i^\prime,$ completing the proof. 
\end{proof}

This connection has also been elucidated clearly in the proof of Lemma \ref{lemma1}. However, the CLAW procedure distinguishes itself from the knockoff filters, specifically designed for variable selection in regression problems, in several important ways. 

\subsubsection{Differences between CLAW and knockoff filters}

As we acknowledged in the main text, CLAW draws inspiration from the techniques employed in knockoff filters for variable selection problems \citep{barber15knockoff, ren23knockoff}, as well as the empirical Bayes approach utilized in AdaDetect \citep{marandon22mlfdr} to enhance the conformal Benjamini-Hochberg (BH) algorithm. However, we would like to emphasize several key differences between CLAW and knockoff filters:  

\begin{enumerate}[(a)]
        \item \emph{{The problem setups are different.}} The knockoff filter serves as a variable selection technique within regression frameworks, specifically designed to test for conditional independence. In this context, the null hypothesis is formulated as \(H_i: Y \perp X_i | \mathbf{X}_{-i}\), where \(Y\) represents the response variable, \(X_i\) is the predictor of interest, and \(\mathbf{X}_{-i}\) denotes the remaining predictors. In contrast, CLAW is focused on detecting outliers that deviate from the ``norm state'' rather than on selecting important variables. The null hypothesis in this setting is specified as \(H_i: T_i \sim F_0\), where \(T_i\) represents the test data or summary statistic, and \(F_0\) denotes a known null distribution (classical setup) or an unknown distribution derived from a given set of null samples (semi-supervised setup). Notably, the CLAW procedure does not involve a response variable \(Y\), as its primary objective is to assess deviations from the expected patterns rather than establishing a relationship with $Y$.
                
        \item \emph{{The underpinning assumptions are different.}} 
        Although the high-level concepts of pairwise exchangeability are similar in the two approaches, the fundamentally differing problem setups -- specifically, with and without response -- give rise to distinct assumptions necessary for each method. Concretely, the knockoff filters \citep{barber15knockoff,ren23knockoff} impose the exchangeability condition on all predictors:
\begin{equation}\label{exc-knock}
(X_j,\tilde{X}_j,\mathbf{X}_{-j},\tilde{\mathbf{X}}_{-j}) \overset{d}{=} (\tilde{X}_j,X_j,\mathbf{X}_{-j},\tilde{\mathbf{X}}_{-j}), \quad \forall j \in [m].
\end{equation}
The knockoff filter fails to control the FDR if the condition in \eqref{exc-knock} does not hold for any \(j_0\), including the case of \(j_0 \notin \mathcal{H}_0\), where \(\mathcal{H}_0\) denotes the index set of all null hypotheses. Therefore, constructing valid knockoff variables that satisfy this exchangeability condition is a pivotal step in methodological developments. In contrast, CLAW operates under a different notion of pairwise exchangeability \eqref{data_pwexch}:
\begin{equation*}
\left( (\mathbf{T},\tilde{\mathbf{T}})_{\mathrm{swap}(\mathcal{J})} \big| \mathbf{T}^{tr},\mathbf{S} \right) \overset{d}{=} \left( \mathbf{T},\tilde{\mathbf{T}} \big| \mathbf{T}^{tr},\mathbf{S} \right), \quad \forall \mathcal{J} \subset \mathcal{H}_{0},
\end{equation*}
The exchangeability condition is only required to hold on $\mathcal{H}_{0}$, rather than on all \(j \in [m]\). 

        \item \emph{{The algorithmic structures and operations are different.}} In \citet{ren23knockoff}, an adaptive strategy is employed that aligns with the AdaPT framework \citep{lei18adapt}, where varying p-value thresholds are established along the ordered sequence to approximate the oracle rule. This approach leverages side information to sequentially update both the thresholds and the masked data. In contrast, the core strategy of CLAW, which operates within the conformal framework, involves utilizing calibration samples, test samples, and side information to construct the most powerful conformity scores. While in \citet{ren23knockoff} the scores \(\{(Z_j, \tilde{Z}_j)\}_{j=1}^{p}\) are fixed and do not adapt to the side information -- relying solely on adaptively adjusted thresholds for oracle approximation -- CLAW employs a universal threshold across all conformity scores and integrates side information directly into the calculation of these scores to enhance the approximation to the oracle. 

        \item \emph{{The methodological focuses are different.}} The differing problem setups have led to variations in the key areas of methodological development in the two approaches. For knockoff methods, the primary methodological challenge is to construct knockoff copies that fulfill the pairwise exchangeability condition \eqref{exc-knock}. When this condition is met, the test statistics \(\{W_j\}_{j=1}^{p}\), derived from symmetric fitting algorithms, can control the FDR when utilized within the knockoff filter. In contrast, CLAW is predicated on the exchangeability condition for null samples [\eqref{data_pwexch}]. The key methodological challenge for CLAW, which is highly nontrivial, involves developing pairwise exchangeable scores that integrate information from null samples, test samples, and side information to mimic the oracle. As highlighted by Referee 1, while CLAW shares some overarching techniques -- such as mixing and empirical Bayes -- with AdaDetect \citep{marandon22mlfdr}, it brings several innovations to existing strategies to accommodate side information more effectively. 

\end{enumerate}

\subsubsection{Comparison with other BC-based multiple testing methods}  

An insightful reviewer noted that CLAW is conceptually connected with both AdaPT \citep{lei18adapt} and PLIS \citep{zhao2023plis}, as all three methods, in their simplified forms, utilize the BC algorithm in their basic operations. This significant connection has been thoroughly elucidated in Section \ref{app:kn}. Below, we provide additional discussions to clarify the key differences among the three methods.

\begin{enumerate}[(a)]
    
    \item The types of structural information utilized vary across methods. In AdaPT and CLAW, side information is encoded as a generic covariate sequence.  In contrast, PLIS represents structural information using a graphical model (e.g., hidden Markov models or Ising models), which captures the dependency structure among latent states. Different types of structural information require distinct methods and frameworks: The graphical model encapsulates our prior knowledge of the dependence structures of latent states, which cannot be adequately represented through a covariate sequence. Conversely, the structural information encoded in a covariate sequence cannot be effectively captured by a graphical model.
     
  \item  The approaches for counting false positives vary across methods. AdaPT counts the number of ``large'' p-values derived from the \emph{test data}, whereas CLAW employs a bivariate score function that counts the number of ``small'' null scores using \emph{calibration data}. Moreover, AdaPT requires that the null p-values be uniform or mirror conservative, whereas CLAW requires pairwise exchangeability between the test and calibration scores under the null.
       
       \item The strategies for incorporating side information differ across methods. AdaPT employs a flexible iterative approach, refining hypothesis rankings based on side information or user feedback. The flexibility of this framework includes: (i) it generalizes to the important interactive multiple testing scenario \citep{STAR20}; (ii) it requires only mutually independent and mirror-conservative p-values under the null hypothesis, thereby accommodating composite nulls. In contrast, CLAW employs a ``conformal'' approach, employing Clfdr-type scores. While being able to tackling multivariate test data and exhibiting higher power, CLAW can only handle the task of testing single/sharp null hypotheses. Finally, PLIS creates baseline data to preserve dependency structure, making it particularly well-suited to scenarios where a pre-specified model class, such as a hidden Markov model, is known \emph{a priori}.
\end{enumerate}

\subsection{Further discussions on the optimality theory}
\label{app:opt_or}

Proposition \ref{prop:opt} aims to establish the optimality of \( R(t,S_i) \) defined in \eqref{or_Rt}. This conclusion holds when (a) the true data-generating process is the covariate-adaptive model \eqref{model:mixture}, (b) the labeled null samples are independently drawn from \( f_0 \), and (c) the class of decision rules is restricted to the candidate rejection set \( \mathcal{A} = \{i: u_i < \tilde{u}_i\} \). While the constraint on \( \mathcal{A} \), which can be conceptualized as a screening mechanism, may initially seem stringent, there exists a large class of meaningful conformity score functions \( g(t,S_i) \) for which the two subsets $\{i:R(T_i,S_i)<R(\tilde{T}_i,S_i)\}$ and $\{i:g(T_i,S_i)<g(\tilde{T}_i,S_i)\}$ are identical. Below are some important examples. 

    \begin{enumerate}
        \item 
Consider the problem of testing grouped hypotheses discussed in Section \ref{app:subsub-pu-group}. We examine two score functions: the first is the oracle score function employed by CLAW, and the second is the density ratio function \( r(t, k) \), utilized within specific groups with $S_i=k$. The function \( r(t, k) \) has been commonly adopted in the machine learning literature. These two score functions differ by a factor of \( {\pi}_{k} \). For a particular group, let us assume that we employ \( r(t, k) \) as a screening index, which excludes candidate hypotheses when \( r(T_i, k) \geq r(\tilde{T}_i, k) \). Elementary calculations reveal that 
$$R(t, k) =  \frac{(1-{\pi}_{k}) r(t, k)}{1 - (1-{\pi}_{k}) r(t, k)},
$$ 
and \( \mathrm{sign}(R(T_i, k) - R(\tilde{T}_i, k)) = \mathrm{sign}(r(T_i, k) - r(\tilde{T}_i, k)) \). This leads to the conclusion that, for \( S_i = k \), \( \mathcal{A} = \{i : R(T_i, k) < R(\tilde{T}_i, k)\} = \{i : r(T_i, k) < r(\tilde{T}_i, k)\} \). Thus, our utilization of \( \mathcal{A} = \{i : R(T_i, S_i) < R(\tilde{T}_i, S_i)\} \) aligns with intuitive screening rules. 
             
        \item Next, we illustrate that a broad class of meaningful conformity score functions will produce the same rejection set \( \mathcal{A} \), which is determined solely by the relative significance levels of the test data and calibration data for each test unit. We provide two illustrative examples. First, in a conventional multiple testing framework, it is typically observed that a larger absolute value (or norm) of the statistic \( T_i \) provides stronger evidence against the null hypothesis. Given that \( \tilde{T}_i \) follows the null distribution, it is reasonable to employ score calculation mechanisms that satisfy \( g(a, S_i) \leq g(b, S_i) \) whenever \( |a| > |b| \). Specifically, the sign of \( g(T_i, S_i) - g(\tilde{T}_i, S_i) \) should be dictated by the sign of \( |T_i| - |\tilde{T}_i| \), which ensures that \( \{i : R(T_i, S_i) < R(\tilde{T}_i, S_i)\} = \{i : |T_i| > |\tilde{T}_i|\} \). 
For the second example, suppose we convert \( T_i \) into a p-value, which corresponds to significance levels. In this case, we similarly arrive at \( \{i : R(T_i, S_i) < R(\tilde{T}_i, S_i)\} = \{i : p_i < \tilde{p}_i\} \). Basically, for the same unit with identical covariates \( S_i \), the relative significance between \( T_i \) and \( \tilde{T}_i \) can be inferred directly from the observations themselves, irrespective of the specific score calculation mechanism applied. 
        
\end{enumerate}

The following corollary follows from Proposition \ref{prop:opt}, providing an optimality theory for BC-type algorithms. 
 
    \begin{corollary}
Consider the scores \((u_i, \tilde{u}_i)_{i=1}^m\) calculated by the oracle score functions in Proposition \ref{prop:opt}. Let \(W_i = \mathrm{sign}(\tilde{u}_i - u_i)(u_i \wedge \tilde{u}_i)^{-1}\) and \(W_i' = \mathrm{sign}(\tilde{u}_i - u_i)L(u_i, \tilde{u}_i)\) for any non-negative symmetric function \(L\). For the BC-type decision rules \(\mathcal{R} = \{W_i \geq t\}\) and \(\mathcal{R}' = \{W_i' \geq t'\}\), the candidate rejection set is given by \(\mathcal{A} = \{i: W_i > 0\} = \{i: W_i' > 0\}\).
 If $\mathrm{mFDR}(\mathcal{R}) = \alpha$ and 
    $\mathrm{mFDR}(\mathcal{R}')\leq\alpha$, then $\EE(|\mathcal{R}\cap\mathcal{H}_{0}^{c}|)\geq \EE(|\mathcal{R}'\cap\mathcal{H}_{0}^{c}|)$.
    \end{corollary}

\section{Supplementary numerical results}
\label{app:simu}

This section presents additional numerical results to complement the simulation studies discussed in the main text. We include simulation results in more complex settings, such as those involving multivariate auxiliary data (Section \ref{subsec:spatial-2D}), equally correlated data (Section \ref{app:simu-exch}), and non-exchangeable data (Section \ref{app:simu-pwexch}). Furthermore, we compare CLAW with its NEB counterpart, the Clfdr procedure, in high-dimensional settings (Section \ref{appsimu:group}).
To illustrate the augmentation strategy of CLAW for continuous covariates presented in Section \ref{subsub-pu-aug}, Section \ref{appsimu:srandom} provides simulation results under setups where \(\mathbf{S}\) are continuous random variables. Additionally, Section \ref{app:yeast} includes auxiliary tables and figures for real data applications.

\subsection{Multiple testing with multivariate covariates}
\label{subsec:spatial-2D}

This section extends the simulation in Section \ref{subsec:spatial} to situations where the covariate $S_i$ corresponds to the location in two-dimensional spatial regions. 

The data are generated according to the following location-adaptive mixture model on a $m=100\times100$ lattice, where the covariate $S_{i}\in[100]\times[100]\subset\mathbb{R}^{2}$ denotes the location of $T_{i}:$  
\begin{equation*}
    T_i|(\theta_{i},S_{i}=\pmb{s}) \overset{ind.}{\sim} (1-\theta_{i})\mathcal{N}(0,1)+\theta_{i}F_{1}, \quad i\in[m].
 \end{equation*}
Let $\pmb{s}=(x,y)$. We consider the following spatial patterns:
\begin{enumerate}[I.]
    \item $F_{1}=\mathcal{N}(\mu,1)$; $\pi_{s}=0.75$ for $\{\pmb{s}: 10\leq(x-30)^{2}+(y-70)^{2}\leq20 \}\cup \{\pmb{s}:62\leq x\leq90 \text{ and } 10\leq y\leq38 \}$, and $\pi_{s}=0.02$ otherwise. 
    \item $F_{1}=\mathcal{N}(2.8,1)$; $\pi_{s}=\pi$ for $\{\pmb{s}: 10\leq(x-30)^{2}+(y-70)^{2}\leq20 \}\cup \{\pmb{s}:62\leq x\leq90 \text{ and } 10\leq y\leq38 \}$, and $\pi_{s}=0.02$ otherwise. 
    \item $F_{1}=\mathcal{N}(2.5,1)$; $\pi_{s}=0.75$ for $\{\pmb{s}: R/2\leq(x-30)^{2}+(y-70)^{2}\leq R \}\cup \{\pmb{s}:62\leq x\leq90 \text{ and } 10\leq y\leq38 \}$, and $\pi_{s}=0.02$ otherwise.
\end{enumerate}

The covariate-adaptive weights are chosen as $w_{ij}=\phi(||S_{i}-S_{j}||_{2}/15)$, where $||\cdot||_{2}$ denotes the Euclidean norm in $\mathbb{R}^{2}$. The calibration data are generated as $\Tilde{T}_{i}\overset{i.i.d.}{\sim}\mathcal{N}(0,1)$ for $i\in[m]$. We apply AdaDetect, BH, CLAW, LAWS and SABHA to the simulated data and summarize the simulation results in Figure \ref{fig:spatial_2D}. 

\begin{figure}[!htbp]
    \centering
    \includegraphics[width=0.9\linewidth,height=0.5\linewidth]{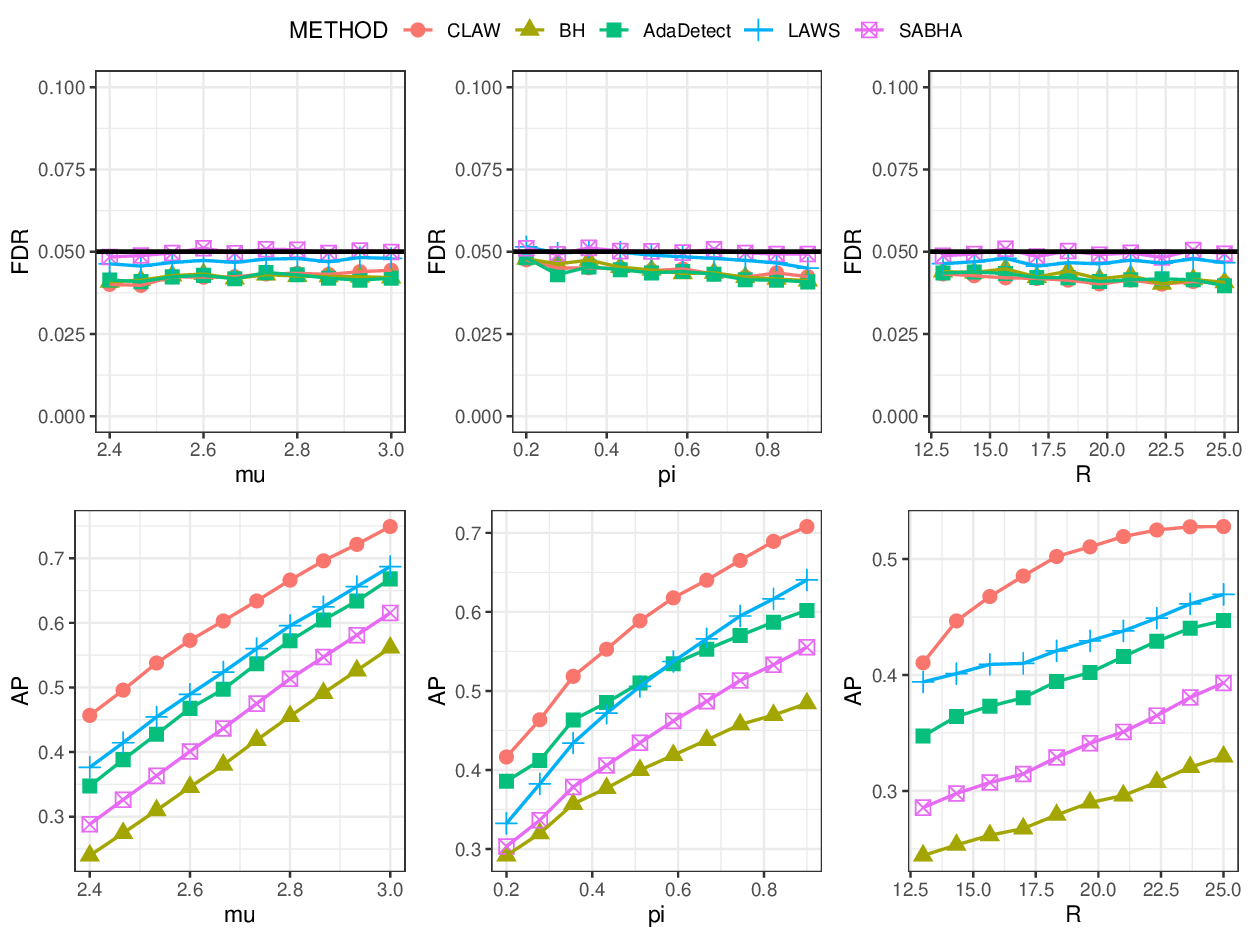}
    \caption{FDR and AP comparison for multiple testing for two-dimensional covariates at $\alpha=0.05$. The left, middle and right columns are corresponding to settings I, II and III, respectively.}\label{fig:spatial_2D}
\end{figure}

We can see that all methods effectively control the FDR, with CLAW exhibiting slight conservativeness. The power of BH and AdaDetect can be significantly enhanced by structure-adaptive methods such as LAWS and SABHA. CLAW further improves the power of both LAWS and SABHA by employing efficient scores that emulate the oracle rule. 


Finally we visualize a toy example to gain further insights. The test data $\mathbf{T}$ are generated on a $100\times100$ lattice: 
$T_{i}\overset{i.i.d.}{\sim}\mathcal{N}(\mu,1)$ if its location $S_{i}\in\{\pmb{s}=(x,y): 10\leq(x-30)^{2}+(y-70)^{2}\leq20 \}\cup \{\pmb{s}:62\leq x\leq90 \text{ and } 10\leq y\leq38 \}$ and $T_{i}\overset{i.i.d.}{\sim}\mathcal{N}(0,1)$ otherwise. 
In this setup, all signals are clustered either within the ring or the square area (the first column in Figure \ref{fig:spatial_visualize}). 

\begin{figure}[!htbp]
    \centering
    \includegraphics[width=0.9\linewidth]{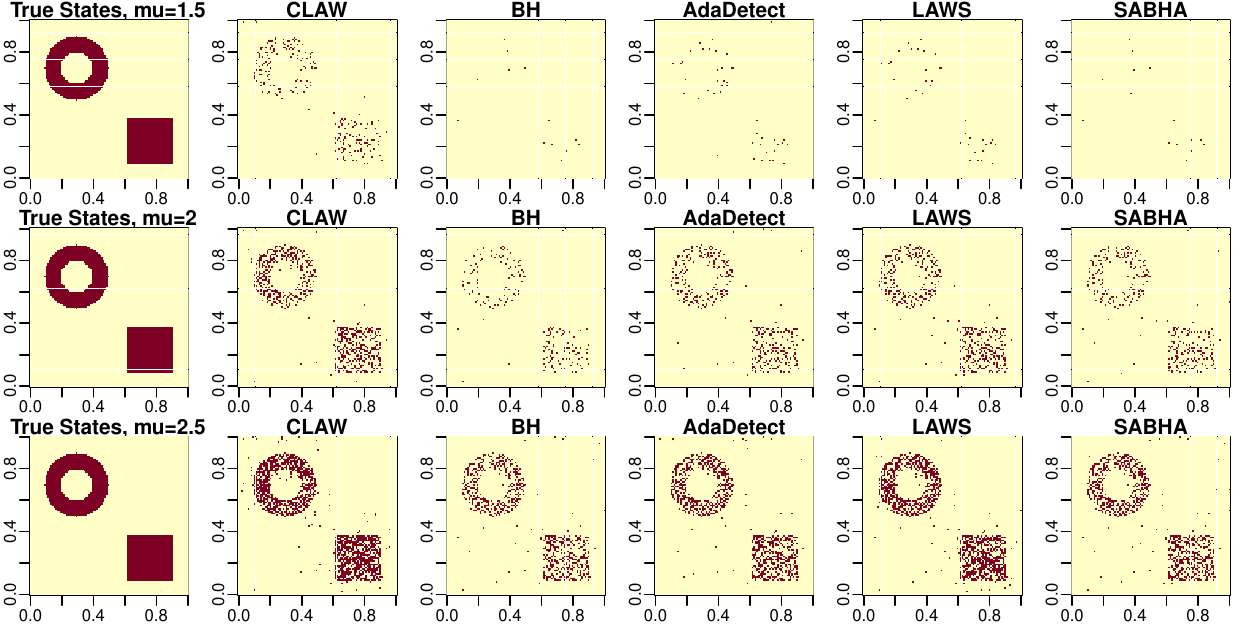}
    \caption{An example of signal recovering with different signal strength $\mu$. The claret dots are the discoveries by each multiple testing procedure at the nominal FDR level $\alpha=0.05$. The first, second and third row presents the results when $\mu=1.5$, $\mu=2$ and $\mu=2.5$, respectively.}\label{fig:spatial_visualize}
\end{figure}

We apply various methods at the nominal FDR level $\alpha=0.05$ and visualize the results in Figure \ref{fig:spatial_visualize}, illustrating the discovered locations by each method (columns 2-5) at different signal strengths $\mu=1.5$, $\mu=2$, and $\mu=2.5$ (rows 1-3). Notably, CLAW stands out as the most effective method for revealing the ring and square shapes. This is accomplished by CLAW's ability to adaptively exploit the structures in the test data and auxiliary data, ultimately constructing the most effective scores.

\subsection{Numerical results for exchangeable data}
\label{app:simu-exch}

This section presents simulation results to compare different methods under similar settings as described in Section \ref{subsec:spatial}, with the difference being that the data are not mutually independent. The test data are generated according to the following model:
\begin{equation*}
    T_{i}|(\theta_{i},S_{i}=s) \sim (1-\theta_{i})\mathcal{N}(0,1)+\theta_{i}F_{1s},\quad i=1,\cdots,3000,
\end{equation*}
where $\PP(\theta_{i}=1|S_{i}=s)=\pi_{s}$. The calibration data are generated as $\tilde{T}_{i}\sim\mathcal{N}(0,1)$. Additionally, the null data $(T_{i}:i\in\mathcal{H}_{0})\cup(\Tilde{T}_{i}:i\in[m])$ are generated from a multivariate Gaussian distribution, independent of the non-null test data $(T_{i}:i\notin\mathcal{H}_{0})$. The expected value of $(T_{i}:i\in\mathcal{H}_{0})\cup(\Tilde{T}_{i}:i\in[m])$ is zero, and the covariance matrix $\Sigma=(\sigma_{ij})$ is defined such that $\sigma_{ii}=1$ and $\sigma_{ij}=\rho\in[0,1)$ for $i\neq j$. This equi-correlated structure  within $(T_{i}:i\in\mathcal{H}_{0})\cup(\Tilde{T}_{i}:i\in[m])$ implies that the null data points are jointly exchangeable, thereby satisfying the conditional exchangeable assumption \eqref{jointexch-covariate}. 

The following settings are considered:
\begin{enumerate}[I.]
    \item $\rho=0.5$; $F_{1s}\equiv F_{1}=\mathcal{N}(\mu,1)$; $\pi_{s}=0.6$ for $s\in[201,350]\cup[1501,1650]$, $\pi_{s}=0.3$ for $s\in[801,1000]\cup[2101,2300]$, and $\pi_{s}=0.02$ otherwise. 
    \item $\rho=0.5$; $F_{1s}=\mathcal{N}(-2.5,1)$ if $s\in[1,1500]$, $F_{1s}=\mathcal{N}(3.6,1.5^2)$ if $s\in[1501,3000]$; $\pi_{s}=2\pi$ for $s\in[201,350]\cup[1501,1650]$, $\pi_{s}=\pi$ for $s\in[801,1000]\cup[2101,2300]$, and $\pi_{s}=0.02$ otherwise. 
    \item $\rho=0.5$; $F_{1s}=\mathcal{N}(\mu+0.15\sin(0.6s),1)$; $\pi_{s}=0.4(1+\sin(0.02s))$ for $s\in[201,500]\cup[801,1100]\cup[1501,1800]\cup[2101,2400]$ and $\pi_{s}=0.02$ otherwise. 
    
    \item $F_{1s}\equiv F_{1}=\mathcal{N}(3,1)$; $\pi_{s}=0.6$ for $s\in[201,350]\cup[1501,1650]$, $\pi_{s}=0.3$ for $s\in[801,1000]\cup[2101,2300]$, and $\pi_{s}=0.02$ otherwise. 
    \item $F_{1s}=\mathcal{N}(-2.5,1)$ if $s\in[1,1500]$, $F_{1s}=\mathcal{N}(3.6,1.5^2)$ if $s\in[1501,3000]$; $\pi_{s}=0.6$ for $s\in[201,350]\cup[1501,1650]$, $\pi_{s}=0.3$ for $s\in[801,1000]\cup[2101,2300]$, and $\pi_{s}=0.02$ otherwise. 
    \item $F_{1s}=\mathcal{N}(3+0.15\sin(0.6s),1)$; $\pi_{s}=0.4(1+\sin(0.02s))$ for $s\in[201,500]\cup[801,1100]\cup[1501,1800]\cup[2101,2400]$ and $\pi_{s}=0.02$ otherwise. 
\end{enumerate}

In Settings I-III, we fix $\rho=0.5$ as a constant, while in Settings IV-VI, we vary the correlation to explore its impacts on various methods. We apply AdaDetect, AdaPT, BH, CLAW, LAWS and SABHA to the simulated data and summarize the simulation results in Figure \ref{fig:spatial_exch}. 

\begin{figure}[!htbp]
    \centering
    \includegraphics[width=0.9\linewidth,height=0.5\linewidth]{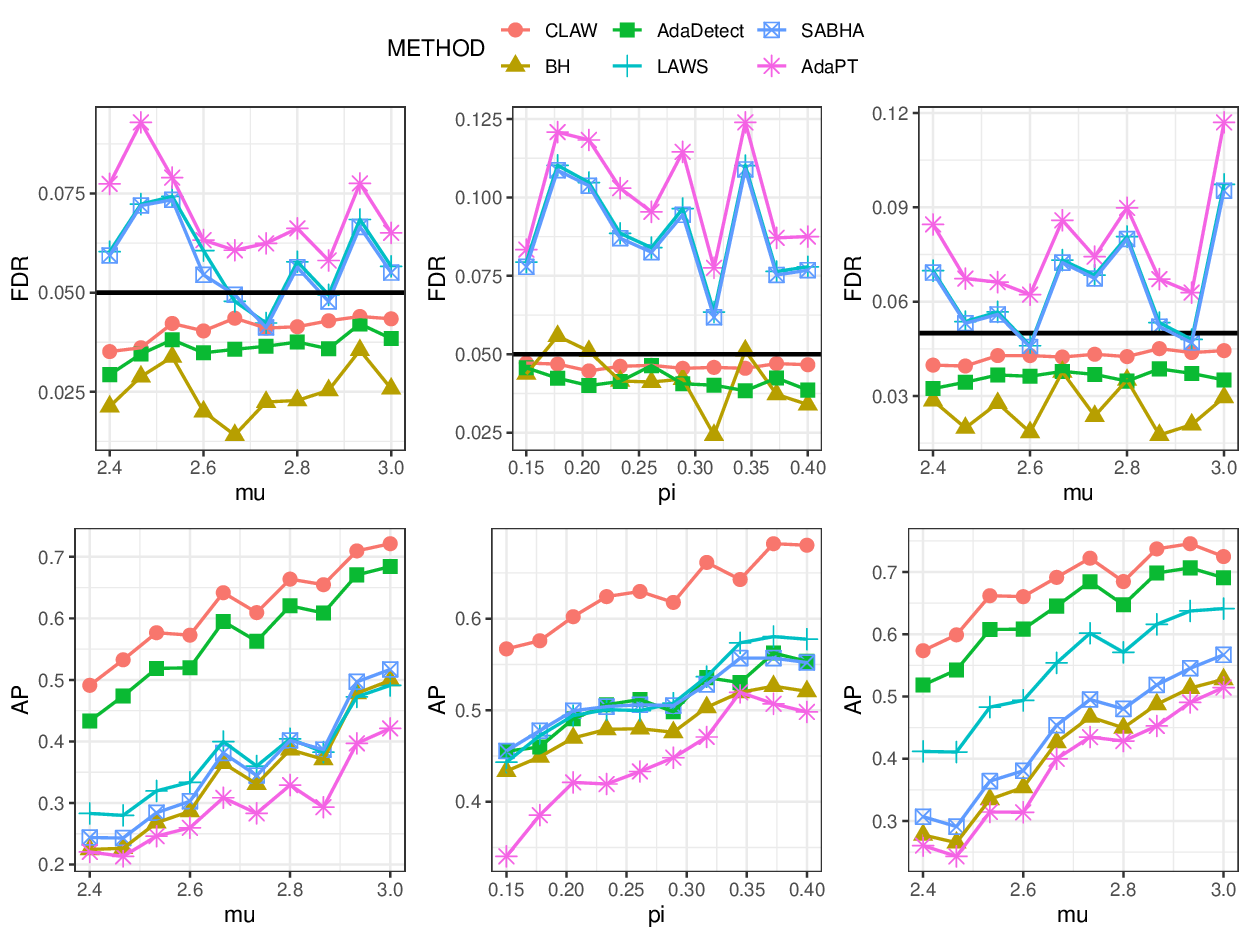}
    \includegraphics[width=0.9\linewidth,height=0.5\linewidth]{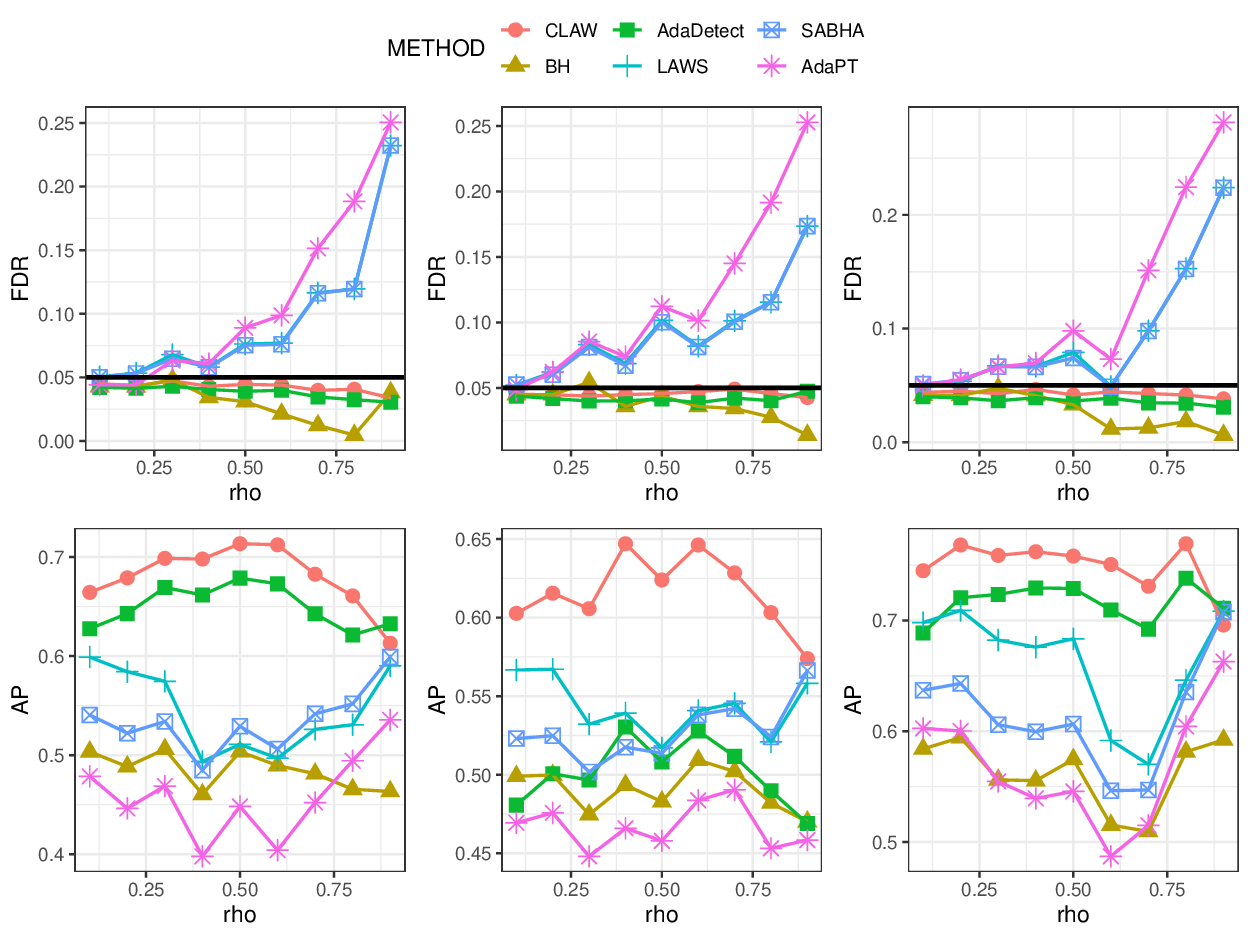}
    \caption{FDR and AP comparison for multiple testing for ordered sequences at $\alpha=0.05$ with jointly exchangeable null data. For the top two rows, the left, middle and right columns are corresponding to settings I, II and III, respectively. For the bottom two rows, the left, middle and right columns are corresponding to settings IV, V and VI, respectively.}\label{fig:spatial_exch}
\end{figure}

We observe that CLAW, BH, and AdaDetect effectively control the FDR in all settings where the null samples exhibit exchangeability. Conversely, LAWS, SABHA, and AdaPT fail to maintain FDR control at the nominal level in the presence of dependency. As the degree of dependency increases in Settings IV-VI, the inflation in FDR levels for LAWS, SABHA, and AdaPT becomes more pronounced. In contrast, CLAW maintains FDR control across all settings and outperforms other methods in terms of power.

\subsection{Numerical results for non-exchangeable data}
\label{app:simu-pwexch}

This section presents numerical studies aimed at evaluating the performance of various conformal methods for non-exchangeable data. Our investigation is structured into two parts. The first subsection focuses on experiments involving data that do not satisfy the joint exchangeability condition but meet the pairwise exchangeability condition \eqref{data_pwexch}. In this context, existing conformal methods, such as AdaDetect \citep{marandon22mlfdr}, lack theoretical guarantees for false discovery rate (FDR) control; however, CLAW remains provably valid for FDR control. The second part examines a scenario in which pairwise exchangeability is also violated. Under these circumstances, all methods fail to control the FDR, indicating that the condition \eqref{data_pwexch} appears to be indispensable within the CLAW framework.

\subsubsection{Numerical results for non-exchangeable but pairwise exchangeable data}

To generate pairwise exchangeable data samples, we begin by simulating data from a stationary AR(1) process $(y_{i}:i\in[3000])$. Each $y_i$ follows a marginal distribution of $\mathcal{N}(0,1)$, and the auto-regression coefficients are defined as $\mathrm{cor}(y_{i},y_{j})=\rho^{|i-j|}$, where $\rho\in(-1,1)$.

The test and calibration data $\mathbf{T}$ and $\Tilde{\mathbf{T}}$ are generated according to the following model: 
\begin{equation*}
T_{i}|(\theta_{i}=0,S_{i}=s) = y_{i}+\epsilon_{i}, \quad T_{i}|(\theta_{i}=1,S_{i}=s) \sim F_{1s}, \quad \Tilde{T}_{i} = y_{i}+\epsilon_{i+3000},
\end{equation*}
where $\PP(\theta_{i}=1|S_{i}=s)=\pi_{s}$, and $\{\epsilon_{i}:i\in[6000]\}$ are i.i.d. $\mathcal{N}(0,0.01)$ noises, and $S_i$ indicates the sequential order of each observation. Furthermore, the non-null data $(T_{i}:i\notin\mathcal{H}_{0})$ are drawn from $F_{1s}$ conditional on $S_{i}$ and are independent of the null samples $(T_{i}:i\in\mathcal{H}_{0})\cup(\Tilde{T}_{i}:i\in[m])$. This ensures that the pairwise exchangeability between null data samples \eqref{data_pwexch} is satisfied (see also the justifications in Example 4 of Section \ref{app:pwexch}).

We consider the six settings in our simulation studies. In Settings I-III, we fix $\rho=0.5$ as a constant, while in Settings IV-VI, we vary the correlation to explore its impacts on various methods. 

\begin{enumerate}[I.]
    \item $\rho=0.5$; $F_{1s}\equiv F_{1}=\mathcal{N}(\mu,1)$; $\pi_{s}=0.6$ for $s\in[201,350]\cup[1501,1650]$, $\pi_{s}=0.3$ for $s\in[801,1000]\cup[2101,2300]$, and $\pi_{s}=0.02$ otherwise. 
    \item $\rho=0.5$; $F_{1s}=\mathcal{N}(-2.5,1)$ if $s\in[1,1500]$, $F_{1s}=\mathcal{N}(3.6,1.5^2)$ if $s\in[1501,3000]$; $\pi_{s}=2\pi$ for $s\in[201,350]\cup[1501,1650]$, $\pi_{s}=\pi$ for $s\in[801,1000]\cup[2101,2300]$, and $\pi_{s}=0.02$ otherwise. 
    \item $\rho=0.5$; $F_{1s}=\mathcal{N}(\mu+0.15\sin(0.6s),1)$; $\pi_{s}=0.4(1+\sin(0.02s))$ for $s\in[201,500]\cup[801,1100]\cup[1501,1800]\cup[2101,2400]$ and $\pi_{s}=0.02$ otherwise. 
    
    \item $F_{1s}\equiv F_{1}=\mathcal{N}(3,1)$; $\pi_{s}=0.6$ for $s\in[201,350]\cup[1501,1650]$, $\pi_{s}=0.3$ for $s\in[801,1000]\cup[2101,2300]$, and $\pi_{s}=0.02$ otherwise.
    \item $F_{1s}=\mathcal{N}(-2.5,1)$ if $s\in[1,1500]$, $F_{1s}=\mathcal{N}(3.6,1.5^2)$ if $s\in[1501,3000]$; $\pi_{s}=0.6$ for $s\in[201,350]\cup[1501,1650]$, $\pi_{s}=0.3$ for $s\in[801,1000]\cup[2101,2300]$, and $\pi_{s}=0.02$ otherwise. 
    \item $F_{1s}=\mathcal{N}(3+0.15\sin(0.6s),1)$; $\pi_{s}=0.4(1+\sin(0.02s))$ for $s\in[201,500]\cup[801,1100]\cup[1501,1800]\cup[2101,2400]$ and $\pi_{s}=0.02$ otherwise. 
\end{enumerate}

We apply AdaDetect, AdaPT, BH, CLAW, LAWS and SABHA to the simulated data and summarize the simulation results in Figure \ref{fig:spatial_non-exch}. The following observations can be made. First, AdaDetect, BH, and CLAW effectively control the FDR, despite the lack of rigorous theoretical guarantees for BH and AdaDetect. Second, LAWS, SABHA, and AdaPT fail to control the FDR in certain scenarios. However, the FDR inflation observed is smaller compared to the scenarios discussed in Section \ref{app:simu-exch}. Third, the relatively weak dependency in the AR(1) process, characterized by the exponential decrease in correlation coefficient, plays a significant role. This explains why AdaDetect and BH appear to control the FDR. 
Finally, CLAW demonstrates the highest power in most cases. However, it may exhibit reduced power under very strong correlations. This observation suggests that the conformity scores generated by CLAW may not be as effective under intricate dependence structures.

\begin{figure}[!htbp]
    \centering
    \includegraphics[width=0.9\linewidth,height=0.5\linewidth]{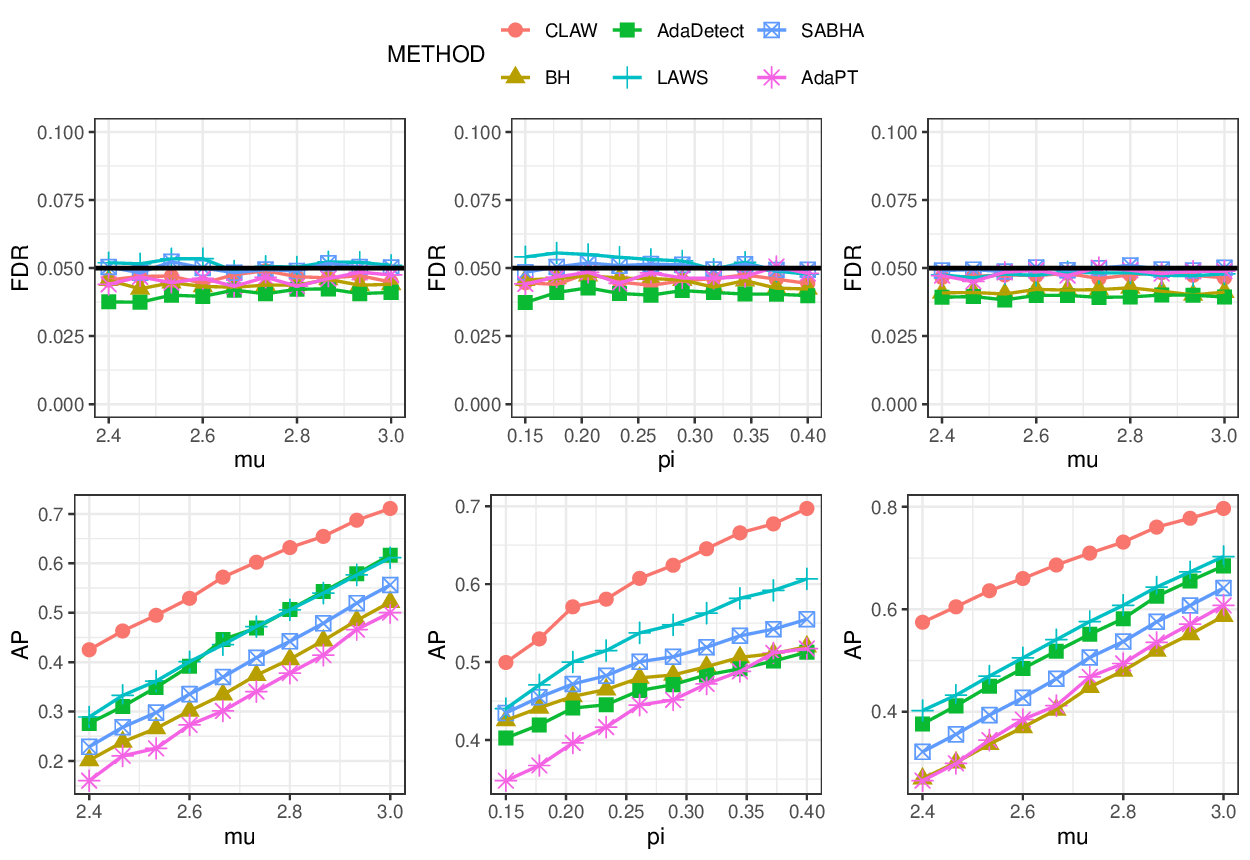}
    \includegraphics[width=0.9\linewidth,height=0.5\linewidth]{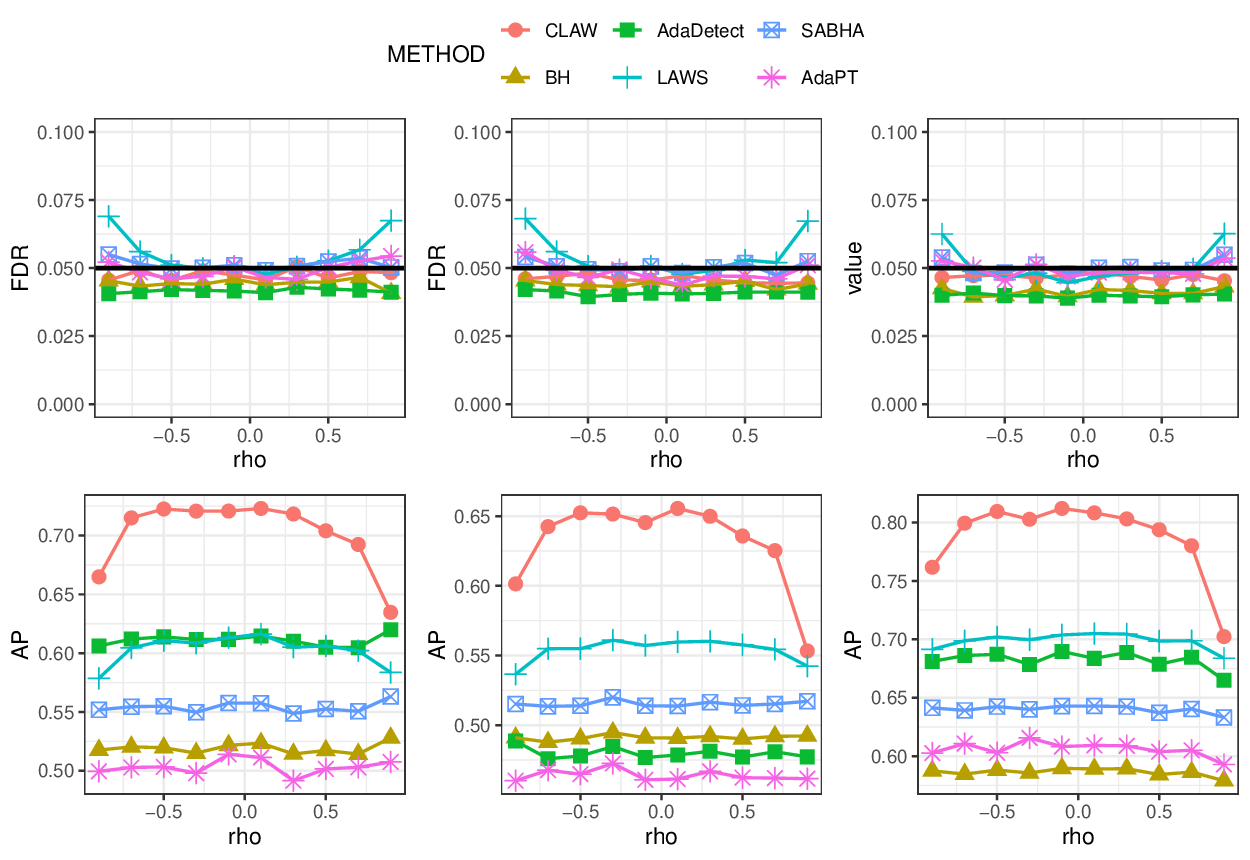}
    \caption{FDR and AP comparison for multiple testing for ordered sequences at $\alpha=0.05$ with pairwise exchangeable null data. For the top two rows, the left, middle and right columns are corresponding to settings I, II and III, respectively. For the bottom two rows, the left, middle and right columns are corresponding to settings IV, V and VI, respectively.}\label{fig:spatial_non-exch}
\end{figure}

\subsubsection{Numerical results for data without (pairwise) exchangeability}

We employ the strategies outlined in the previous section to generate null test data from an AR(1) process. However, the calibration data \(\tilde{\mathbf{T}}\) is generated as i.i.d. \(\mathcal{N}(0,1)\) variables. In this scenario, neither the joint exchangeability assumption \eqref{jointexch-covariate} nor the pairwise exchangeability assumption \eqref{data_pwexch} is satisfied. The simulation settings are identical with Settings IV-VI presented in the previous section. We analyze the performance of various methods across different values of \(\rho\), with the simulation results illustrated in Figure \ref{fig:total_non-exch}. 

\begin{figure}[!htbp]
    \centering
    \includegraphics[width=0.9\linewidth,height=0.5\linewidth]{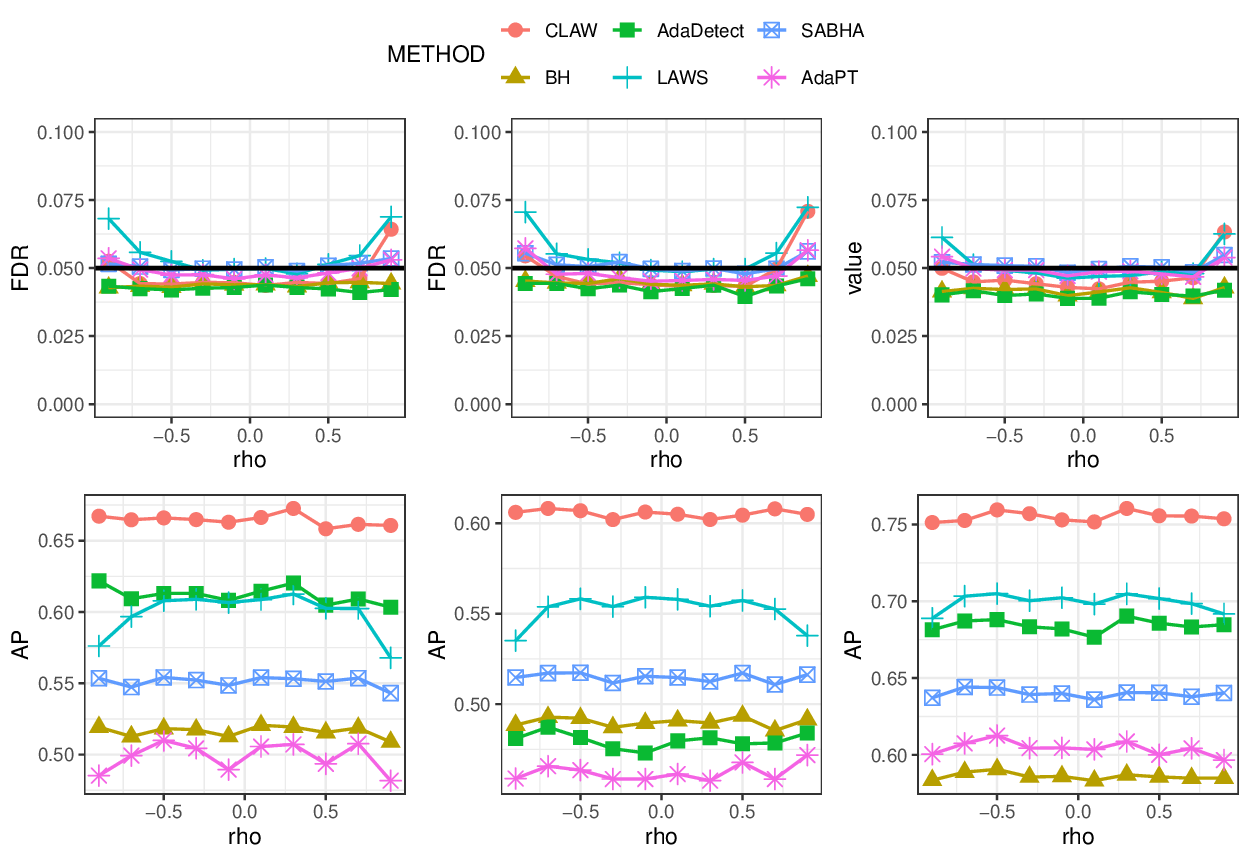}
    \caption{FDR and AP comparison for multiple testing for ordered sequences at $\alpha=0.05$ with AR(1) test data and i.i.d. calibration data. The left, middle and right columns are corresponding to settings IV, V and VI in Section E.3.1, respectively.}\label{fig:total_non-exch}
\end{figure}

Our analysis reveals that when the correlations are small to moderate, all methods effectively control the FDR at the nominal level. However, several methods, including LAWS, SABHA, AdaPT, and CLAW, fail to maintain FDR control at the nominal level when \(|\rho|\) is large. Notably, CLAW demonstrates the highest power across all scenarios.

The inflation of FDR levels observed for CLAW becomes particularly pronounced under conditions of strong correlation. This phenomenon arises from the increased discrepancy between the joint distributions of the test and calibration samples. Specifically, while large correlations exist within \(\mathbf{T}\), the calibration samples \(\tilde{\mathbf{T}}\) consist of independent and identically distributed (i.i.d.) samples, leading to significant violations of the exchangeability condition and resulting in FDR inflation.

Our findings concerning dependence are preliminary and limited. Addressing the complex issue of developing valid and efficient FDR methods under dependency extends beyond the scope of this work. We view this as a promising direction for future research.

\subsection{Comparison with the oracle Clfdr method}
\label{appsimu:group}

The oracle CLfdr procedure, proposed by \citet{cs09}, is optimal in the setting where the covariate-adaptive mixture model is known. However, the validity of the data-driven Clfdr procedure relies on consistent estimates of the CLfdr statistics. Moreover, the data-driven Clfdr procedure only offers asymptotic FDR control. This subsection provides numerical evidence to illustrate the challenge of achieving consistent estimation in high-dimensional settings, where the data-driven CLfdr procedure may encounter severely inflated FDR levels. In contrast, CLAW demonstrates efficacy and robustness in controlling the FDR at the nominal level across all settings we have investigated.

Our simulation considers multiple testing with grouped hypotheses, where the data are generated according to the following model:
\begin{equation*}
    T_{i}|(\theta_{i},S_{i}=k) \overset{ind.}{\sim} (1-\theta_{i})\mathcal{N}_{d}(0,\mathbf{I}_{d})+\theta_{i}\mathcal{N}_{d}(\pmb{\mu}_{k},\mathbf{I}_{d}),\quad i\in[m],\quad k\in\{1,2\}.
\end{equation*}
Here, $\mathcal{N}_{d}(0,\mathbf{I}_{d})$ represents $d$-dimensional standard normal random vectors. 
In the first group ($S_{i}=1$), the number of tests is $m_{1}=1000$. We set $\pi_{1}=0.2$, and $$\pmb{\mu}_{1}=(\sqrt{2\log{d}},\sqrt{2\log{d}},\sqrt{2\log{d}},\sqrt{2\log{d}},0,\cdots,0)^{\top}\in\mathbb{R}^{d},$$ with the exception that $\pmb{\mu}_{1}=(\sqrt{2\log{d}},\sqrt{2\log{d}})$ when $d=2$. For the second group ($S_{i}=2$), the number of tests is $m_{2}=2500$. We set $\pi_{2}=0.1$ 
and $\pmb{\mu}_{2}=(2,2,\cdots,2)^{\top}\in\mathbb{R}^{d}$. 

Estimating the non-null proportion poses a challenge in the high-dimensional setting. To focus on the key message, we implement both CLAW and Clfdr by assuming known values for $\pi_{k}$. Another possibility is to fix $\pi_k\equiv 0$, as done in \cite{marandon22mlfdr}. The simulation results are depicted in Figure \ref{fig:group2_multi}.

\begin{figure}[!htbp]
    \centering
    \includegraphics[width=0.9\linewidth]{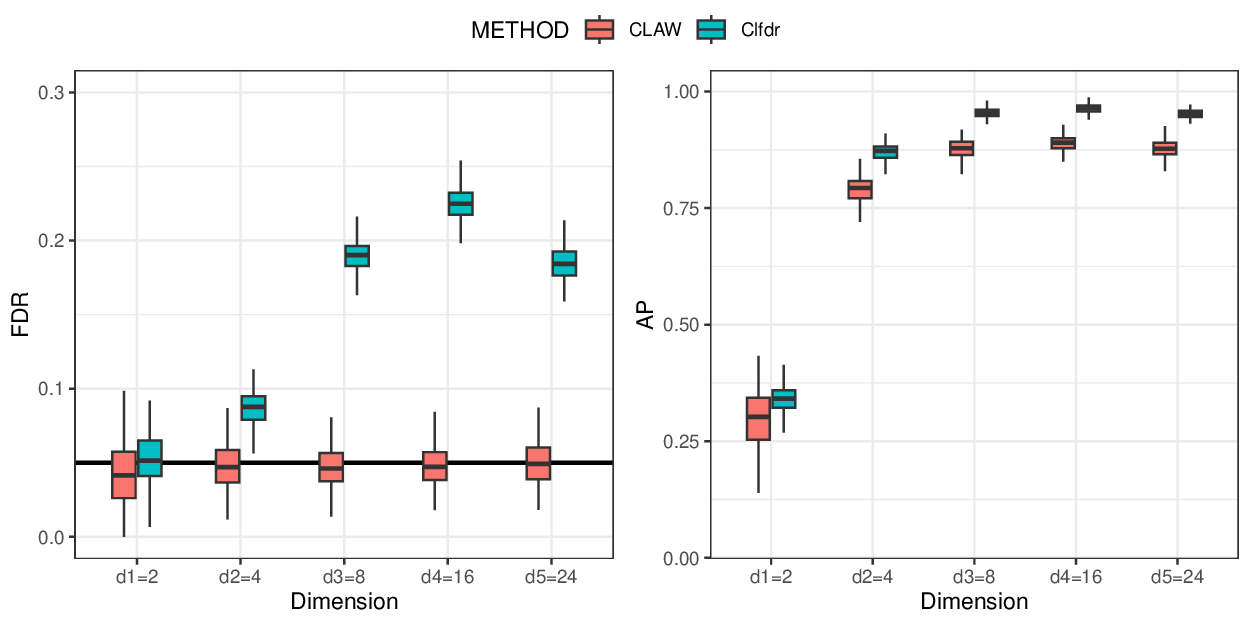}
    \caption{FDR and AP comparison for grouped multiple testing for multivariate test data.}\label{fig:group2_multi}
\end{figure}

As the dimension $d$ increases, the kernel density estimator suffers from the curse of dimensionality. Consequently, we can see that the Clfdr method, which relies on consistently estimated Clfdr statistics, fails to effectively control the FDR at the designated level. In contrast, CLAW, despite using inaccurately estimated scores, still effectively controls the FDR. Additionally, Figure \ref{fig:group2_multi} illustrates that the heights of the boxes representing CLAW are greater than those for Clfdr. This disparity arises from the randomized nature of CLAW, which incorporates both test data and calibration data into its operation.

\subsection{Numerical results for continuous random covariates}
\label{appsimu:srandom}

This section presents simulation results under setups where $\mathbf{S}$ are continuous random variables. We first consider the situation where the covariates are directly observable, then turn to constructing $\mathbf{S}$ from the raw observations. The CLAW procedure is implemented using the augmentation strategy described in Section \ref{subsub-pu-aug} throughout this section.

\textbf{Simulation Study 1: } Given covariates $\mathbf{S}$, the test data are generated conditional on $\mathbf S$ according to the following model: 
\begin{equation*}
    T_{i}|(\theta_{i},S_{i}=s) \overset{ind.}{\sim} (1-\theta_{i})\mathcal{N}(0,1)+\theta_{i}F_{1s},\quad i=1,\cdots,3000,
\end{equation*}
where $\pi_{s}=\PP(\theta_i=1|S_i=s)$. The calibration samples are i.i.d. $\mathcal{N}(0,1)$ variables.  The following settings are considered:
\begin{enumerate} 
    \item $S_i\overset{i.i.d.}{\sim}\mathrm{Beta}(2,5)$; $F_{1s}=\mathcal{N}(\mu,(1.5s)^2)$; $\pi_{s}=s$. Let $\mu$ vary.
    \item $S_i\overset{i.i.d.}{\sim}\mathrm{Laplace}(3)$; $F_{1s}=\mathcal{N}(2.8+0.3\mathrm{sign}(s),|s|)$; $\pi_{s}=\min\{1,\pi |s|\}$. Let $\pi$ vary.
    \item $S_i\overset{i.i.d.}{\sim}\mathrm{Laplace}(\nu)$; $F_{1s}=\mathcal{N}(2.8+0.3\mathrm{sign}(s),|s|)$; $\pi_{s}=\min\{1, 0.6|s|\}$. Let $\nu$ vary.
\end{enumerate}
We apply CLAW, BH, AdaDetect, LAWS, SABHA and AdaPT at FDR level $\alpha=0.05$ to the simulated data and summarize the simulation results in Figure \ref{fig:srandom}.  

\begin{figure}[!htbp]
    \centering
    \includegraphics[width=0.9\linewidth,height=0.5\linewidth]{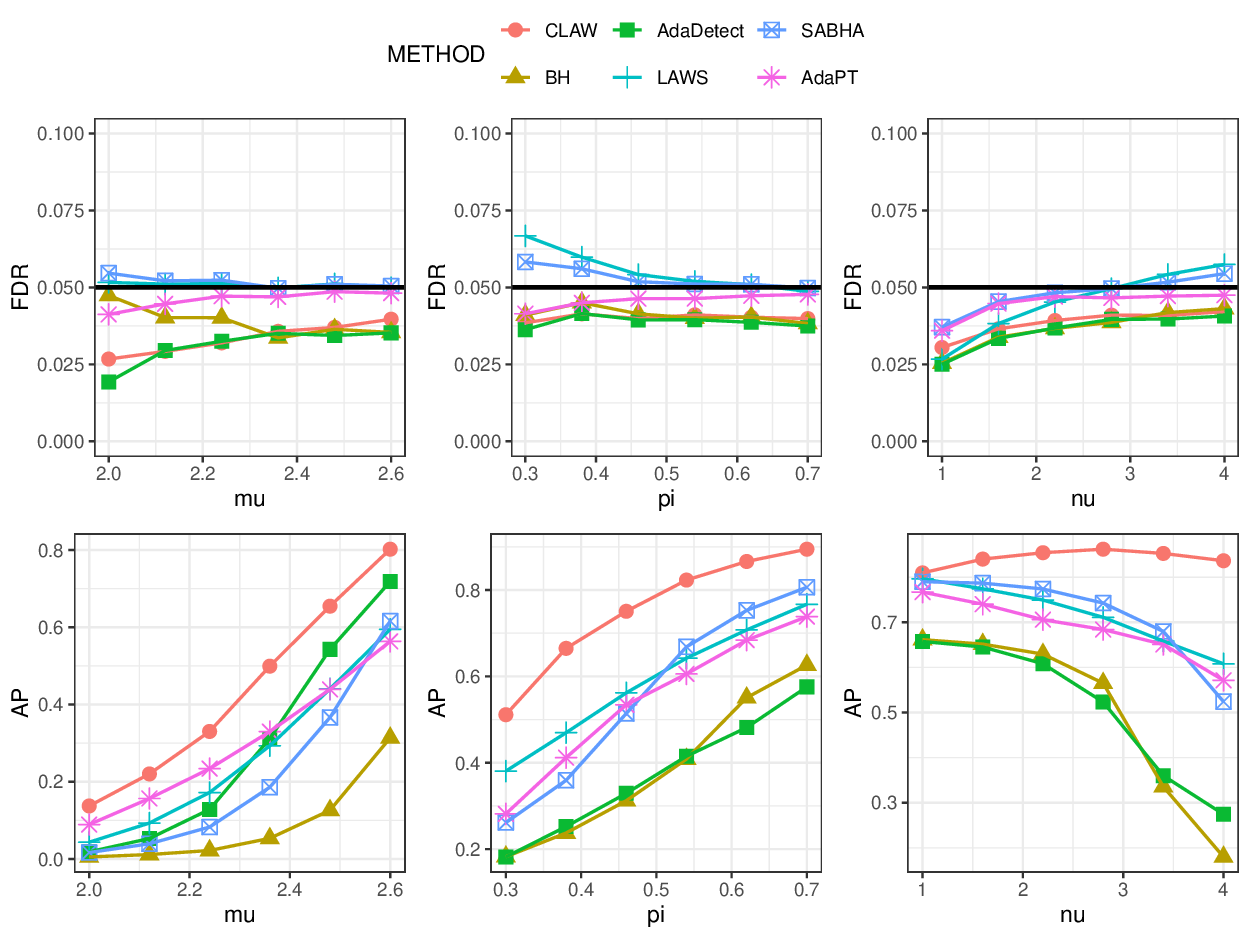}
    \caption{FDR and AP comparison when covariates are observable continuous random variables. 
    The left, middle and right columns are corresponding to settings 1, 2 and 3, respectively.}\label{fig:srandom}
\end{figure}

The following patterns can be noted from  Figure \ref{fig:srandom}. 
\begin{itemize}
    \item The FDR levels of CLAW, BH, AdaDetect, AdaPT, which are provably valid for FDR control in finite samples, strictly stay below the nominal level. These methods are relatively conservative in some scenarios. By contrast, SABHA and LAWS have mild inflations in FDR levels, although the violations seem to be small. This is consistent with the theory as both methods only offer asymptotic control of the FDR if the estimation is accurate. 
    \item  CLAW is  the most effective method in most cases, as it integrates all pertinent side information within the covariate-adaptive model. This includes variance, signal magnitude, and local sparsity levels, all of which contribute to the efficiency gain of conformity scores utilized by CLAW, which closely emulate the rankings produced by Clfdr.
    
      \item The first column of Figure \ref{fig:srandom} demonstrates that the power of AdaDetect increases as the strength of the signals increases (indicated by rising values of \(\mu\)). However, when \(\mu\) is correlated with the covariates, the power of AdaDetect can be lower than that of the BH procedure. This phenomenon is also observed in the results for nonrandom covariates that suggest sequential ordering. The diminished power in these cases can be attributed to the presence of heterogeneous signals -- particularly when both positive and negative signals coexist within the sequence (as illustrated in Figure \ref{fig:spatial_1D}). In such scenarios, the mixing strategy employed by AdaDetect may offset the increased signal strength, ultimately resulting in significantly reduced power.
      
  \end{itemize}

\textbf{Simulation Study 2: } The data generation process of Example 1 in Section \ref{app:pwexch} is considered. Specifically, we let $m=3000$, $n_x=n_y=1$, and the following settings are considered:
\begin{enumerate}
    \item $\theta_i\overset{i.i.d.}{\sim} \mathrm{Bernoulli}(0.1)$, 
    $X_i|\theta_i \overset{ind.}{\sim} (1-\theta_i)\mathcal{N}(0,1)+\theta_i \mathcal{N}(1,1)$, 
    $Y_i|\theta_i \overset{ind.}{\sim} (1-\theta_i)\mathcal{N}(0,0.7^2)+\theta_i \mathcal{N}(-\mu,0.7^2)$ for $i=1,\cdots,3000$.
    \item $\theta_i\overset{i.i.d.}{\sim} \mathrm{Bernoulli}(\pi)$, 
    $X_i|\theta_i \overset{ind.}{\sim} (1-\theta_i)\mathcal{N}(0,1)+\theta_i \mathcal{N}(1,1)$, 
    $Y_i|\theta_i \overset{ind.}{\sim} (1-\theta_i)\mathcal{N}(0,0.7^2)+\theta_i \mathcal{N}(-2.2,0.7^2)$ for $i=1,\cdots,3000$.
    \item  $X_i\overset{i.i.d.}{\sim} \mathcal{N}(3,1)$ for $i\in[801,1001+N]$, and $X_i\overset{i.i.d.}{\sim} \mathcal{N}(0,1)$ otherwise;
    $Y_i\overset{i.i.d.}{\sim} \mathcal{N}(-0.5,0.7^2)$ for $i\in[1001,2000]$, and $Y_i\overset{i.i.d.}{\sim} \mathcal{N}(0,0.7^2)$ otherwise; $\theta_i:=\II\{\EE[X_i]\neq\EE[Y_i]\}$.
\end{enumerate}

To test $H_i:\EE[X_i]=\EE[Y_i]$, i.e., $H_i:\theta_i=0$, we construct the test statistics $\mathbf{T}$ and covariates $\mathbf{S}$ as illustrated in \citet{CARS} and Example 1 in Section \ref{app:pwexch},
\begin{equation*}
    (T_{i},S_{i})= \sqrt{\frac{1}{2}} \left( \frac{{X}_{i}-{Y}_{i}}{\sigma_{pi}} , \frac{{X}_{i}+\kappa_{i}{Y}_{i}}{\sqrt{\kappa_{i}}\sigma_{pi}}  \right), \quad i\in[m],
\end{equation*}
where $\sigma_{pi}^{2}=(1^2+0.7^2)/2$ and $\kappa_{i}=1/0.7^2$. To apply conformal methods, the null calibration data are generated as i.i.d. $\mathcal{N}(0,1)$ variables. The simulation results are displayed in Figure \ref{fig:twosample}, from which we can draw similar conclusions in the experiments with observable continuous covariates (Figure \ref{fig:srandom}).

While all methods effectively control the FDR at the nominal level, those that successfully integrate side information related to the vector support exhibit improved efficiency. In most cases, CLAW demonstrates the highest power. This is attributed to CLAW being a conformalized version of the CARS procedure, which is optimal in this context. For further details, please refer to Example 1 in Section \ref{app:pwexch} of the Supplement.

\begin{figure}[!htbp]
    \centering
    \includegraphics[width=0.9\linewidth,height=0.5\linewidth]{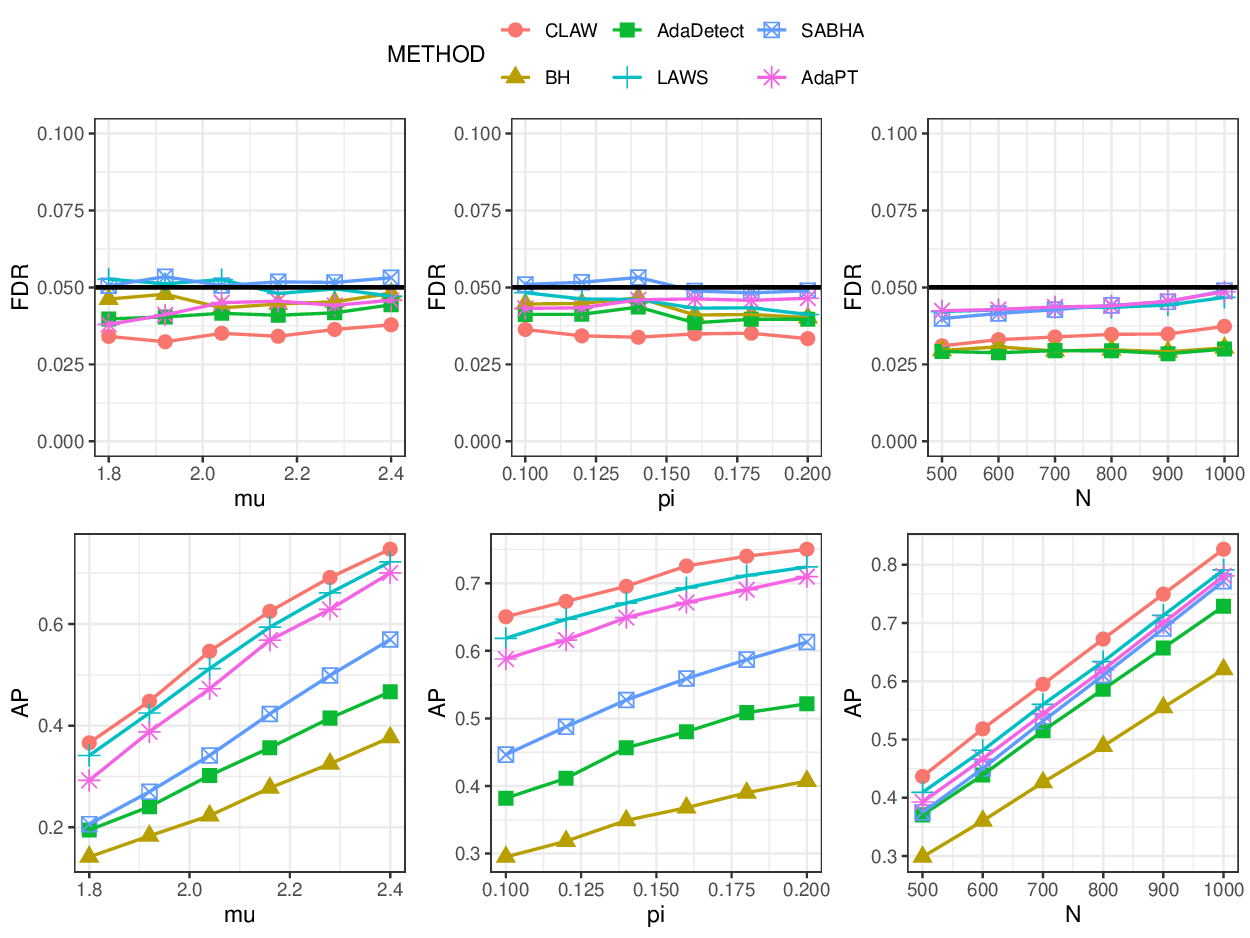}
    \caption{FDR and AP comparison for large-scale two-sample comparisons. 
    The left, middle and right columns are corresponding to settings 1, 2 and 3, respectively.}\label{fig:twosample}
\end{figure}

\subsection{Supplementary Tables and Figures in Real Data Applications}
\label{app:yeast}

This section presents supplementary results pertaining to the real data examples discussed in Section \ref{sec:application}. Specifically, Table \ref{table1} summarizes the number of discoveries made by various methods applied to the MNIST dataset (Section 6.1), while Figure \ref{fig:yeast} illustrates the number of discoveries resulting from different testing procedures applied to the yeast proteins dataset (Section 6.2).

\subsubsection{Supplementary information for MNIST data analysis}

\begin{table}[H]
\caption{The number of discoveries (true discoveries) by different methods applied to two experimental settings on the MNIST dataset. The nominal FDR level is $\alpha=0.05$.}
\centering
\begin{tabular}{lcccccc}
\hline
               & \multicolumn{3}{c}{Setting 1}    & \multicolumn{3}{c}{Setting 2}     \\ \hline
GROUP          & 1         & 2         & ALL       & 1         & 2         & ALL       \\ \hline
PooledAD(KD)   & 0         & 0         & 0         & 0         & 0         & 0         \\
SeparateAD(KD) & 0         & 0         & 0         & 0         & 0         & 0         \\
CLAW(KD)       & 0         & 0         & 0         & 0         & 0         & 0         \\
PooledAD(RF)   & 103 (96)  & 465 (455) & 568 (551) & 108 (104) & 312 (308) & 420 (412) \\
SeparateAD(RF) & 82 (79)   & 462 (453) & 544 (532) & 108 (105) & 360 (356) & 468 (461) \\
CLAW(RF)       & 107 (100) & 476 (465) & 583 (565) & 114 (109) & 368 (363) & 482 (472) \\ \hline
\end{tabular}
\label{table1}
\end{table}

\subsubsection{Supplementary information for protein data analysis}

\begin{figure}[!htbp]
    \centering
    \includegraphics[width=0.9\linewidth]{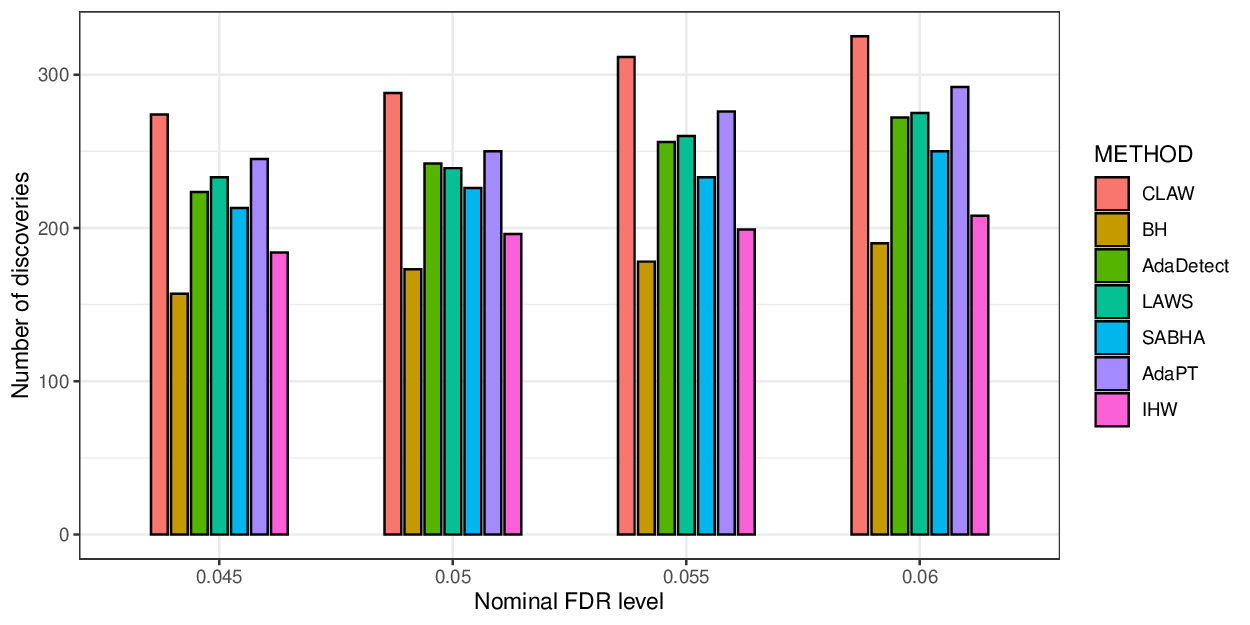}
    \caption{The number of discoveries by different testing procedures at FDR $\alpha=0.045,0.05,0.055,0.06$ for the yeast proteins data.}\label{fig:yeast}
\end{figure}

\end{document}